\documentclass[sn-mathphys,Numbered]{sn-jnl}% Math and Physical Sciences Reference Style
\usepackage{graphicx}
\usepackage{subfigure}
\usepackage{multirow}%
\usepackage{amsmath,amssymb,amsfonts}%
\usepackage{amsthm}%
\usepackage{mathrsfs}%
\usepackage[title]{appendix}%
\usepackage{xcolor}%
\usepackage{textcomp}%
\usepackage{manyfoot}%
\usepackage{booktabs}%
\usepackage{algpseudocode}%
\usepackage{listings}%
\usepackage{matlab-prettifier}
\usepackage{changes}
\RequirePackage{algorithm,algpseudocode}
\usepackage{rotating}
\usepackage{tikz}
\usepackage{tabularx}
\usepackage{adjustbox}
\usetikzlibrary{arrows, decorations.pathmorphing, backgrounds, positioning, fit, petri, automata}
\usetikzlibrary{shadows,arrows,positioning}

\usepackage{indentfirst}
\newtheorem{thm}{Theorem}
\newtheorem{defin}{Definition}
\newtheorem{lem}{Lemma}
\newtheorem{assum}{Assumption}
\newtheorem{rem}{Remark}

\begin{document}

\title[Goodness-of-fit test for multi-layer stochastic block models]{Goodness-of-fit test for multi-layer stochastic block models}

\author*[1]{\fnm{Huan} \sur{Qing}}\email{qinghuan@u.nus.edu}
\affil[1]{\orgdiv{School of Economics and Finance}, \orgname{Chongqing University of Technology}, \city{Chongqing}, \postcode{400054}, \state{Chongqing}, \country{China}}


\abstract{Community detection in multi-layer networks is a fundamental task in complex network analysis across various areas like social, biological, and computer sciences. However, most existing algorithms assume that the number of communities is known in advance, which is usually impractical for real-world multi-layer networks. To address this limitation, we develop a novel goodness-of-fit test for the popular multi-layer stochastic block model based on a normalized aggregation of layer-wise adjacency matrices. Under the null hypothesis that a candidate community count is correct, we establish the asymptotic normality of the test statistic using recent advances in random matrix theory; conversely, we prove its divergence when the model is underfitted. This dual theoretical foundations enable two computationally efficient sequential testing algorithms to consistently determine the number of communities without prior knowledge. Numerical experiments on simulated and real-world multi-layer networks demonstrate the accuracy and efficiency of our approaches in estimating the number of communities. }

\maketitle
\section{Introduction}\label{sec1}
In recent years, multi-layer networks have emerged as a fundamental framework for modeling complex systems where interactions occur across multiple contexts or time points \citep{mucha2010community,de2013mathematical,kivela2014multilayer,boccaletti2014structure}. Such networks capture the richness of relational data by encoding heterogeneous connectivity patterns while preserving shared structural properties. For instance, in social media platforms, user interactions can be represented across diverse platforms (e.g., Facebook, Twitter, LinkedIn, and WeChat) to reveal cross-platform behavioral patterns. In biological systems, gene co-expression networks at different stages of development can be modeled as layers to understand developmental brain disorders \citep{narayanan2010simultaneous,liu2018global,lei2020consistent,lei2023bias}. Similarly, in international trade, relationships between countries can be layered by food product types to analyze trade dynamics \citep{de2015structural}. The ability of multi-layer networks to integrate such heterogeneous yet interconnected information makes them important in domains like social science, neuroscience, systems biology, and economics.  

The study of community structures—groupings of nodes that exhibit cohesive intra-group interactions across layers—is critical for understanding the functional and structural organization of these complex network systems \citep{fortunato2010community,papadopoulos2012community,fortunato2016community,bedi2016community,javed2018community,jin2021survey}. For example, identifying communities in social networks can help design targeted marketing strategies, while detecting functional modules in biological networks can reveal insights into disease mechanisms. For community detection in multi-layer networks, the multi-layer stochastic block model (multi-layer SBM) serves as a popular statistical framework for this purpose. The multi-layer SBM model extends the classical stochastic block model (SBM) \citep{holland1983stochastic} by allowing layer-specific connectivity probabilities while maintaining a consistent community structure across layers. This flexibility enables the model to capture both the shared affiliations of nodes and the variations in interaction patterns across different layers.  Several works have proposed community detection algorithms under the multi-layer SBM model for multi-layer networks, including spectral methods based on the sum of adjacency matrices and matrix factorization methods \citep{han2015consistent,paul2020spectral}, least squares estimation \citep{lei2020consistent}, pseudo-likelihood based algorithm \citep{paul2021null,wang2021fast,fu2023profile}, tensor-based method \citep{xu2023covariate}, and bias-adjusted spectral clustering techniques \citep{lei2023bias,qing2025communityEAAI}. Despite their success, these approaches typically assume that the number of communities $ K$ is known in advance, a restrictive condition that limits their applicability to real-world networks where $ K$ is often unknown in real-world applications, and this restrictive assumption limits the practical utility of these approaches. Estimating $ K$ remains a critical open problem in multi-layer network analysis, as it directly impacts downstream tasks such as model selection and community detection.  

In single-layer networks modeled by the SBM model, the problem of determining the number of communities has been addressed through various statistical methods. To estimate $K$, notable approaches include network cross-validation \citep{chen2018network,li2020network}, Bayesian inference or composite likelihood Bayesian information criterion \citep{mcdaid2013improved,saldana2017many,hu2020corrected}, likelihood ratio tests \citep{LikeAos2017,ma2021determining}, spectrum of the Bethe Hessian matrices \citep{CanMK2022,hwang2024estimation}, 
and goodness-of-fit tests \citep{bickel2016hypothesis,lei2016goodness,dong2020spectral,hu2021using,jin2023optimal,wu2024spectral}. For brief reviews, see \citep{jin2023optimal}. However, multi-layer networks introduce unique challenges due to their layered complexity and intra-layer heterogeneity. Consequently, single-layer estimation techniques can not be immediately applied to estimate the number of communities for multi-layer networks. Therefore, specialized methodologies are required to handle the aggregation of information across multiple layers while preserving statistical tractability. To address this, we propose a novel, principled framework for estimating the number of communities \(K\) in multi-layer networks under the multi-layer SBM model. Our contributions are threefold:
\begin{itemize}
\item \textbf{A novel test statistic and its theoretical foundation}: We introduce a new goodness-of-fit test statistic based on a normalized aggregation of the layer-wise adjacency matrices. This construction is key to handling heterogeneous edge probabilities across layers. For this statistic, we establish its asymptotic normality under the null hypothesis \(H_0: K = K_0\) and prove its divergence under the alternative hypothesis \(H_1: K > K_0\). These theoretical results are derived using recent advances in random matrix theory, particularly the analysis of linear spectral statistics for generalized Wigner matrices. Notably, our theoretical analysis does not depend on the specific community detection algorithm used to estimate the community assignments. Instead, it only requires that the estimated communities satisfy certain misclustering error bounds—a condition that is achievable by many existing methods (e.g., bias-adjusted spectral clustering developed in \citep{lei2023bias}). This generality ensures the applicability of our framework across diverse algorithmic choices. These properties provide the essential theoretical backbone for model specification testing.
\item \textbf{Efficient sequential testing algorithms}: Building directly upon the above theoretical results, we design two computationally efficient sequential testing algorithms: the normalized aggregation spectral test (NAST) and the sequential ratio-based NAST (SR-NAST). We prove that both algorithms consistently estimate the true \(K\). To our knowledge, this is the first work to establish rigorous theoretical guarantees for estimating the number of communities under the multi-layer SBM. 
  \item \textbf{Empirical validation}: Extensive experiments on synthetic and real-world multi-layer networks demonstrate the high accuracy of our approaches in recovering the number of communities. 
\end{itemize}

The rest of this paper is organized as follows. In Section \ref{sec:model}, we formalize the multi-layer stochastic block model. Section \ref{sec:ideal} introduces the ideal test statistic and establishes its asymptotic properties. Section \ref{sec:proposed} presents the feasible statistic and its theoretical foundation. The two consistent sequential testing algorithms, NAST and SR-NAST, are detailed in Section \ref{sec:algorithm}. Their empirical performance is validated in Sections \ref{sec:experiments} and \ref{sec:realdata} using synthetic and real-world data, respectively. We conclude with discussions in Section \ref{sec:conclusion}. Technical proofs are provided in the Appendix.
\section{Multi-layer stochastic block model}\label{sec:model}
To address community detection in multi-layer networks, we adopt the multi-layer stochastic block model (multi-layer SBM) as our foundational framework. This model extends the classical SBM by incorporating layer-specific connectivity patterns while preserving a shared community structure across layers. Below we formalize this model, which serves as the basis for our goodness-of-fit testing procedure.
\begin{defin}[Multi-layer stochastic block model (multi-layer SBM)]\label{def:msbm}
Consider an undirected multi-layer network with \(n\) nodes and \(L\) layers. Let \(\theta \in \{1, \ldots, K\}^n\) be the community assignment vector, where \(K\) is the true number of communities. For each layer \(\ell = 1, \ldots, L\), the adjacency matrix \(A_\ell\) is generated independently as follows:
\begin{itemize}
    \item For \(1 \leq i < j \leq n\), \(A_{\ell,ij} \sim \text{Bernoulli}(B_{\ell,\theta_i\theta_j})\) independently.
    \item Diagonal elements \(A_{\ell,ii} = 0\) for all \(i\).
\end{itemize}
Here, \(B_\ell \in [0,1]^{K \times K}\) is a symmetric connectivity matrix for layer \(\ell\). The community assignment vector \(\theta\) is fixed across layers while \(B_\ell\) may vary, and the layers are independent given \(\theta\).
\end{defin}

The multi-layer SBM model captures two key aspects of real-world multi-layer networks: 1) Consistent community memberships $\theta$ across layers reflect nodes' inherent affiliations, and 2) Layer-specific connectivity matrices encode varying interaction patterns (e.g., social vs. professional contexts). This flexibility enables analysis of multi-layer network data while maintaining a unified community structure.

Despite its flexibility, a fundamental limitation hinders its practical application: the true number of communities \(K\) is typically unknown in real-world multi-layer networks. Most existing methodologies presuppose knowledge of \(K\) \citep{han2015consistent,paul2020spectral,lei2020consistent,lei2023bias,qing2025communityEAAI}, creating a significant gap between theoretical models and real-world applications. This necessitates robust statistical procedures to determine \(K\) before downstream analysis.

To address this limitation, we formulate a sequential hypothesis testing framework. We wish to test \(H_0: K = K_0\) against \(H_1: K > K_0\). The test proceeds sequentially for \(K_0 = 1, 2, \ldots\) until \(H_0\) is accepted. The core challenge lies in developing a computationally feasible test statistic that can reliably distinguish adequate community structure (\(H_0\)) from underfitting (\(H_1\)) without prior knowledge of \(K\). As will be developed in subsequent sections, our solution leverages a normalized aggregation of layer-wise adjacency matrices and uses recent advances in random matrix theory to establish asymptotic normality under the null hypothesis.
\section{Ideal test statistic and its asymptotic normality}\label{sec:ideal}
The foundation of our goodness-of-fit test relies on constructing a suitable aggregate representation of the multi-layer network that facilitates asymptotic analysis. In complex networks with multiple layers, a critical challenge lies in combining edge information across layers while preserving statistical tractability under the null hypothesis. To address this, we define the \textit{ideal normalized aggregation matrix} using the true but unknown parameters:
\begin{equation} \label{eq:A_ideal}  
\widetilde{A}^{\text{ideal}}_{ij} = 
\begin{cases} 
  \dfrac{\sum_{\ell=1}^L (A_{\ell,ij} - P_{\ell,ij})}{\sqrt{n \sum_{\ell=1}^L P_{\ell,ij}(1-P_{\ell,ij})}} & i \neq j, \\
  0 & i = j,
\end{cases}
\end{equation}
where \(P_{\ell,ij} = B_{\ell,\theta_i,\theta_j}\). This construction serves two fundamental purposes: first, it centers each edge by subtracting its true expectation \(P_{\ell,ij}\), ensuring the resulting matrix has mean zero under the model; second, it rescales the sum by the standard deviation of the aggregated edges across layers. The normalization by \(\sqrt{n \sum_{\ell} P_{\ell,ij}(1-P_{\ell,ij})}\) is particularly crucial as it stabilizes the variance across node pairs and scales the entries appropriately for high-dimensional asymptotics. This matrix can be interpreted as a multi-layer generalization of the centered and scaled adjacency matrices used in single-layer goodness-of-fit tests \citep{lei2016goodness,wu2024spectral}, extended to handle heterogeneous edge variances across layers.  

The statistical properties of \(\widetilde{A}^{\text{ideal}}\) make it suitable for random matrix theory analysis. As emphasized in Lemma \ref{lem:ideal_props}, this matrix exhibits three essential characteristics that align with generalized Wigner matrices: zero mean, controlled variance, and conditional independence. These properties emerge directly from the Bernoulli structure of the multi-layer SBM and the layer-wise independence given community assignments. The variance stabilization to \(1/n\) for off-diagonal elements is especially noteworthy because it ensures that the spectral properties of \(\widetilde{A}^{\text{ideal}}\) converge to the semicircle law, mirroring behavior seen in classical random matrix ensembles \citep{bai2010spectral,wang2021generalization}. 

\begin{lem}\label{lem:ideal_props}
\(\widetilde{A}^{\text{ideal}}\) satisfies:
\begin{enumerate}
    \item \(\mathbb{E}[\widetilde{A}^{\text{ideal}}_{ij}] = 0\) for all \(i,j\).
    \item \(\operatorname{Var}(\widetilde{A}^{\text{ideal}}_{ij}) = \frac{1}{n}\) for \(i \neq j\).
    \item Given \(\theta\), \(\{\widetilde{A}^{\text{ideal}}_{ij}\}_{i<j}\) are independent.
\end{enumerate}
\end{lem}

Building on the ideal normalized aggregation matrix, we define our ideal test statistic as the scaled trace of the matrix cubed: 
\[
T^{\text{ideal}} = \frac{1}{\sqrt{6}} \operatorname{tr}\left( \left( \widetilde{A}^{\text{ideal}} \right)^3 \right),
\] 
where $\operatorname{tr}(\cdot)$ denotes the trace operator. The asymptotic normality of \(T^{\text{ideal}}\) relies on controlling edge probabilities to avoid degeneracies. We formalize this through the following assumption.
\begin{assum}\label{assump:a1}
There exists \(\delta\in(0,\frac{1}{2})\) such that \(B_{\ell,kl} \in [\delta, 1-\delta]\) for all \(\ell, k, l\).
\end{assum}
Assumption \ref{assump:a1} ensures all edge probabilities are bounded away from 0 and 1, preventing cases where the variance term \(\sum_{\ell} P_{\ell,ij}(1-P_{\ell,ij})\) vanishes or blows up. This is essential for the variance stabilization in Equation \eqref{eq:A_ideal} to hold uniformly across node pairs. While the sparsity parameter \(\rho = \max_{\ell,k,l} B_{\ell,kl}\) may vary with \(n\), Assumption \ref{assump:a1} requires \(\rho \in (\delta,1-\delta)\), ensuring the network remains neither too sparse nor too dense, which is a common requirement for goodness-of-fit test for SBM in \citep{lei2016goodness,lei2015consistency}. Within this regime, Lemma \ref{lem:ideal_normal} establishes the standard normal limit for \(T^{\text{ideal}}\). The proof leverages recent advances in the central limit theorem for linear spectral statistics of inhomogeneous Wigner matrices \citep{wang2021generalization}, where the variance formula accounts for heterogeneous fourth moments across blocks. Crucially, the trace's cubic form and the variance scaling \(1/\sqrt{6}\) emerge from combinatorial calculations involving non-backtracking walks on three distinct nodes, as detailed in the proof of this lemma.

\begin{lem}\label{lem:ideal_normal}
Suppose that Assumption \ref{assump:a1} holds, then under the null hypothesis \(H_0: K = K_0\), we have
\[
T^{\text{ideal}} \xrightarrow{d} N(0,1),
\]
where \(\xrightarrow{d}\) means convergence in distribution and $N(0,1)$ denotes the standard normal distribution.
\end{lem}

\section{Proposed test statistic and its theoretical properties}\label{sec:proposed}
While the ideal test statistic \(T^{\text{ideal}}\) developed in Section \ref{sec:ideal} provides a theoretically sound foundation for testing \(H_0: K = K_0\) under the multi-layer SBM, it relies critically on the true community assignment vector \(\theta\) and the true layer-specific connectivity matrices \(\{B_{\ell}\}_{\ell=1}^{L}\). However, in practical applications involving real-world multi-layer networks, these parameters are unknown, creating a fundamental gap between the theoretical ideal and practical implementation. Consequently, \(T^{\text{ideal}}\) cannot be computed directly from observed data. To bridge this gap, in this section, we develop a feasible test statistic \(T\) by replacing the unknown true parameters \(\theta\) and \(\{B_{\ell}\}\) with suitable estimates \(\hat{\theta}\) and \(\{\hat{B}_{\ell}\}\) derived from the data. The primary goal of this section is to construct this practical test statistic and establish its asymptotic normality under \(H_0\).

Suppose that $\hat{\theta}\in\{1,2,\ldots, K_{0}\}^{n}$ is the estimated community label vector returned by any community detection algorithm $\mathcal{M}$ with target number of communities $K_0$ for the multi-layer network. Let $\hat{C}_k=\{i:\hat{\theta}_i=k\mathrm{~for~}i\in[n]\}$ be the set of nodes and $\hat{n}_k:=|\hat{C}_k|$ be the number of nodes 
belonging to the $k$-th estimated community for $k\in[K_0]$. Similar to \citep{lei2016goodness,wu2024spectral}, we estimate the $L$ connectivity matrices $\{B_\ell\}^{L}_{\ell=1}$ as follows:
\begin{equation} \label{eq:B_hat}  
\hat{B}_{\ell,kl} = 
\begin{cases} 
  \dfrac{1}{\hat{n}_k \hat{n}_l} \sum_{i \in \hat{C}_k} \sum_{j \in \hat{C}_l} A_{\ell,ij}, & k \neq l, \\
  \dfrac{1}{\hat{n}_k (\hat{n}_k - 1)/2} \sum_{i < j \in \hat{C}_k} A_{\ell,ij}, & k = l,\ \hat{n}_k \geq 2, \\
  0, & \text{otherwise},
\end{cases}
\end{equation}
where $k\in[K_0], l\in[K_0], \ell\in[L]$. We then estimate the $L$ probability matrices $\{P_\ell\}^{L}_{\ell=1}$ as follows:
\begin{equation} \label{eq:P_hat}  
\hat{P}_{\ell,ij} = \hat{B}_{\ell,\hat{\theta}_i\hat{\theta}_j},
\end{equation}
where $i\in[n], j\in[n], \ell\in[L]$. Note that if we let $\Theta$ be a $n\times K_0$ matrix such that $\Theta_{ik}=1$ if $\hat{\theta}_i=k$ and 0 otherwise for $i\in[n], k\in[K_0]$, we have \(\hat{P}_\ell = \Theta \hat{B}_\ell \Theta^\top\) for $\ell\in[L]$. Algorithm \ref{alg:estimation} summarizes the details of parameter estimation for multi-layer networks.
\begin{algorithm}[H]
\caption{Parameter estimation for multi-layer SBM}\label{alg:estimation}
\begin{algorithmic}[1]
\Require Adjacency matrices \(\{A_\ell\}_{\ell=1}^L\) and number of communities \(K_0\)
\Ensure Estimated probability matrices \(\{\hat{P}_{\ell}\}^{L}_{\ell=1}\)
\State Run any community detection algorithm $\mathcal{M}$ to \(\{A_\ell\}_{\ell=1}^L\) with $K_0$ communities to get $\hat{\theta}$.
\State Estimate connectivity matrices via Equation (\ref{eq:B_hat}).
\State Estimate probability matrices via Equation (\ref{eq:P_hat}).
\end{algorithmic}
\end{algorithm}

After obtaining the estimated probability matrices, we can construct the normalized aggregation matrix as follows:
\begin{equation} \label{eq:A_hat}  
\widetilde{A}^{\text{agg}}_{ij} = 
\begin{cases} 
  \dfrac{\sum_{\ell=1}^L (A_{\ell,ij} - \hat{P}_{\ell,ij})}{\sqrt{n \sum_{\ell=1}^L \hat{P}_{\ell,ij}(1-\hat{P}_{\ell,ij})}} & i \neq j, \\
  0 & i = j,
\end{cases}
\end{equation}
where $i\in[n], j\in[n]$. Based on the normalized aggregation matrix $\widetilde{A}^{\text{agg}}$, our test statistic is designed as follows:
\begin{equation} \label{eq:T}  
T = \frac{1}{\sqrt{6}} \operatorname{tr}\left( \left( \widetilde{A}^{\text{agg}} \right)^3 \right)
\end{equation}

To develop the asymptotic normality of the proposed test statistics $T$, we need the following assumptions.
\begin{assum}\label{assump:a2}
There exists \(c > 0\) such that \(\min_{1 \leq k \leq K} |C_k| \geq c n / K\).
\end{assum}
Assumption \ref{assump:a2} means that the community sizes are balanced, where this assumption is also needed in the goodness-of-fit test for SBM in \citep{lei2016goodness,wu2024spectral}.

Let $m:=\|\hat{\theta} - \theta\|_0$ denotes the number of misclustered nodes for any community detection algorithm $\mathcal{M}$, where similar to the analysis in \citep{jin2024mixed}, we assume that the community labels are aligned via a permutation that minimizes the number of misclustered nodes throughout this paper. For our theoretical analysis, we also need the following assumption to control the growth rates of the number of layers $L$ and number of communities $K_0$ relative to the number of nodes $n$. And we also need it to control the number of misclustered nodes for any community detection algorithm $\mathcal{M}$.
\begin{assum}\label{assump:a3}
\(\frac{LK^2\mathrm{max}(\log n, m^2)}{n}\to 0\) as \(n \to \infty\).
\end{assum}

Assumption \ref{assump:a3} serves as a critical regularity condition for establishing the asymptotic normality of the test statistic \(T\) under \(H_0: K = K_0\). This assumption governs the interplay between network size (\(n\)), number of communities (\(K\)), layers (\(L\)), and misclustering error (\(m = \|\hat{\theta} - \theta\|_0\)) for any community detection algorithm $\mathcal{M}$. Its role is to control the accumulation of estimation errors in \(\widetilde{A}^{\text{agg}}\) relative to the idealized \(\widetilde{A}^{\text{ideal}}\), thereby preserving the weak convergence \(T\overset{d}{\to} N(0,1)\). We discuss its implications across several key asymptotic regimes below:
\begin{itemize}
  \item The mild misclustering (\(m^2=O(\log n)\)) case: When the misclustering error is moderate such that \(m^2 \leq C \log n\) for some \(C > 0\), Assumption \ref{assump:a3} simplifies to \(\frac{LK^{2} \log n}{n} \to 0\). Here, the dominant constraint is the interplay between \(L\), \(K\), and \(n\). If \(K= O(1)\) (fixed number of communities), \(L\) can grow as \(o(n / \log n)\). This accommodates moderately large multi-layer networks where the number of layers scales sub-linearly with \(n\). If \(K\) grows with \(n\) (e.g., \(K = n^\alpha\)), \(L\) must satisfy \(L = o(n^{1-2\alpha} / \log n)\), implying \(\alpha < 1/2\) is necessary for \(L \geq 1\). 
  \item The severe misclustering (\(m^2 \gg \log n)\) case: When misclustering is substantial (e.g., \(m \propto n^\gamma\) for \(\gamma > 0\)), Assumption 3 reduces to \(\frac{LK^{2} m^{2}}{n} \to 0\). If \(K = O(1)\), this requires \(L m^2 / n \to 0\). For \(L = O(1)\), it demands \(m = o(n^{1/2})\), meaning the community detection algorithm \(\mathcal{M}\) must achieve clustering consistency  with a rate slower than \(\sqrt{n}\). If \(m \propto n^{\gamma}\) (\(\gamma \geq 1/2\)), the assumption fails unless \(L \to 0\), which is impractical. 
  \item The fixed dimensions (\(K = O(1)\) and \(L = O(1)\)) case: When both candidate communities and layers are fixed, Assumption \ref{assump:a3} reduces to \(\max(\log n, m^2)/n \to 0\), which holds if \(m^2 = o(n)\). This is equivalent to requiring \(m = o(\sqrt{n})\) if \(m^2 > \log n\), or trivially satisfied if \(m^2 = O(\log n)\). Here, the critical constraint is that the misclustering rate must satisfy \(m / \sqrt{n} \overset{P}{\to} 0\). This regime is feasible with spectral algorithms (e.g., bias-adjusted SoS introduced in \citep{lei2023bias}) under Assumptions \ref{assump:a1}-\ref{assump:a2}, as their typical error rates (e.g., \(m = O_P(1)\) or \(m = O_P(\log n)\)) readily satisfy this condition. 
   \item The high-dimensional regimes (\(K \propto n^\alpha\), \(L \propto n^\beta)\) case: When both \(K\) and \(L\) scale with \(n\), Assumption \ref{assump:a3} becomes \(n^{\beta + 2\alpha} \max(\log n, m^2) / n \to 0\). If \(m^2 = O(\log n)\), this simplifies to \(\beta + 2\alpha < 1\). For example, if \(\alpha = 1/4\), \(\beta < 1/2\) is required. If \(m^2\) grows polynomially (e.g., \(m \propto n^\gamma\)), the assumption necessitates \(\beta + 2\alpha + 2\gamma < 1\). This imposes strict limitations: even moderate growth in \(K\) (e.g., \(\alpha > 0\)) forces \(\beta\) or \(\gamma\) to be negative unless \(m\) decays with \(n\). In practice, this implies the test is only feasible for very large \(n\) when \(K\) and \(L\) are small relative to \(n\), or when community detection is exceptionally accurate (\(\gamma \ll 1/2\)).
\end{itemize}

The above analysis shows that the asymptotic normality of $T$ typically requires $K = o(\sqrt{n})$ and $m = o(\sqrt{n})$ for any community detection algorithm $\mathcal{M}$.
\begin{rem}
In this paper, if we use the bias-adjusted algorithms introduced in \citep{lei2023bias,qing2025community,qing2025communityEAAI} to estimate communities for multi-layer networks, where these methods are spectral algorithms with theoretical guarantees under multi-layer SBM. Define $B_{\ell,0}=B_{\ell}/\rho$ for $\ell\in[L]$. If Assumptions \ref{assump:a1}-\ref{assump:a2} hold and we further assume that $\sum_{\ell}B^{2}_{\ell,0}$ satisfies $\lambda_{\mathrm{min}}(\sum_{\ell}B^{2}_{\ell,0})\geq cL$ for some $c>0$, where $\lambda_{\mathrm{min}}(\cdot)$ denotes the smallest eigenvalue (in magnitude) of a square matrix, then main theorems \citep{lei2023bias,qing2025community,qing2025communityEAAI} guarantees that under \(H_0: K = K_0\), the estimated community label vecto \(\hat{\theta}\) obtained by the bias-adjusted spectral clustering algorithms satisfies
\[
\frac{1}{n} \|\hat{\theta} - \theta\|_0 = O_P\left( \frac{1}{n^2} + \frac{\log(L+n)}{L n^2 \rho^2} \right)=O_P\left( \frac{1}{n^2} + \frac{\log(L+n)}{L n^2 \delta^2} \right).
\]
This means that when we let the community detection algorithm $\mathcal{M}$ be the bias-adjusted spectral algorithms, we have $m=O_P\left( \frac{1}{n} + \frac{\log(L+n)}{L n\delta^2} \right)$ (a value much smaller than $\log n$) which simplifies Assumption \ref{assump:a3} to \(\frac{L\log n}{n}\to 0\) as \(n \to \infty\) when $K=O(1)$, implying that $L$ should grow slower than $n/\log n$ as $n$ grows. Community detection approaches developed in \citep{han2015consistent,lei2020consistent,paul2020spectral} are specifically designed for multi-layer SBM  and their number of misclustered nodes also much smaller than $\sqrt{n}$ under mild conditions, which imply that many community detection algorithms for multi-layer SBM satisfy $m=o(\sqrt{n})$ and can be used for our parameter estimation. In this paper, we choose the bias-adjusted SoS algorithm developed in \citep{lei2023bias} for the parameter estimation in Algorithm \ref{alg:estimation} as numerical results in \citep{lei2023bias} show that bias-adjusted SoS generally outperforms methods developed in \citep{han2015consistent,paul2020spectral} in detecting communities, and it is computationally fast. 
\end{rem}
The following lemma guarantees that the proposed test statistic is close to the ideal test statistic.
\begin{lem}\label{lem:plugin_error}
Suppose that Assumptions \ref{assump:a1}-\ref{assump:a3} hold, then under the null hypothesis \(H_0: K = K_0\), we have
\[
T - T^{\text{ideal}} = o_P(1).
\]
\end{lem}

The following theorem is the main theoretical result of this paper, as it guarantees the asymptotic normality of the proposed test statistic $T$ under mild conditions.
\begin{thm}[Asymptotic null distribution]\label{thm:main}
Suppose that Assumptions \ref{assump:a1}-\ref{assump:a3} hold, then under the null hypothesis \(H_0: K = K_0\), we have
\[
T \xrightarrow{d} N(0,1).
\]
\end{thm}

The accuracy of the community detection algorithm fundamentally influences the validity of Theorem \ref{thm:main}, as its conclusion relies critically on Assumption \ref{assump:a3} controlling the misclustering error \(m = \|\hat{\theta} - \theta\|_0\). Specifically, the asymptotic normality \(T\xrightarrow{d} N(0,1)\) under \(H_0\) requires \(m = o(\sqrt{n})\) for the estimation errors in \(\widetilde{A}^{\text{agg}}\) to vanish relative to \(\widetilde{A}^{\text{ideal}}\). When community detection is highly accurate (e.g., \(m <C\log n)\) or \(m = \mathcal{O}_P(1)\)), Assumption \ref{assump:a3} simplifies significantly—often reducing to \(\frac{LK^2 \log n}{n} \to 0\)—which permits a broader range of network dimensions (e.g., larger \(L\) or slowly growing \(K\)). Conversely, less accurate algorithms risk violating \(m = o(\sqrt{n})\), particularly in high-dimensional regimes (\(K \propto n^\alpha, L \propto n^\beta\)), where polynomial growth in \(m\) quickly destabilizes the variance structure of \(\widetilde{A}^{\text{agg}}\) and invalidates the asymptotic normality. Thus, selecting algorithms with theoretical guarantees of tight misclustering bounds (e.g., spectral methods with explicit \(m = o_P(\sqrt{n})\) rates under multi-layer SBM) is essential not only for parameter estimation accuracy but also to ensure the test statistic’s limiting distribution holds across practical and scalability conditions. This is why we always prefer community detection methods with tight misclustering bounds in multi-layer SBM for the sequential testing framework used for the estimation of the number of communities in this paper.

The following conditions are needed for the test's power against the multi-layer SBM model.
\begin{assum}\label{assump:a4}
Define the average connection matrix $\bar{B} = \frac{1}{L} \sum_{\ell=1}^{L} B_{\ell}$. There exists a constant $\eta > 0$ such that
\[
\min_{1 \le k \neq k' \le K} \; \max_{1 \le l \le K} \; \bigl| \bar{B}_{k l} - \bar{B}_{k' l} \bigr| \ge \eta.
\]
\end{assum}
Assumption \ref{assump:a4} ensures that the true communities remain distinguishable in the \emph{average} connectivity matrix $\bar{B}$. This is a mild, natural condition: it only requires that after aggregating across layers, every pair of communities differs in how they connect to at least one other community by a non‑vanishing amount $\eta$. It does not demand separation in every individual layer, making it suitable for real multi‑layer data where layers can be heterogeneous. In fact,
this condition is the multi‑layer analogue of the classical row‑separation condition considered in Theorem 3.3 of \citep{lei2016goodness}.  
\begin{assum}\label{assump:a5}
When $K_0 < K$, we require
\[
\frac{n L}{K^4} \to \infty \mathrm{~and~}\frac{n L^{3}}{K^{12}(\log n)^{3}} \to \infty\quad \text{as } n \to \infty.
\]
\end{assum}
 
Assumption \ref{assump:a5} is mild in typical applications because the number of communities $K$ is usually much smaller than the number of nodes $n$. In the simplest case where $K$ is bounded ($K=O(1)$), both conditions reduce to $nL\to\infty$ and $nL^{3}/(\log n)^{3}\to\infty$, which hold naturally. Even if $K$ grows slowly with $n$, the two conditions remain easily satisfied as long as $K$ does not increase too fast relative to $n$ and $L$. Thus, Assumption \ref{assump:a5} describes a realistic scaling regime that covers most practical multi‑layer network scenarios.

The following theorem provides a lower bound of $|T|$ under the alternative model $K_{0}<K$.
\begin{thm}[Asymptotic power guarantee]\label{thm:power}
Suppose that Assumptions \ref{assump:a1}--\ref{assump:a5} hold, and the true number of communities $K$ is greater than the candidate $K_0$ (i.e., $K_0 < K$). Then, there exists a constant $C > 0$ such that
\[
\lim_{n \to \infty} \mathbb{P}\left( |T(K_0)| \ge C \cdot \frac{(nL)^{3/2}}{K^6} \right) = 1.
\]
\end{thm}

Theorem \ref{thm:power} establishes the asymptotic power of the proposed goodness-of-fit test against the underfitting case $K_{0}<K$. This theorem shows that the test statistic $|T(K_0)|$ grows at least as fast as $(nL)^{3/2}/K^6$ in probability. This growth rate is derived from the spectral norm of a signal submatrix that arises because, when $K_0<K$, at least one estimated community must contain nodes from two distinct true communities. The separation condition in Assumption \ref{assump:a4} ensures that the aggregated connectivity profiles of these true communities differ by a non‑vanishing amount $\eta$, which translates into a deterministic signal of order $\sqrt{nL}/K^2$ in the normalized aggregation matrix $\widetilde{A}^{\text{agg}}$. After controlling the noise part via Lemma \ref{lem:noise_full} given in the appendix, the cubic trace $T(K_0)$ inherits a magnitude of order $(nL)^{3/2}/K^6$. 

Together with Theorem \ref{thm:main}, Theorem \ref{thm:power} ensures a sharp phase transition: $|T(K_0)|$ is $O_P(1)$ when $K_0 = K$ but diverges at a polynomial rate when $K_0 < K$. This behavior is exploited in our sequential algorithms to consistently estimate $K$. Thus, Theorem \ref{thm:power} is not merely a technical bound; it is the cornerstone for proving the consistency of our sequential estimation algorithms designed in the next section.
\section{Hypothesis testing algorithms and their consistency in the estimation of $K$}\label{sec:algorithm}
The asymptotic normality of $T(K_0)$ under $H_0$ (Theorem \ref{thm:main}) and its divergence under $H_1$ (Theorem \ref{thm:power}) provide a rigorous foundation for estimating the number of communities $K$ via sequential testing. In this section, we present two practical algorithms: the Normalized Aggregation Spectral Test (NAST) and the Sequential Ratio-based Normalized Aggregation Spectral Test (SR-NAST). Both methods are computationally efficient, requiring only a single community detection step per candidate $K_0$, and both are provably consistent.
\subsection{Normalized Aggregation Spectral Test (NAST)}
NAST directly implements the sequential goodness-of-fit test. Starting with $K_{0}=1$, our NAST repeatedly invokes a community detection routine to obtain an estimated community label vector, estimates layer-wise connectivity and probability matrices, forms the normalized aggregation matrix, and evaluates the cubic-trace test statistic $T$. The loop terminates when \(|T|\) falls below a prescribed threshold \(t_n\) (e.g., $\log n$), at which point the current $K_{0}$ is returned as the estimated number of communities. In this paper, we use the bias-adjusted SoS algorithm of \citep{lei2023bias} as the community detection technique, because its mis-clustering error satisfies $m=o(\sqrt{n})$ under mild conditions, thereby guaranteeing that Assumption \ref{assump:a3} holds and the asymptotic null distribution $N(0,1)$ is preserved. The details of our NAST method is summarized in Algorithm \ref{alg:test}. 

\begin{algorithm}[H]
\caption{Normalized Aggregation Spectral Test (NAST)}\label{alg:test}
\begin{algorithmic}[1]
\Require Adjacency matrices \(\{A_\ell\}_{\ell=1}^L\), threshold $t_{n}$
\Ensure Estimated number of communities \(\hat{K}\)
\State Initialize \(K_0 \gets 1\)
\Repeat
\State Obtain \(\{\hat{P}_\ell\}^{L}_{\ell}\) via Algorithm \ref{alg:estimation} using candidate community number $K_0$.
\State Compute \(T(K_{0})\) via Equation (\ref{eq:T}).
\If{\(|T(K_{0})| < t_{n}\)}
    \State Accept \(H_0\), set \(\hat{K} = K_0\)
\Else
    \State \(K_0 \gets K_0 + 1\)
\EndIf
\Until{\(H_0\) accepted}
\State \Return \(\hat{K}\)
\end{algorithmic}
\end{algorithm}

The threshold $t_n$ must be chosen to separate the null distribution from the alternative. The following theorem provides explicit conditions on $t_n$ that guarantee consistency.

\begin{thm}[Consistency of NAST]\label{thm:consistency}
Suppose that Assumptions \ref{assump:a1}--\ref{assump:a5} all hold. Define the NAST estimator $\hat{K}$ as in Algorithm \ref{alg:test}, where the threshold $t_n$ satisfies:
\begin{itemize}
    \item $t_n \to \infty$,
    \item $t_n = o\!\left( (nL)^{3/2} / K^6 \right)$.
\end{itemize}
Then,
\[
\lim_{n \to \infty} \mathbb{P}\bigl( \hat{K} = K \bigr) = 1.
\]
\end{thm}

Theorem \ref{thm:consistency} follows directly from Theorems \ref{thm:main} and \ref{thm:power}. The conditions on $t_n$ ensure that $t_n$ grows slower than the divergence rate under $H_1$ but faster than the bounded fluctuations under $H_0$. Thus, as $n \to \infty$, $\mathbb{P}(|T(K_0)| > t_n) \to 1$ for all $K_0 < K$, while $\mathbb{P}(|T(K)| \le t_n) \to 1$. Consequently, NAST rejects all $K_0 < K$ and accepts exactly at $K_0 = K$ with probability approaching one. A simple valid choice is $t_n = \log n$, which satisfies both conditions under Assumption \ref{assump:a5} because $(nL)^{3/2}/K^6$ grows much faster than $\log n$. In practice, one might consider thresholds of the form $c \log n$ for some constant $c > 0$, but the asymptotic theory remains valid for any such threshold sequence satisfying conditions in Theorem \ref{thm:consistency}.

\subsection{Sequential Ratio-Based Normalized Aggregation Spectral Test (SR-NAST)}
SR-NAST replaces the absolute test statistic in NAST with a ratio statistics. This method leverages the distinct asymptotic behavior of this ratio: it remains stochastically bounded when the candidate number of communities is below the true value but diverges to infinity when the candidate number equals the true value. This sharp phase transition allows the algorithm to reliably identify the true number of communities using a simple, slowly diverging threshold (e.g., $\log n$).

This method examines the ratio statistic defined below:
\begin{align}\label{etaK0}
\eta_{K_0} = |T(K_0-1)| / |T(K_0)|\qquad{K_{0}=2,3,\ldots,K}.
\end{align}

The behavior of the ratio statistic $\eta_{K_0}$ is characterized by the following theorem, which justifies the stopping rule. 
\begin{thm}[Behavior of the ratio statistic]\label{thm:ratio-behavior}
Under Assumptions \ref{assump:a1}-\ref{assump:a5}:
\begin{enumerate}
  \item For $2 \le K_0 \le K-1$ (underfitting),
  \[
  \eta_{K_0} = O_P(1).
  \]
  \item When $K_0 = K$ (true model),
  \[
  \eta_K \xrightarrow{P} \infty.
  \]
\end{enumerate}
\end{thm}
Theorem \ref{thm:ratio-behavior} captures the phase transition in $\eta_{K_0}$. For $K_0 \le K-1$, both $|T(K_0-1)|$ and $|T(K_0)|$ diverge at the same rate $(nL)^{3/2}/K^6$ by Theorem \ref{thm:power} and its proof, so their ratio is stochastically bounded. At $K_0 = K$, the numerator $|T(K-1)|$ diverges by Theorem \ref{thm:power} while the denominator $|T(K)|$ is $O_P(1)$ by Theorem \ref{thm:main}, causing $\eta_K$ to diverge. This sharp change makes $\eta_{K_0}$ a natural diagnostic for detecting $K$.

SR-NAST sequentially computes $\eta_{K_0}$ and stops when it exceeds a simple threshold $t_{\mathrm{ratio},n}$ (e.g., $t_{\mathrm{ratio},n}= \log n$). Algorithm \ref{alg:sr-nast} provides the complete procedure.

With the result of Theorem \ref{thm:ratio-behavior}, the consistency of SR-NAST follows immediately.
\begin{thm}[Consistency of SR-NAST]\label{thm:srnast-consistency}
Under Assumptions \ref{assump:a1}-\ref{assump:a5}, define the SR-NAST estimator $\hat{K}$ as in Algorithm \ref{alg:sr-nast}, where the threshold $t_{\mathrm{ratio},n}$ satisfies:
\begin{itemize}
    \item $t_{\mathrm{ratio},n} \to \infty$,
    \item $t_{\mathrm{ratio},n} = o\!\left( (nL)^{3/2} / K^6 \right)$.
\end{itemize}
Then,
\[
\lim_{n \to \infty} \mathbb{P}(\hat{K} = K) = 1.
\]
\end{thm}
Theorem \ref{thm:ratio-behavior} implies that at $K_0 = K$, $\eta_K > t_{\mathrm{ratio},n}$ with probability tending to one; while for $K_0 = 2, \dots, K-1$, $\eta_{K_0}$ is bounded and hence $\eta_{K_0} > t_{\mathrm{ratio},n}$ with probability tending to zero. The initial step ($K_0=1$) is handled by the same threshold: if $K=1$, Theorem \ref{thm:main} gives $|T(1)| = O_P(1)$, so $|T(1)| \le t_{\mathrm{ratio},n}$ with high probability; if $K>1$, Theorem \ref{thm:power} gives $|T(1)|$ divergent, so $|T(1)| > t_{\mathrm{ratio},n}$ with high probability, and the algorithm moves to $K_0=2$. Thus, SR-NAST iterates through $K_0=1,2,\dots$ and stops exactly at $K_0=K$ with probability approaching one. The choice $t_{\mathrm{ratio},n}= \log n$ is convenient and satisfies requirements in Theorem \ref{thm:srnast-consistency}.
\begin{algorithm}[H]
\caption{Sequential Ratio-based Normalized Aggregation Spectral Test (SR-NAST)}\label{alg:sr-nast}
\begin{algorithmic}[1]
\Require Multi-layer adjacency matrices \(\{A_\ell\}_{\ell=1}^L\), threshold $t_{\mathrm{ratio},n}>0$
\Ensure Estimated number of communities \(\hat{K}\)
\State Set \(K_0 \gets 1\)
\State Compute \(T(1)\) via Equation (\ref{eq:T})
\If{\(|T(1)| \leq t_{\mathrm{ratio},n}\)}
    \State \(\hat{K} \gets 1\)
    \State \Return \(\hat{K}\)
\EndIf
\State Set \(K_0 \gets 2\)
\Repeat
    \State Compute \(T(K_0)$ via Equation (\ref{eq:T}) with candidate number \(K_0\)
    \State Compute \(\eta_{K_0}\) via Equation (\ref{etaK0})
    \If{\(\eta_{K_0} > t_{\mathrm{ratio},n}\)}
        \State \(\hat{K} \gets K_0\)
    \Else
        \State $K_0 \gets K_0 + 1$
    \EndIf
\Until{$\eta_{K_0} > t_{\mathrm{ratio},n}$}
\State \Return \(\hat{K}\)
\end{algorithmic}
\end{algorithm}

\begin{rem}
SR-NAST is demonstrably more robust than NAST due to fundamental differences in their stopping rules. While both algorithms employ the same test statistic \(T(K_0)\)—whose asymptotic normality under the correct model (\(H_0: K=K_0\)) and divergence under underfitting (\(H_1: K>K_0\)) form the basis for sequential testing—their decision criteria differ critically. NAST stops when the absolute value \(|T(K_0)|\) falls below a predefined threshold \(t_n\). This requires identifying the point where a diverging sequence transitions to an \(O_P(1)\) variable, a task highly sensitive to the precise asymptotic calibration of \(t_n\). In contrast, SR-NAST leverages the ratio statistic \(\eta_{K_0}\), which is \(O_P(1)\) for all \(K_0 < K\) but diverges in probability (\(\eta_K \stackrel{P}{\rightarrow} \infty\)) precisely at \(K_0 = K\). This creates a sharp, discontinuous signal—a jump from stochastic boundedness to explosion—which is inherently easier and more reliable to detect than the continuous attenuation of \(|T(K_0)|\). Consequently, the termination of SR-NAST is less sensitive to the specific choice of its threshold \(t_{\text{ratio},n}\) compared to the delicate balance required for \(t_n\) in NAST. 
\end{rem}
\section{Numerical experiments}\label{sec:experiments}

This section validates the theoretical properties of the proposed test statistic and evaluates the performance of the proposed algorithms. We design experiments to verify the asymptotic normality of the test statistic $T$ under $H_0$, examine the size and power of $T$, and assess the accuracy of NAST and SR-NAST in estimating the true number of communities $K$ across diverse multi-layer networks. All simulations use the bias-adjusted SoS algorithm \citep{lei2023bias} for community detection. Multi-layer network generation follows the multi-layer SBM framework in Definition \ref{def:msbm}, and the community assignments are generated by letting each node belong to each community with equal probability for all experiments.

\texttt{Experiment 1: Asymptotic normality of $T$ under $H_0$.} We validate Theorem \ref{thm:main} by generating multi-layer networks. Networks comprise $n\in\{200,600,1000\}$ nodes, $L = 5$ layers, and $K\in\{2,3,4\}$ communities. Crucially, layer-specific connectivity matrices $B_\ell$ exhibit heterogeneous patterns: diagonal entries $B_{\ell,kk} \sim \text{Uniform}(0.65, 0.75)$ and off-diagonal entries $B_{\ell,kl} \sim \text{Uniform}(0.25, 0.35)$ for $k\neq \ell$, satisfying Assumption \ref{assump:a1} with $\delta = 0.25$ while ensuring $B_\ell \neq B_{\ell'}$ for $\ell \neq \ell'$. For each of 1,000 Monte Carlo replicates, we compute $T$ using Equation (\ref{eq:T}) and plot its empirical distribution in Figure \ref{fig:normality}. The histograms demonstrate alignment with the theoretical \(N(0,1)\) curve across all configurations. This visual concordance empirically confirms Theorem \ref{thm:main}, indicating that the asymptotic null distribution of \(T\) holds robustly even for finite-sized networks.  
\begin{figure}[htp!]
\centering
\resizebox{\columnwidth}{!}{
\subfigure{\includegraphics[width=0.33\textwidth]{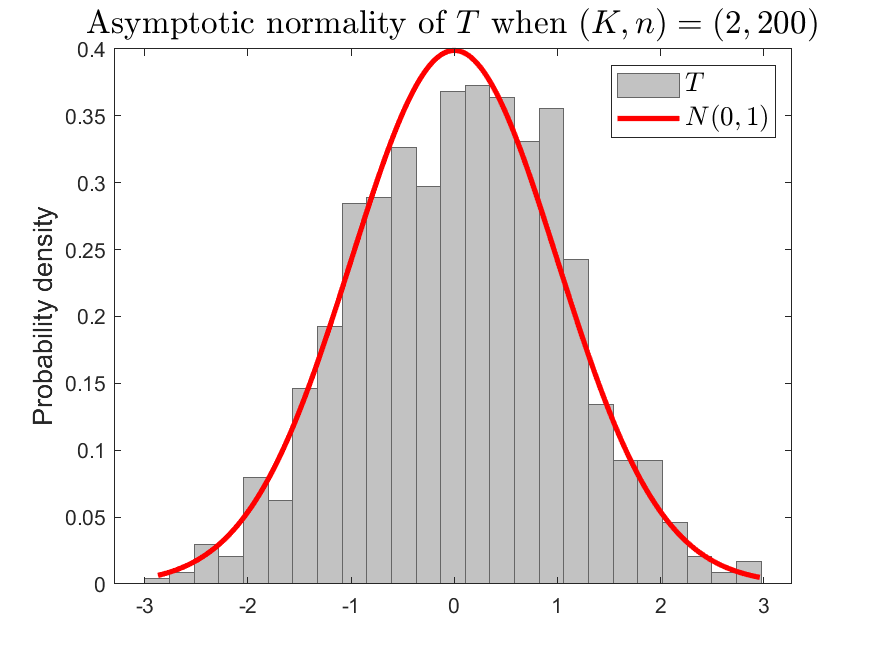}}
\subfigure{\includegraphics[width=0.33\textwidth]{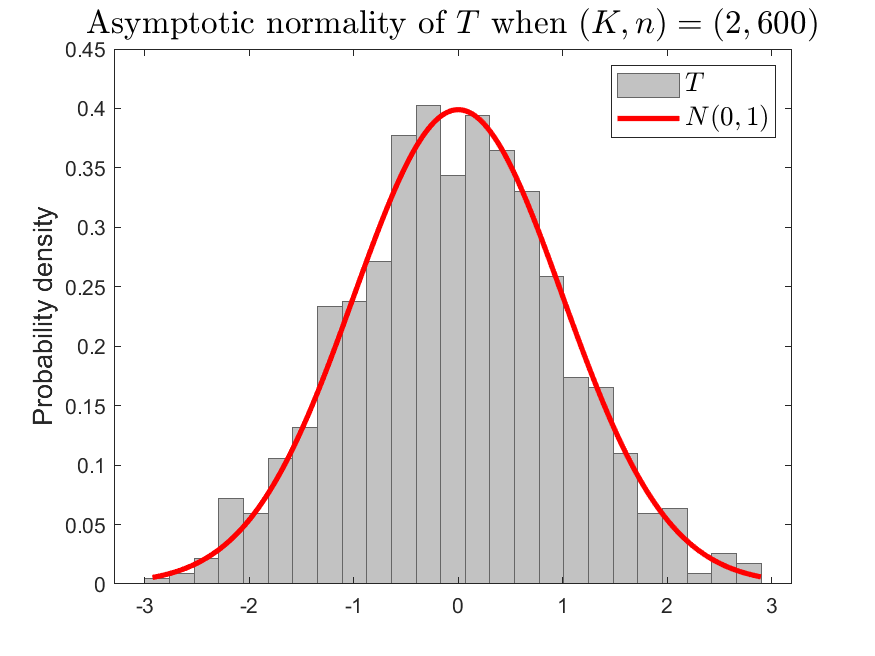}}
\subfigure{\includegraphics[width=0.33\textwidth]{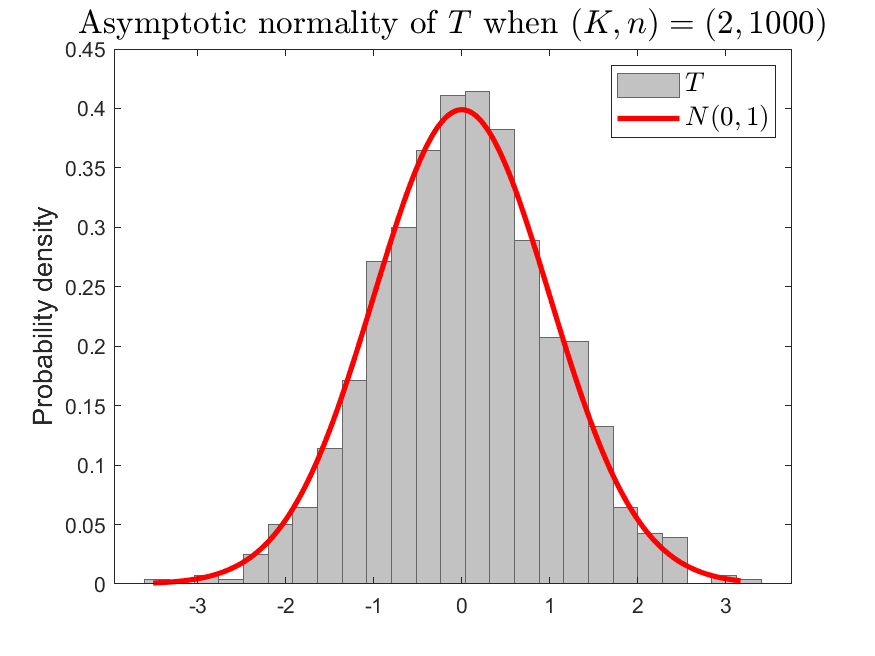}}
}
\resizebox{\columnwidth}{!}{
\subfigure{\includegraphics[width=0.33\textwidth]{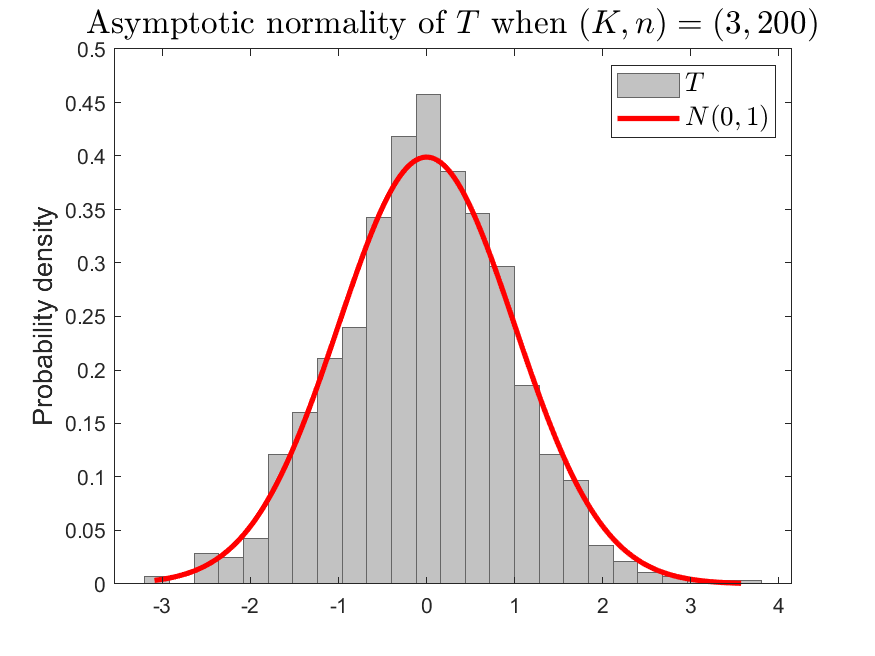}}
\subfigure{\includegraphics[width=0.33\textwidth]{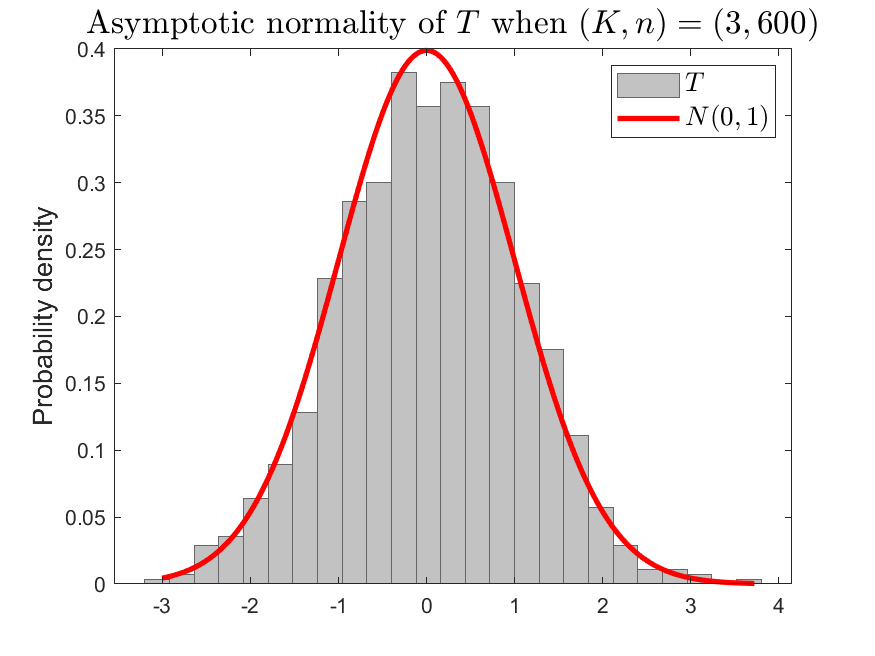}}
\subfigure{\includegraphics[width=0.33\textwidth]{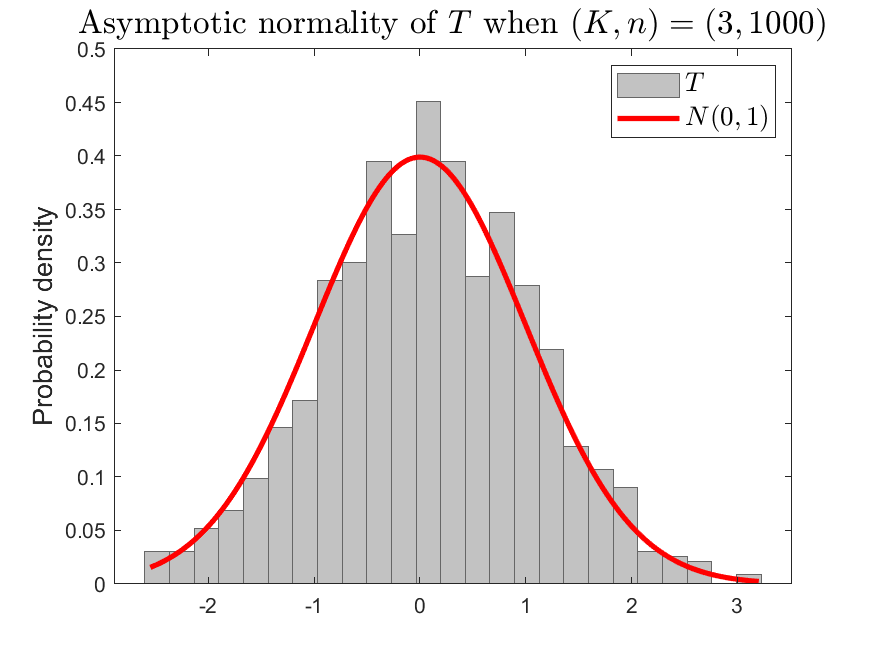}}
}
\resizebox{\columnwidth}{!}{
\subfigure{\includegraphics[width=0.33\textwidth]{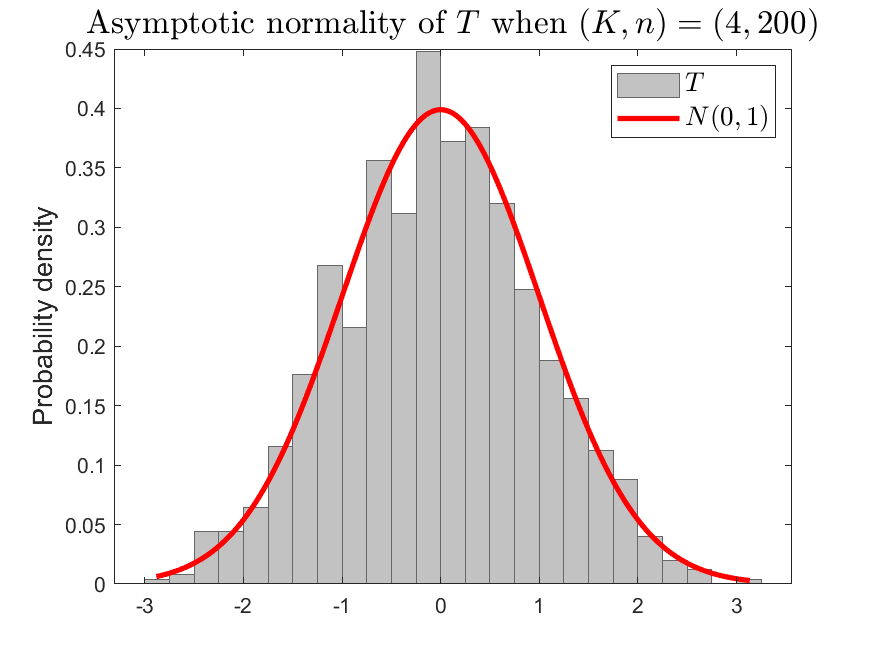}}
\subfigure{\includegraphics[width=0.33\textwidth]{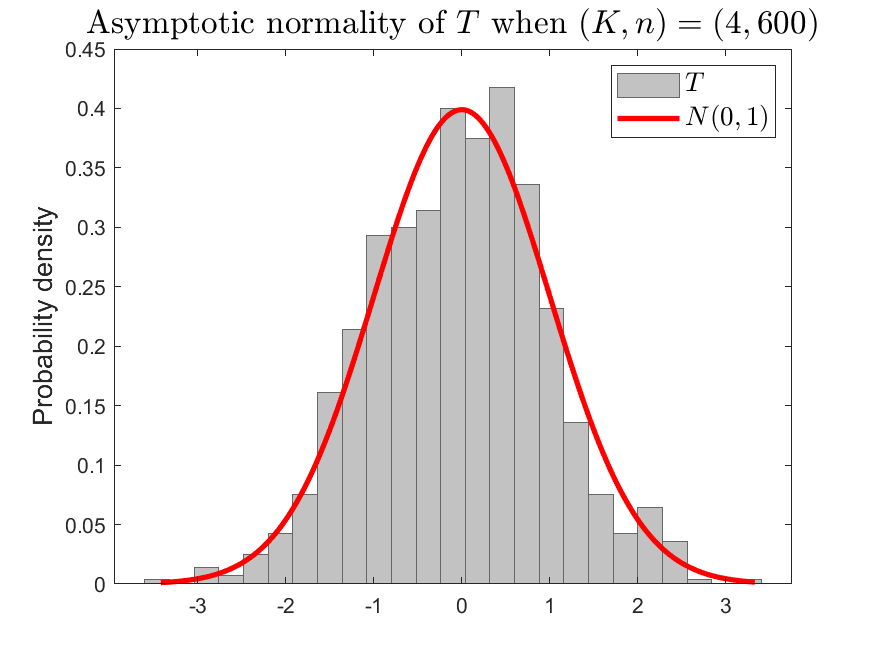}}
\subfigure{\includegraphics[width=0.33\textwidth]{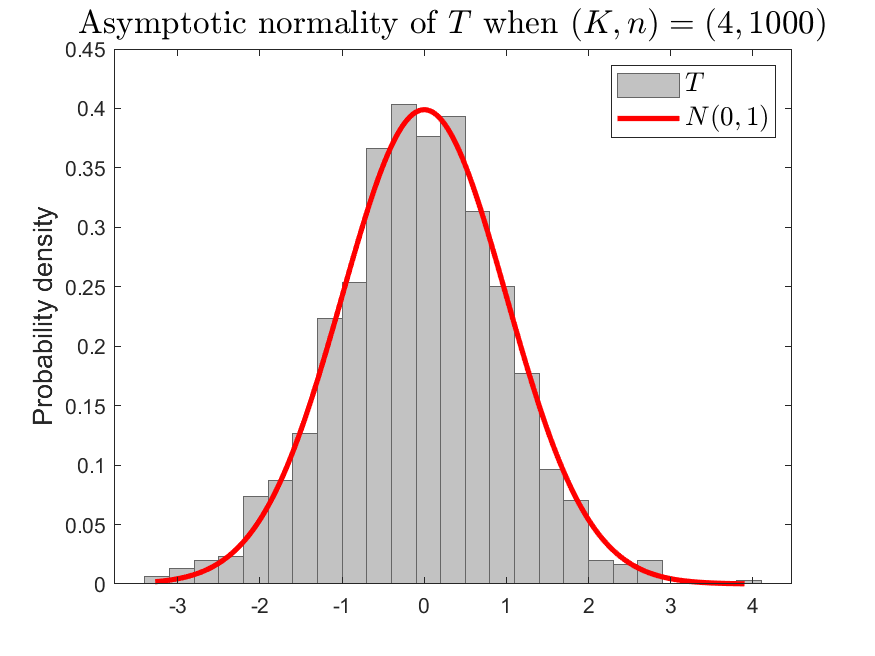}}
}
\caption{Histogram plots of $T$ for different choices of $(K,n)$, where the red curve is the probability density function of $N(0,1)$.}
\label{fig:normality} %% label for entire figure
\end{figure}

\texttt{Experiment 2: Size and power of $T$.} We evaluate the test's validity under $H_0$ and sensitivity under $H_1$ for varying true $K$. Networks ($n = 1000$, $L = 10$) are generated with $K\in \{1,2,3,4,5\}$. Connectivity matrices incorporate heterogeneous patterns: $B_{\ell,kl} = \rho(0.3 + \epsilon_\ell + 0.4 \cdot \mathbf{1}(k = l))$, where $\epsilon_\ell \sim \text{Uniform}(-0.1, 0.1)$ introduces layer-specific deviations and $\rho= 0.5$ controls network's sparsity. Table \ref{table:size_power} reports empirical rejection rates over 200 trials. We see that the proposed goodness-of-fit test demonstrates excellent statistical properties under both the null and alternative hypotheses. Under the true null hypothesis \(H_0: K = K_0\), the empirical rejection rates across all tested values of \(K_0\) (ranging from 1 to 5) consistently align with the nominal significance level \(\alpha = 0.05\), with observed rates of 0.055, 0.052, 0.056, 0.050, and 0.055, respectively. This close agreement validates the asymptotic normality of the test statistic \(T\) established in Theorem \ref{thm:main} and confirms accurate Type I error control in finite samples. Under the alternative hypothesis \(H_1: K = K_0 + 1\), the test achieves perfect empirical power (rejection rate = 1.000) in all configurations, demonstrating exceptional sensitivity to underfitting of the community structure. Crucially, these results hold robustly despite explicit incorporation of layer-specific heterogeneity in connectivity matrices \(B_\ell\), highlighting the test’s reliability under realistic multi-layer network conditions. 

\begin{table}[!ht]
\centering
\normalsize
\caption{\normalsize
Empirical rejection rate ($\alpha = 0.05$) under $H_0$ ($K = K_0$) and $H_1$ ($K = K_0 + 1$) over 200 independent trials.}
\label{table:size_power}
\setlength{\tabcolsep}{2pt} % 调整列间距
\begin{tabular}{c|c|c}
$K_0$ & Size ($K = K_0$) & Power ($K = K_0 + 1$) \\
\hline
1 & 0.055 & 1.000\\
2 & 0.052 & 1.000 \\
3 & 0.056 & 1.000 \\
4 & 0.050 & 1.000 \\
5 & 0.055 & 1.000\\
\end{tabular}
\end{table}

\texttt{Experiment 3: Accuracy of NAST and SR-NAST in estimating $K$.} We assess the accuracy of our proposed methods for true $K \in \{1, 2, 3, 4, 5\}$ under various conditions. Networks ($n = 1000$, $L = 10$) use connectivity matrices $B_{\ell,kl} = \rho(0.3 + \epsilon_\ell + 0.4 \cdot \mathbf{1}(k = l))$, where $\epsilon_\ell \sim \text{Uniform}(-0.1, 0.1)$ introduces distinct perturbations per layer and block, with \(\rho\in\{0.01,0.05,0.1,0.15,0.2,0.25,0.3\}\) controlling network's sparsity. We increase $K_0$ sequentially. Table \ref{table:accuracy} shows the proportion of correct estimates $\hat{K} = K$ over 200 trials. Based on the numerical results, we see that the proposed NAST and SR-NAST algorithms demonstrate robust performance in estimating the true number of communities \(K\) across varying sparsity levels \(\rho\), particularly for \(\rho \geq 0.1\). At the lowest sparsity level (\(\rho = 0.01\)), estimation accuracies are highly sensitive to \(K\): while accuracies remain high for \(K = 1\), they drop drastically for \(K\geq2\), indicating that extreme sparsity impedes reliable community recovery in more complex structures. However, as \(\rho\) increases moderately to \(0.1\), accuracies sharply improve across all \(K\), exceeding 92\% in all cases.

\begin{table}[!ht]
\centering
\caption{\normalsize Proportion of correct estimates of $K$ over 200 trials at different sparsity levels $\rho$ for NAST and SR-NAST algorithms.}
\label{table:accuracy}
\small
\begin{tabular}{ccccccccc}
\toprule
 & \multirow{2}{*}{Algorithm} & \multicolumn{7}{c}{Sparsity level ($\rho$)} \\
\cmidrule(lr){3-9}
$K$ & & 0.01 & 0.05 & 0.1 & 0.15 & 0.2 & 0.25 & 0.3 \\
\midrule
\multirow{2}{*}{1} & NAST & 1 & 1 & 1 &1 & 1 & 1 & 1 \\
 & SR-NAST & 1 & 1& 1& 1 & 1 & 1 & 1 \\
\midrule
\multirow{2}{*}{2} & NAST & 0.705 & 1 & 1 & 1 & 1 & 1 &1 \\
 & SR-NAST & 0.545 & 1 & 1 & 1 & 1 & 1 & 1 \\
\midrule
\multirow{2}{*}{3} & NAST & 0& 1 & 1 & 1 & 1 & 1 & 1 \\
 & SR-NAST & 0 & 1 & 1 & 1 & 1 & 1 & 1\\
\midrule
\multirow{2}{*}{4} & NAST & 0 & 0.970 & 1 & 1& 1 & 1& 1\\
 & SR-NAST & 0 & 0.845 & 1 & 1 & 1 & 1 & 1 \\
\midrule
\multirow{2}{*}{5} & NAST & 0 & 0.120 & 1 & 1 & 1 & 1 & 1 \\
 & SR-NAST & 0 & 0.5 & 0.920 & 1& 1 & 1 & 1 \\
\bottomrule
\end{tabular}
\end{table}

\begin{table}[!ht]
\centering
\caption{\normalsize Average values of the test statistic \( T \) over 200 Monte Carlo trials for different candidate community numbers \( K_0 \) when the true number of communities \( K \in \{1, 2, 3, 4, 5\} \). The first \( T \) value satisfying \( |T| < t_{n}\) with $t_{n}=\log n\approx6.9078$ when $n=1000$ (i.e., the null hypothesis \( H_0: K = K_0 \) is accepted) for each true \( K \) is highlighted in bold. The settings of network parameters \((n,L,\theta,\{B_{\ell}\}^{L}_{\ell=1})\) are the same as Experiment 3, with the sparsity parameter $\rho$ being $0.1$.}\label{table:TK0}
\begin{tabular}{c c c cccc cccc ccc}
\hline
 & \multicolumn{10}{c}{\(K_0\)} \\ 
\cmidrule{2-11}
& 1& 2&3&4 &5&6&7&8&9&10&&&\\
\hline
$K=1$ & \textbf{0.0349} & 0.0654 & 0.0776 &0.0817&0.0823&0.0874&0.0870 &0.0909&0.0912&0.0875\\
$K=2$ & 332.1328 & \textbf{0.0270} & 0.5367 &0.8206&5.8074&11.0763&18.5395&27.0380&36.8827&47.2314\\
$K=3$ & 213.7155 & 81.6456 & \textbf{-0.0336} &-0.0277&-0.0151&-0.0083&0.0293&0.1513&0.2080&0.4551\\
$K=4$& 156.9399 & 87.3642 & 33.8397&\textbf{-0.0118}&-0.0059&-0.0009&0.0010&0.0089 &0.0055&0.0216&\\
$K=5$ & 121.0352 &78.9828 & 44.1187 &17.4084&\textbf{0.0947}&0.0952&0.0969&0.1027&0.1072&0.1077\\
\hline
\end{tabular}
\end{table}

Table \ref{table:TK0} empirically validates the theoretical properties of the goodness-of-fit test statistic \( T \) under the sequential testing framework. For each true \( K \), the statistic \( T \) exhibits a sharp phase transition at \( K_0 = K \):  
\begin{itemize}
  \item For the underfitting regime (\( K_0 < K \)), we see that \( |T| \) is exceptionally large (e.g., \( T \geq 17.4084\) for \( K_0 < K \)) when the true $K$ is 5, reflecting systematic model misspecification. This aligns with Theorem \ref{thm:power}, where \( T \) diverges under \( H_1: K > K_0 \) due to unmodeled community structure. 
  \item For the critical transition regime (\( K_0 = K \)), we observe that \( T \) collapses near zero (bold values), with \( |T| \leq 0.0947 \) across all \( K \). This sudden drop confirms asymptotic normality shown in Theorems \ref{thm:main} and \ref{thm:power}.  
  \item For the overfitting regime (\( K_0 > K \)), we see that \( T \) remains stable near zero when $K_0$ is slightly larger than the true $K$ across all \( K \), indicating no evidence against \( H_0 \). Consequently, the test lacks power to reject overfitted models (\(K_0 > K\)). However, our NAST avoids this limitation by its sequential testing starting from \(K_0 = 1\), which ensures the termination at the smallest \(K_0\) where \(H_0\) is accepted, which is typically the true \(K\) (as validated in Table 3).
\end{itemize}

The consistency of this pattern in the test statistic $T$: steep divergence for \( K_0 < K \), immediate normalization at \( K_0 = K \), and sustained stability for \( K_0 > K \), validates the efficacy of the sequential testing algorithm NAST in identifying the true community count.

\begin{table}[!ht]
\centering
\caption{\normalsize Average values of the ratio statistic $\eta_{K_0}$ over 200 Monte Carlo trials for different candidate community numbers $K_0$ when the true number of communities $K \in \{2, 3, 4, 5, 6\}$. The first $\eta_{K_0}$ value satisfying $\eta_{K_0} > t_{\text{ratio},n}$ with $t_{\text{ratio},n} = \log n \approx 6.9078$ when $n = 1000$ for each true $K$ is highlighted in bold. Note that $\eta_{K_0}$ is defined from $K_0 = 2$, while $K_0 = 1$ is recorded in Table~\ref{table:TK0}. The settings of network parameters $(n,L,\theta,\{B_{\ell}\}^{L}_{\ell=1})$ are the same as Experiment 3, with the sparsity parameter $\rho$ being $0.1$.}
\label{table:etaK0}
\begin{tabular}{c c cccccccc}
\hline
 & \multicolumn{9}{c}{$K_0$} \\ 
\cmidrule{2-10}
& 2& 3&4&5&6&7&8&9&10\\
\hline
$K=2$ & \textbf{1629.01} & 1.10 & 1.66 &1.71&1.49&0.91&0.85&0.96&0.95\\
$K=3$ & 2.62 & \textbf{521.08} & 1.09 &1.03&1.03&1.05&1.46&1.57&2.40\\
$K=4$ & 1.80 & 2.59 & \textbf{132.78} &1.01&1.09&1.00&1.00&1.08&1.10\\
$K=5$ & 1.53 & 1.79 & 2.55 &\textbf{63.43}&1.00&1.10&1.04&1.27&1.03\\
$K=6$ & 1.40 & 1.52 & 1.78 &2.52&\textbf{53.66}&1.05&0.99&1.03&1.33\\
\hline
\end{tabular}
\end{table}

Table~\ref{table:etaK0} provides decisive empirical validation of the theoretical phase transition of the ratio statistic \(\eta_{K_0}\), as established in Theorem \ref{thm:ratio-behavior}. The data exhibit a clear and consistent pattern across all true community counts \(K\):
\begin{itemize}
    \item Underfitting regime (\(K_0 < K\)): For every \(K \in \{2,3,4,5,6\}\), the values of \(\eta_{K_0}\) remain stochastically bounded, with all entries satisfying \(\eta_{K_0} \leq 2.62\). This observation confirms the theoretical prediction \(\eta_{K_0} = O_P(1)\) for \(K_0 < K\).
    \item Critical transition (\(K_0 = K\)): A distinct and dramatic peak in \(\eta_{K_0}\) occurs precisely at \(K_0 = K\) (bold values). The statistic jumps to values orders of magnitude larger than those in the underfitting regime (e.g., \(\eta_4 = 132.78\) for \(K=4\), compared to \(\eta_3 = 2.59\)). This sharp increase empirically demonstrates the divergence in probability \(\eta_K \xrightarrow{P} \infty\), driven by the diverging numerator \(|T(K-1)|\) under \(H_1\) and the \(O_P(1)\) denominator \(|T(K)|\) under \(H_0\).
    \item Overfitting regime (\(K_0 > K\)): For \(K_0 > K\), \(\eta_{K_0}\) returns to magnitudes comparable to the underfitting regime, reflecting the absence of further structural signal beyond the true model order.
\end{itemize}

The empirical dichotomy—boundedness for all underspecified models versus a pronounced, unique peak at the true model—robustly underpins the stopping rule of the SR-NAST algorithm. By sequentially testing until \(\eta_{K_0}\) surpasses a slowly growing threshold (e.g., \(t_{\mathrm{ratio},n} = \log n\)), the algorithm consistently and efficiently identifies the true \(K\), leveraging this theoretically guaranteed phase transition. The consistency of this pattern across all tested values of \(K\) confirms the practical reliability of the ratio-based approach for sequential model selection.
\section{Real data}\label{sec:realdata}
In this section, we consider eight real-world networks and report their basic information in Table \ref{realdata}, where the four single-layer networks can be downloaded from \url{http://www-personal.umich.edu/~mejn/netdata/}, and the four multi-layer networks are available at \url{https://manliodedomenico.com/data.php}. 
\begin{table}[h!]
\footnotesize
	\centering
	\caption{Basic information of real-world networks used in this paper.}\label{realdata}
	\begin{tabular}{lp{2cm}p{2cm}p{2cm}llll}
\hline\hline
Dataset&Source&Node meaning&Edge meaning&Layer meaning&$n$&$L$&True $K$\\
\hline
Dolphins&\citep{lusseau2003bottlenose}&Dolphin&Companionship&NA&62&1&2\\
Football&\citep{girvan2002community}&Team&Regular-season game&NA&110&1&11\\
Polbooks&Krebs (unpublished)&Book&Co-purchasing of books by the same buyers&NA&92&1&2\\
UKfaculty&\citep{nepusz2008fuzzy}&Faculty&Friendship&NA&79&1&3\\
Lazega Law Firm&\citep{snijders2006new}&Partners and associates&Partnership&Social type&71&3&Unknown\\
C.Elegans&\citep{chen2006wiring}&Caenorhabditis elegans&Connectome&Synaptic junction&279&3&Unknown\\
CS-Aarhus&\citep{magnani2013combinatorial}&Employees&Relationship&Social type&61&5&Unknown\\
FAO-trade&\citep{de2015structural}&Countries&Trade relationship&Food product&214&364&Unknown\\
\hline\hline
\end{tabular}
\end{table}

\begin{figure}[htp!]
\centering
\resizebox{\columnwidth}{!}{
\subfigure{\includegraphics[width=0.33\textwidth]{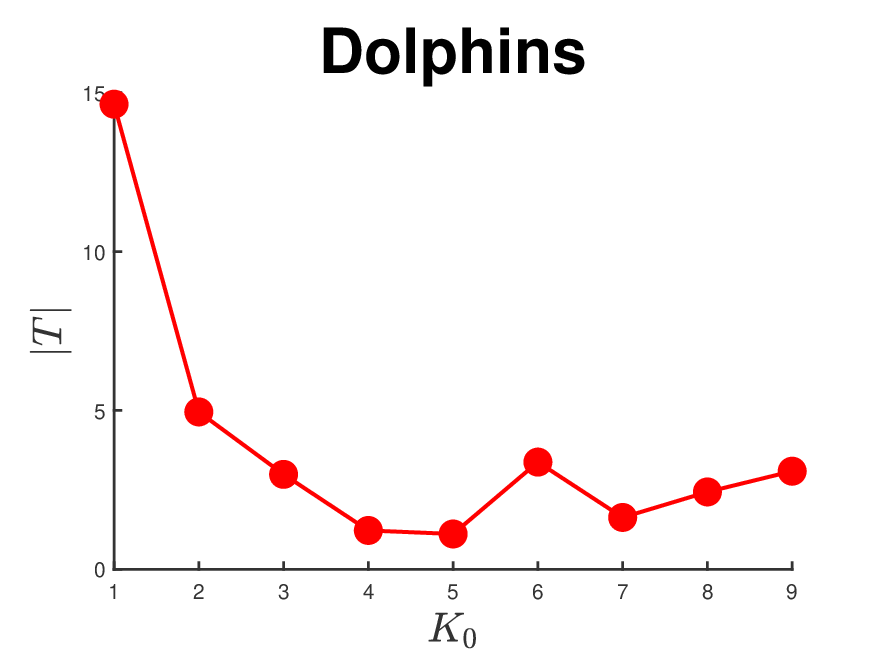}}
\subfigure{\includegraphics[width=0.33\textwidth]{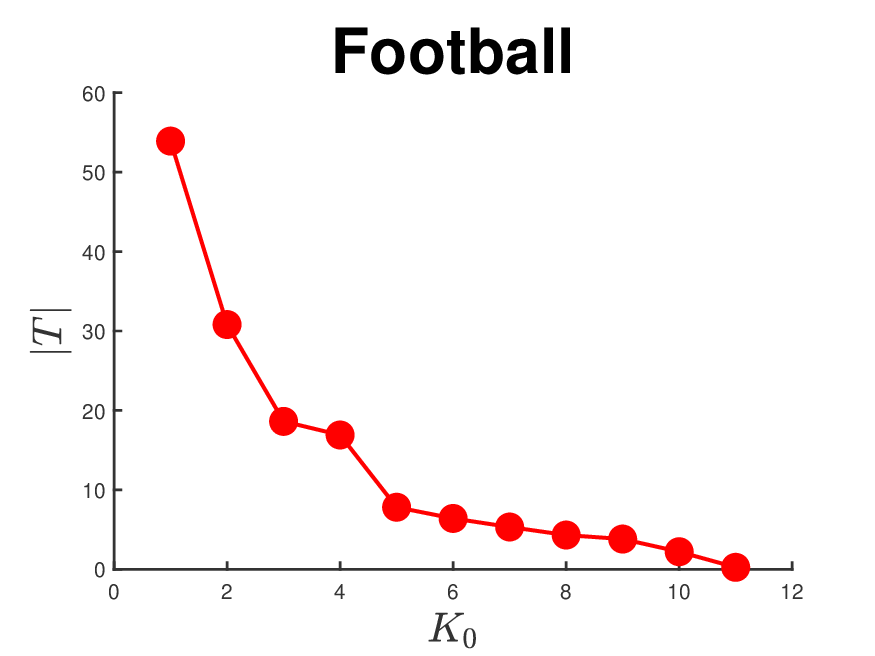}}
\subfigure{\includegraphics[width=0.33\textwidth]{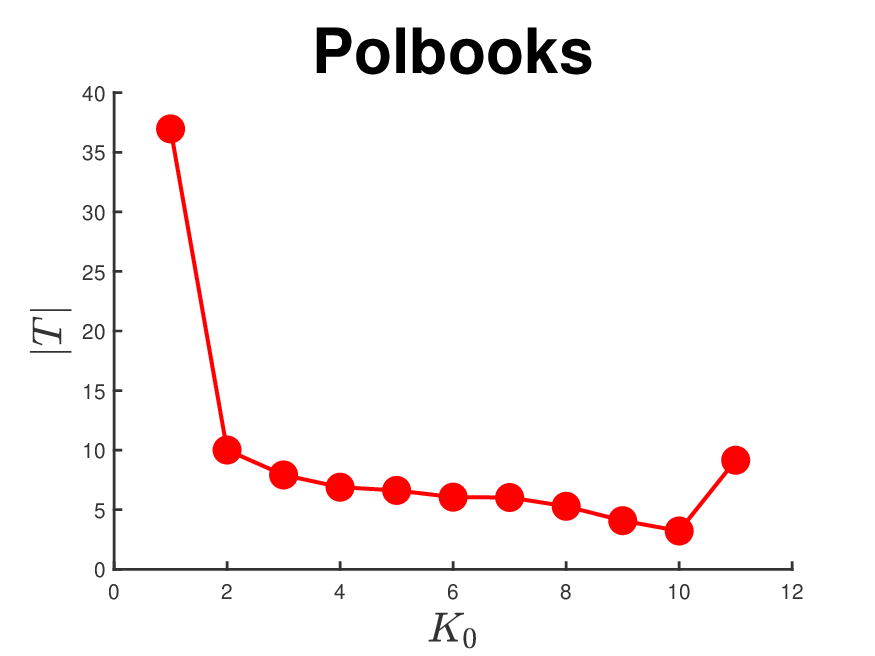}}
\subfigure{\includegraphics[width=0.33\textwidth]{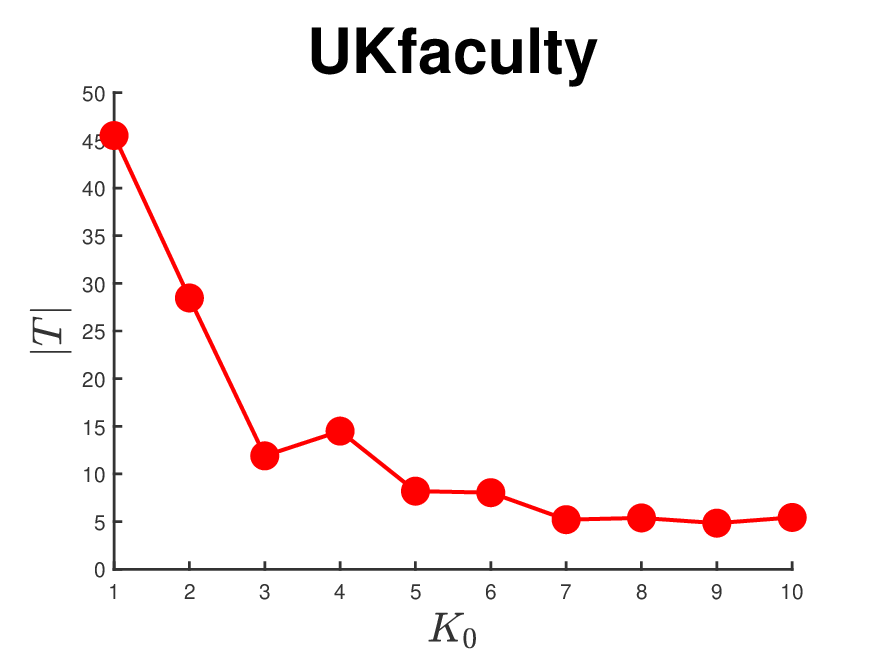}}
}
\resizebox{\columnwidth}{!}{
\subfigure{\includegraphics[width=0.33\textwidth]{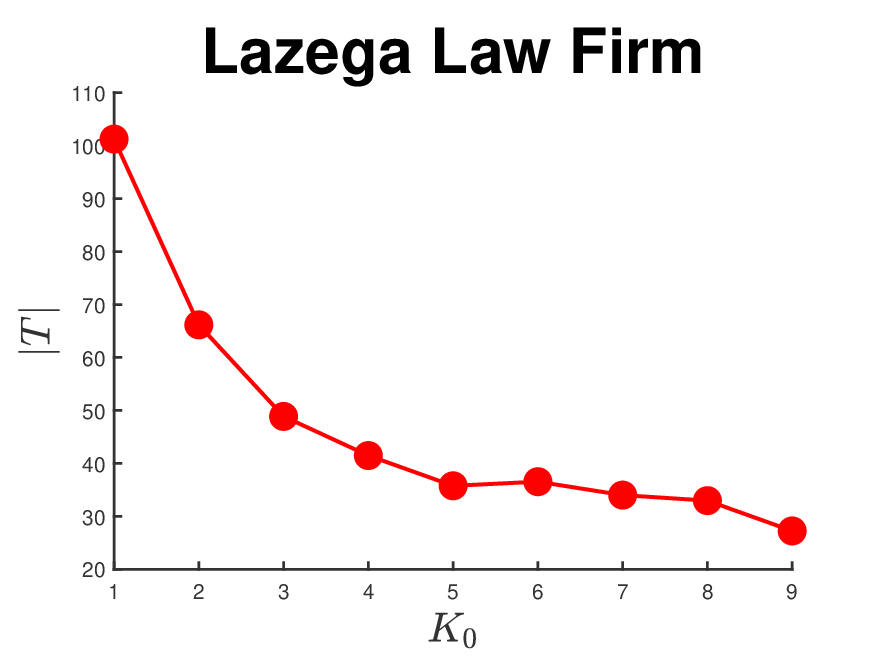}}
\subfigure{\includegraphics[width=0.33\textwidth]{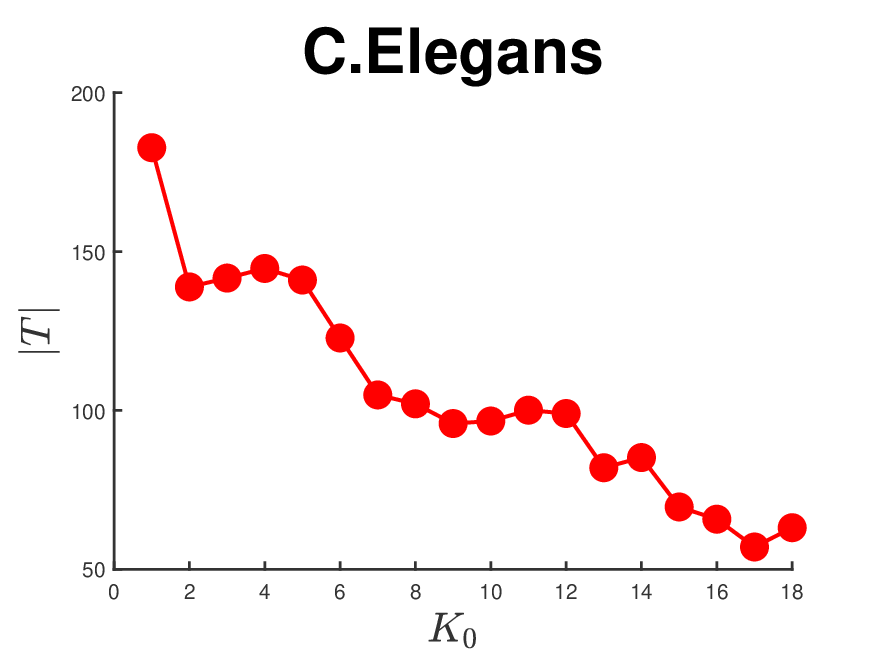}}
\subfigure{\includegraphics[width=0.33\textwidth]{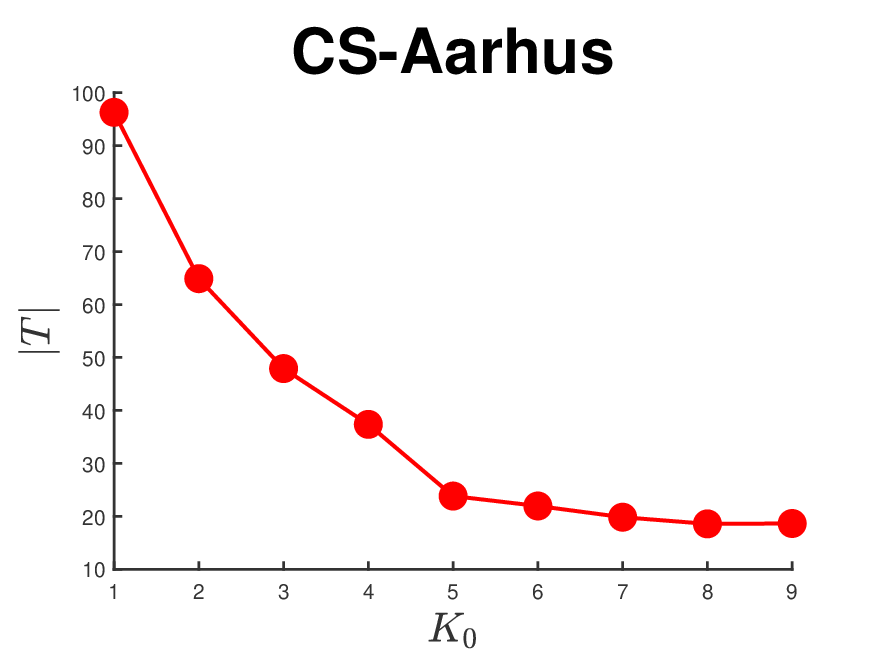}}
\subfigure{\includegraphics[width=0.33\textwidth]{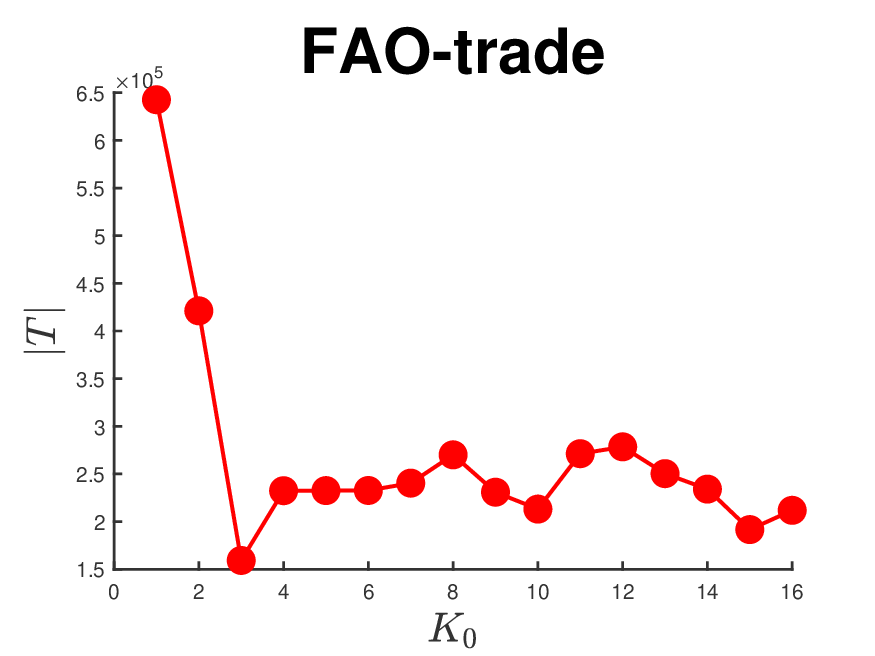}}
}
\caption{$|T|$ against increasing $K_{0}$ for the real-world networks used in this paper.}
\label{fig:realdataTK0} %% label for entire figure
\end{figure}

In Figure \ref{fig:realdataTK0}, we plot the test statistic \(|T|\) against candidate community numbers \(K_0\) for each real network. The NAST algorithm stops when \(|T(K_0)|\) falls below a threshold \(t_n\) (e.g., \(t_n = \log n\), approximately 4–6 here). However, for most networks—except FAO-trade—\(|T(K_0)|\) exceeds \(t_n\) across nearly all \(K_0\), meaning NAST would generally fail to stop at any plausible value. In such cases, the ratio statistic \(\eta_{K_0}\) used in SR-NAST, which relies on relative changes in \(|T|\), provides a more robust alternative. Figure \ref{fig:realdataetaK0} illustrates how \(\eta_{K_0}\) successfully identifies the number of communities by detecting sharp drops in \(|T|\) across all eight real-world networks. Based on Figure \ref{fig:realdataetaK0}, we have:

\begin{figure}[htp!]
\centering
\resizebox{\columnwidth}{!}{
\subfigure{\includegraphics[width=0.33\textwidth]{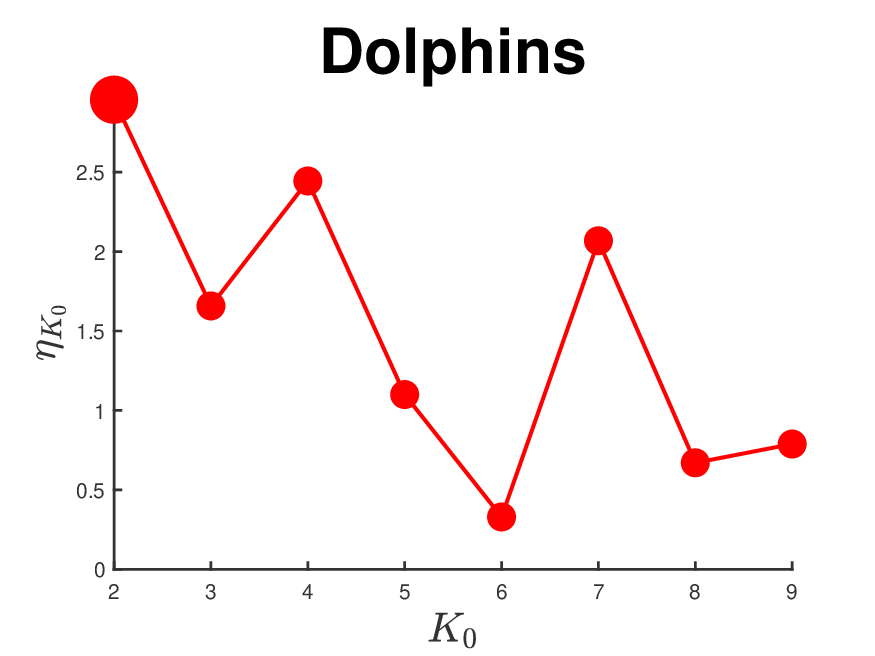}}
\subfigure{\includegraphics[width=0.33\textwidth]{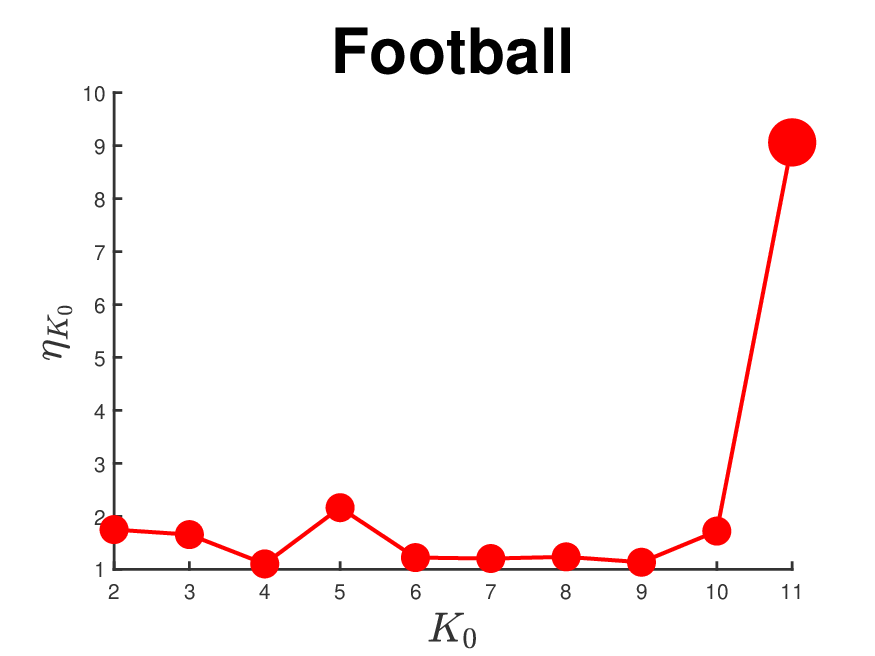}}
\subfigure{\includegraphics[width=0.33\textwidth]{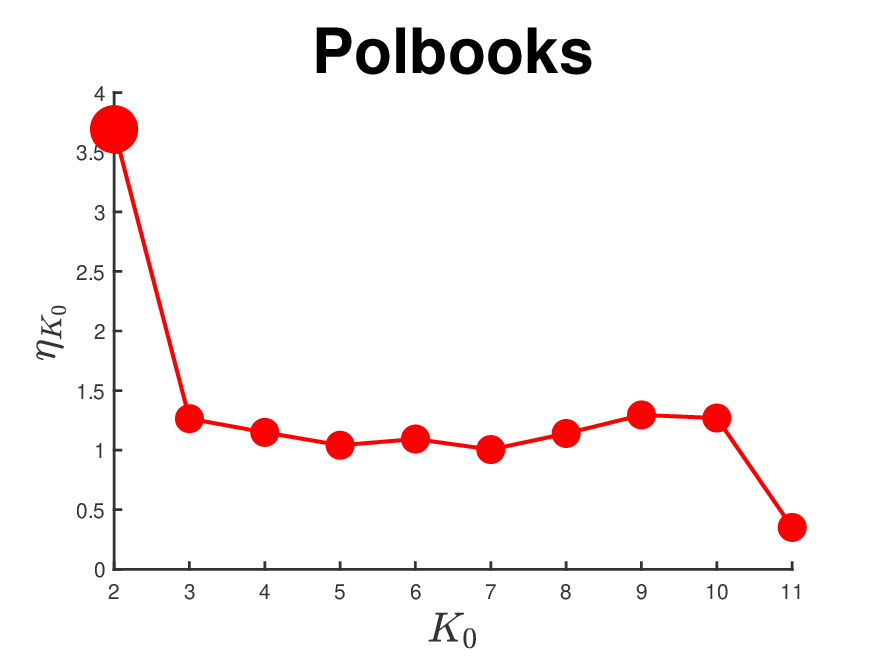}}
\subfigure{\includegraphics[width=0.33\textwidth]{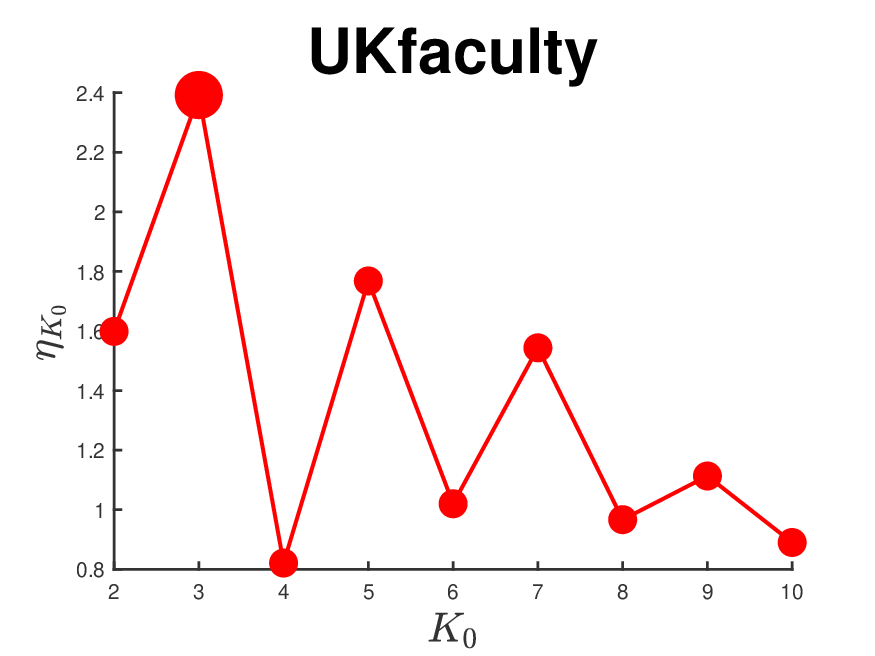}}
}
\resizebox{\columnwidth}{!}{
\subfigure{\includegraphics[width=0.33\textwidth]{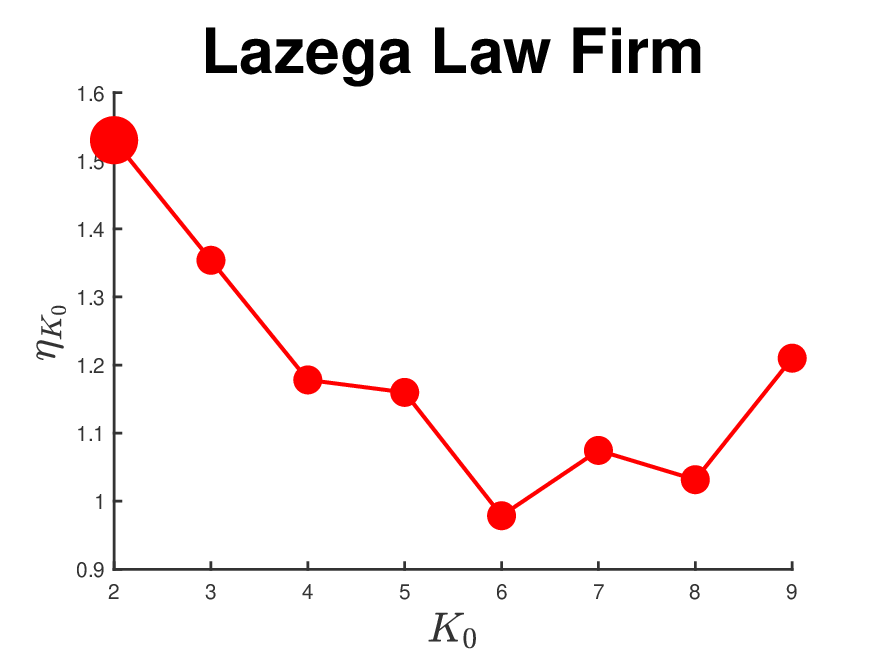}}
\subfigure{\includegraphics[width=0.33\textwidth]{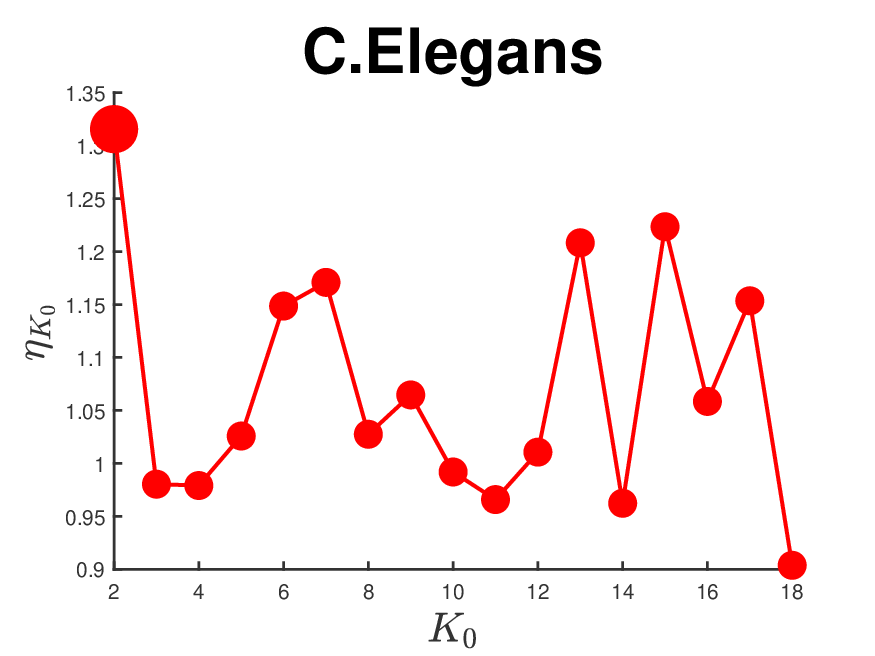}}
\subfigure{\includegraphics[width=0.33\textwidth]{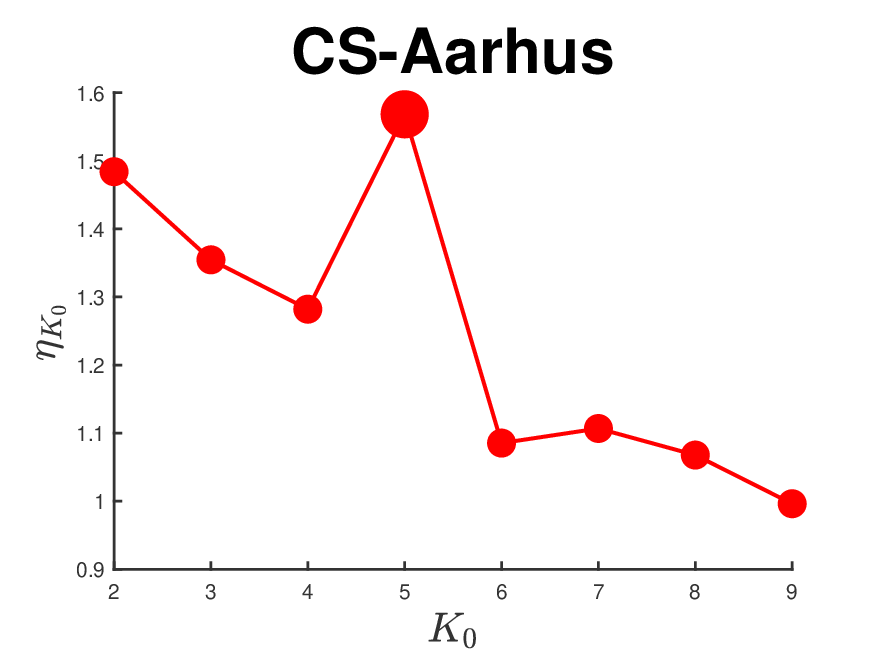}}
\subfigure{\includegraphics[width=0.33\textwidth]{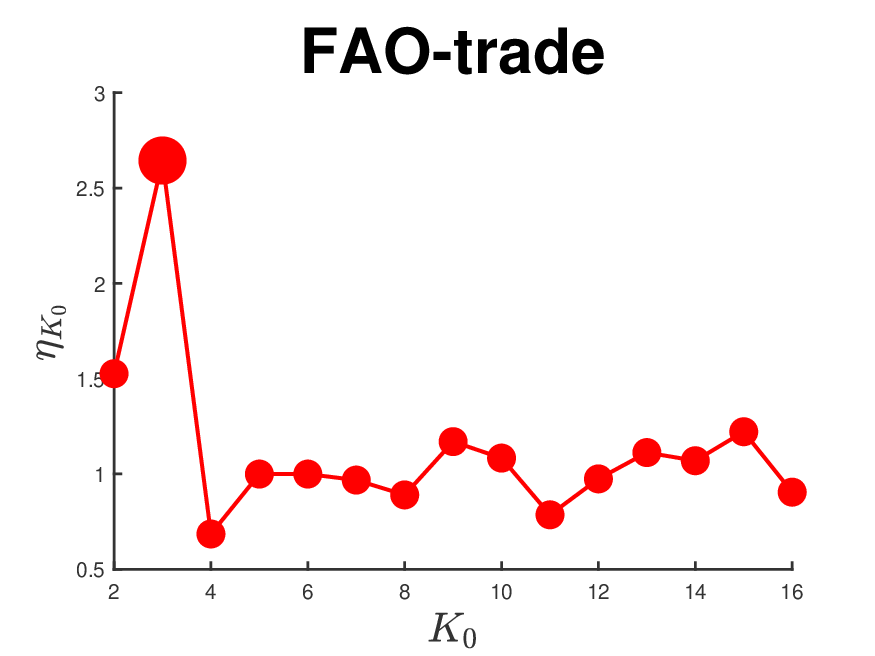}}
}
\caption{$\eta_{K_{0}}$ against increasing $K_{0}$ for the real-world networks used in this paper, with the largest $\eta_{K_{0}}$ value highlighted by a larger dot.}
\label{fig:realdataetaK0} %% label for entire figure
\end{figure}
\begin{itemize}
  \item SR-NAST exactly determines the true \(K\) for all four single-layer networks with known ground truth \(K\).   
      For Dolphins, Football, Polbooks, and UKfaclty, $\eta_{K_{0}}$ peaks at \(K_0=2, 11, 2\), and \(3\), respectively, aligning with their true number of communities.
  \item For the four real multi-layer networks with unknown true $K$, SR-NAST identifies a proper number of communities for them. In detail, SR-NAST estimates the number of communities for Lazega Law Firm, C.Elegans, CS-Aarhus, and FAO-trade as 2, 2, 5, and 3, respectively. 
\end{itemize}
\section{Conclusion}\label{sec:conclusion}
This paper introduces a principled framework for determining the number of communities in the multi-layer stochastic block model, addressing a critical gap in the analysis of complex multi-layer networks. Our approaches centers on a novel spectral-based goodness-of-fit test leveraging a normalized aggregation of layer-wise adjacency matrices. Under mild regularity conditions, we establish the asymptotic normality of a test statistic derived from the trace of the cubed normalized matrix when the candidate community count \(K_0\) is correct; we also build the asymptotic power of this test statistic under the alternative hypothesis. The two theoretical foundations facilitate two computationally efficient sequential testing procedures, where both algorithms iteratively evaluate increasing values of \(K_0\) until the null hypothesis \(H_0: K = K_0\) is accepted. Numerical experiments on both simulated and real-world multi-layer networks demonstrate the accuracy and efficiency of our methods in recovering the true number of communities. To the best of our knowledge, this is the first work for determining \(K\) in multi-layer networks with rigorous theoretical guarantees.

Looking forward, several extensions offer promising research directions. An important direction is extending this testing framework to multi-layer degree-corrected SBM to better accommodate the heterogeneous degree distributions common in real-world multi-layer networks. Extending the methodology to directed multi-layer networks or mixed membership models (where nodes belong to multiple communities) presents another significant challenge. Furthermore, developing dynamic versions of the test to handle time-varying node memberships in temporal multi-layer networks, or incorporating node/edge covariates into the framework for covariate-assisted community number estimation, could greatly broaden its applicability.

\bmhead{Acknowledgements} Qing's work was sponsored by the Scientific Research Foundation of Chongqing University of Technology (Grant No. 2024ZDR003), and the Science and Technology Research Program of Chongqing Municipal Education Commission (Grant No. KJQN202401168).
\bmhead{Author Contributions} Huan Qing: Conceptualization, Methodology, Investigation, Software, Formal analysis, Data curation, Writing-original draft, Writing-reviewing \& editing, Funding acquisition.
\section*{Declarations}
\textbf{Conflict of interest} The author declares no conflict of interest.
\begin{appendices}
\section{Proofs of theoretical results}
\subsection{Proof of Lemma \ref{lem:ideal_props}}
\begin{proof}
For expectation, we have \(\mathbb{E}[A_{\ell,ij} - P_{\ell,ij}] = 0\), so \(\mathbb{E}[\widetilde{A}^{\text{ideal}}_{ij}] = 0\). For variance for \(i \neq j\), we have
\[
\operatorname{Var}(\widetilde{A}^{\text{ideal}}_{ij}) = \frac{\operatorname{Var}\left( \sum_{\ell=1}^L (A_{\ell,ij} - P_{\ell,ij}) \right)}{n \sum_{\ell=1}^L P_{\ell,ij}(1-P_{\ell,ij})}.
\]

By Assumption \ref{assump:a3}, the variance of the sum is \(\sum_{\ell=1}^L \operatorname{Var}(A_{\ell,ij}) = \sum_{\ell=1}^L P_{\ell,ij}(1-P_{\ell,ij})\). Thus,
\[
\operatorname{Var}(\widetilde{A}^{\text{ideal}}_{ij}) = \frac{\sum_{\ell=1}^L P_{\ell,ij}(1-P_{\ell,ij})}{n \sum_{\ell=1}^L P_{\ell,ij}(1-P_{\ell,ij})} = \frac{1}{n}.
\]

Given \(\theta\), the entries \(A_{\ell,ij}\) are independent across edges and layers. Thus, \(\{\widetilde{A}^{\text{ideal}}_{ij}\}_{i<j}\) are functions of disjoint independent random variables, hence independent.
\end{proof}
\subsection{Proof of Lemma \ref{lem:ideal_normal}}
\begin{proof}
By Lemma \ref{lem:ideal_props}, we know that \(\widetilde{A}^{\text{ideal}}\) is a real symmetric matrix with:
\begin{itemize}
    \item Diagonal elements: 0,
    \item Off-diagonal elements: Independent given \(\theta\), mean 0, variance \(\frac{1}{n}\).
\end{itemize}

Thus, \(\widetilde{A}^{\text{ideal}}\) is a generalized Wigner matrix with zero diagonal. By Assumption \ref{assump:a1}, we have 
  \[
  \mathbb{E}\left[ \left( \widetilde{A}^{\text{ideal}}_{ij} \right)^4 \right] \leq \frac{C_1}{n^2} \quad \text{for} \quad C_1>0 \mathrm{~is~a~constant}.
  \]
  since \(\widetilde{A}^{\text{ideal}}_{ij} = \frac{Y}{\sqrt{n \sigma^2}}\) where \(Y = \sum_{\ell} (A_{\ell,ij}-P_{\ell,ij})\), \(\sigma^2 = \sum_{\ell} P_{\ell,ij}(1-P_{\ell,ij}) \geq L\delta(1-\delta)\), and  
  \[
  \mathbb{E}[Y^4] = \sum_{\ell} \mathbb{E}[(A_{\ell,ij}-P_{\ell,ij})^4] + 6\sum_{\ell < m} \text{Var}(A_{\ell,ij})\text{Var}(A_{m,ij}) \leq C_2L + C_3L^2,
   \]
   where $C_2$ and $C_3$ are two positive constants. Let $\tilde{A}=\sqrt{n}\widetilde{A}^{\text{ideal}}$, we see that every element of $\tilde{A}$ has mean 0, variance 1, and finite fourth moment. Set $S_{n}=\frac{\tilde{A}^{2}}{n}$. By Theorem 5.8 in \citep{bai2010spectral}, we know that the largest eigenvalue of $S_{n}$ tends to 4 almost surely, i.e., w.h.p. $\|S_{n}\|=\|\frac{\tilde{A}^{2}}{n}\|=\|(\widetilde{A}^{\text{ideal}})^{2}\|$ is 4, which implies that w.h.p. $\|\widetilde{A}^{\text{ideal}}\|=2$. Hence, we have \(\|\widetilde{A}^{\text{ideal}}\| = O_P(1)\).

Let \(\lambda_1, \dots, \lambda_n\) denote the eigenvalues of \(\widetilde{A}^{\text{ideal}}\). For \(f(x) = x^3\), define the linear spectral statistic:  
\[
G_n(f) = \operatorname{tr} \left( f(\widetilde{A}^{\text{ideal}}) \right) - n \int_{-\infty}^{\infty} f(u)  dF(u) = \operatorname{tr} \left( (\widetilde{A}^{\text{ideal}})^3 \right),
\]  
where \(F(u)\) is the semicircle law \(\frac{\sqrt{4 - u^2}}{2\pi} \mathbf{1}_{[-2,2]}(u)\). The integral vanishes because \(f(u) = u^3\) is odd and \(F(u)\) is symmetric:  
\[
\int_{-2}^{2} u^3 \frac{\sqrt{4 - u^2}}{2\pi}  du = 0.
\]

For the trace, we have 
\[
\operatorname{tr} \left( (\widetilde{A}^{\text{ideal}})^3 \right) = \sum_{i,j,k} \widetilde{A}^{\text{ideal}}_{ij} \widetilde{A}^{\text{ideal}}_{jk} \widetilde{A}^{\text{ideal}}_{ki},
\]  
where terms with repeated indices vanish (\(\widetilde{A}^{\text{ideal}}_{ii} = 0\)). For distinct \(i,j,k\), the expectation is zero:  
\[
\mathbb{E}[\widetilde{A}^{\text{ideal}}_{ij} \widetilde{A}^{\text{ideal}}_{jk} \widetilde{A}^{\text{ideal}}_{ki}] = 0,
\]  
due to independence of \(\{\widetilde{A}^{\text{ideal}}_{ij}\}_{i<j}\) and zero mean. Hence, we have
\[
\mathbb{E}\left[ \operatorname{tr} \left( (\widetilde{A}^{\text{ideal}})^3 \right) \right] = 0.
\]
 
By Theorem 2.1 of \citep{wang2021generalization}, for generalized Wigner matrices with \(\mathbb{E}[W_{ij}^4] \leq C/n^2\), \(G_n(f)\) converges weakly to \(N(0, \sigma_f^2)\). The variance \(\sigma_f^2\) is computed as:  
\[
\sigma_f^2 = \frac{1}{4\pi^2} \int_{-2}^{2} \int_{-2}^{2} f'(x)f'(y) V(x,y)  dx  dy, \quad f'(x) = 3x^2,
\]  
where \(V(x,y)\) incorporates fourth-moment dependencies. To compute \(\sigma_f^2 = 6\) for \(f(x) = x^3\), we use a combinatorial approach that leverages the structure of the trace expansion and the properties of \(\widetilde{A}^{\text{ideal}}\). Expanding the trace obtains 
\[
G_n(f) = \sum_{i,j,k} \widetilde{A}^{\text{ideal}}_{ij} \widetilde{A}^{\text{ideal}}_{jk} \widetilde{A}^{\text{ideal}}_{ki}.
\]  

Since \(\widetilde{A}^{\text{ideal}}_{ii} = 0\) (zero diagonal), non-vanishing terms require distinct \(i,j,k\). Thus, we have 
\[
G_n(f) = \sum_{i \neq j, j \neq k, k \neq i} \widetilde{A}^{\text{ideal}}_{ij} \widetilde{A}^{\text{ideal}}_{jk} \widetilde{A}^{\text{ideal}}_{ki}.
\]  

For each unordered triple of distinct nodes \(\{a,b,c\}\), there are \(3! = 6\) ordered permutations \((i,j,k)\) corresponding to the same product \(\widetilde{A}^{\text{ideal}}_{ab}\widetilde{A}^{\text{ideal}}_{bc}\widetilde{A}^{\text{ideal}}_{ca}\) (by symmetry of \(\widetilde{A}^{\text{ideal}}\)). Therefore, we get  
\[
G_n(f) = 6 \sum_{1 \leq a < b < c \leq n} \widetilde{A}^{\text{ideal}}_{ab}\widetilde{A}^{\text{ideal}}_{bc}\widetilde{A}^{\text{ideal}}_{ca}.
\]  

The variance is \(\operatorname{Var}(G_n(f)) = \mathbb{E}[G_n(f)^2]\) (since \(\mathbb{E}[G_n(f)] = 0\)). We have
\[
G_n(f)^2 = 36 \sum_{\{a,b,c\}} \sum_{\{a',b',c'\}} \left( \widetilde{A}^{\text{ideal}}_{ab}\widetilde{A}^{\text{ideal}}_{bc}\widetilde{A}^{\text{ideal}}_{ca} \right) \left( \widetilde{A}^{\text{ideal}}_{a'b'}\widetilde{A}^{\text{ideal}}_{b'c'}\widetilde{A}^{\text{ideal}}_{c'a'} \right).
\]  

By independence of entries and \(\mathbb{E}[\widetilde{A}^{\text{ideal}}_{ij}] = 0\), \(\mathbb{E}[G_n(f)^2]\) is non-zero only when \(\{a,b,c\} = \{a',b',c'\}\). For each such triple:  
\[
\mathbb{E}\left[ \left( \widetilde{A}^{\text{ideal}}_{ab}\widetilde{A}^{\text{ideal}}_{bc}\widetilde{A}^{\text{ideal}}_{ca} \right)^2 \right] = \mathbb{E}\left[ (\widetilde{A}^{\text{ideal}}_{ab})^2 \right] \mathbb{E}\left[ (\widetilde{A}^{\text{ideal}}_{bc})^2 \right] \mathbb{E}\left[ (\widetilde{A}^{\text{ideal}}_{ca})^2 \right] = \left( \frac{1}{n} \right)^3,
\]  
since \(\operatorname{Var}(\widetilde{A}^{\text{ideal}}_{ij}) = \frac{1}{n}\) for \(i \neq j\) by Lemma \ref{lem:ideal_props}. The number of unordered triples is \(\binom{n}{3}\). Thus, we have  
\[
\mathbb{E}[G_n(f)^2] = 36 \cdot \binom{n}{3} \cdot \frac{1}{n^3} = 36 \cdot \frac{n(n-1)(n-2)}{6} \cdot \frac{1}{n^3} = 6 \cdot \frac{(n-1)(n-2)}{n^2}.
\]  

As \(n \to \infty\), we get 
\[
\operatorname{Var}(G_n(f)) = 6 \cdot \frac{(n-1)(n-2)}{n^2} \to 6.
\]  

Hence, we have \(\sigma_f^2 = 6\). Combining \(\sigma_f^2 = 6\) with $\mathbb{E}\left[ \operatorname{tr} \left( (\widetilde{A}^{\text{ideal}})^3 \right) \right] = 0$ gives
\[
G_n(f) \overset{d}{\to} N(0, 6).
\]

Since \(G_n(f) = \operatorname{tr} \left( (\widetilde{A}^{\text{ideal}})^3 \right)\) and \(T^{\text{ideal}} = \frac{1}{\sqrt{6}} G_n(f)\),  we get
\[
T^{\text{ideal}} = \frac{G_n(f)}{\sqrt{6}} \overset{d}{\to} \frac{1}{\sqrt{6}} \cdot N(0,6) = N(0,1),
\]  
where Slutsky's theorem applies as the scaling is deterministic. 
\end{proof}
\subsection{Proof of Lemma \ref{lem:plugin_error}}
\begin{proof}
We show \(\left| \operatorname{tr}\left( (\widetilde{A}^{\text{agg}})^3 \right) - \operatorname{tr}\left( (\widetilde{A}^{\text{ideal}})^3 \right) \right| \xrightarrow{P} 0\). For any \(\ell,k,l\), by Lemma \ref{BlockErrorMain} and Assumptions \ref{assump:a2}-\ref{assump:a3}, we have
\[
|\hat{B}_{\ell,kl} - B_{\ell,kl}| = O_P\left( \frac{K_0(m + \sqrt{\log n})}{n} \right) \triangleq r_n.
\]

For any \(i,j,\ell\),
\[
|\hat{P}_{\ell,ij} - P_{\ell,ij}| \leq \max_{k,l} |\hat{B}_{\ell,kl} - B_{\ell,kl}| + \mathbf{1}\{\hat{\theta}_i \neq \theta_i \text{ or } \hat{\theta}_j \neq \theta_j\}.
\]

The second term has expectation \(\leq 2m/n = o(1)\) by Assumption \ref{assump:a3}, so \(\|\hat{P}_\ell - P_\ell\|_{\max} = O_P(r_n)\). For \(i \neq j\), define 
\[
D_{ij} = \sum_{\ell} P_{\ell,ij}(1-P_{\ell,ij}), \quad \hat{D}_{ij} = \sum_{\ell} \hat{P}_{\ell,ij}(1-\hat{P}_{\ell,ij}).
\]  

By Assumption \ref{assump:a1}, \(D_{ij} \geq L\delta(1-\delta)\). By Lemma \ref{UniformLowerBoundForEstimatedVariance}, \(\hat{D}_{ij} \geq L\delta(1-\delta)/2\) w.h.p. Now:  
\[
\widetilde{A}^{\text{agg}}_{ij} - \widetilde{A}^{\text{ideal}}_{ij} = \frac{Y}{\sqrt{n \hat{D}_{ij}}} - \frac{Z}{\sqrt{n D_{ij}}},
\]  
where \(Y = \sum_{\ell} (A_{\ell,ij} - \hat{P}_{\ell,ij})\), \(Z = \sum_{\ell} (A_{\ell,ij} - P_{\ell,ij})\). Decomposing  $\frac{Y}{\sqrt{n \hat{D}_{ij}}} - \frac{Z}{\sqrt{n D_{ij}}}$ gives
\[
\frac{Y}{\sqrt{n \hat{D}_{ij}}} - \frac{Z}{\sqrt{n D_{ij}}}= \underbrace{\left( \frac{Y}{\sqrt{n \hat{D}_{ij}}} - \frac{Y}{\sqrt{n D_{ij}}} \right)}_{\mathrm{term}~(I)} + \underbrace{\frac{Y - Z}{\sqrt{n D_{ij}}}}_{\mathrm{term}~(II)}.
\]  

For term (I), by the mean value theorem, we have
  \[
  \left| \frac{1}{\sqrt{\hat{D}_{ij}}} - \frac{1}{\sqrt{D_{ij}}} \right| = \frac{1}{2} \xi_{ij}^{-3/2} |\hat{D}_{ij} - D_{ij}|,
  \]  
  for \(\xi_{ij} \in [\min(D_{ij},\hat{D}_{ij}), \max(D_{ij},\hat{D}_{ij})] \geq L\delta(1-\delta)/2\) w.h.p. By simple analysis, we have
  \[
  |\hat{D}_{ij} - D_{ij}| \leq L \|\hat{P} - P\|_{\max} + O(1) = O_P(L r_n).
  \]  
  
Since \(|Y|\leq L\), we have
  \[
  |\frac{Y}{\sqrt{n \hat{D}_{ij}}} - \frac{Y}{\sqrt{n D_{ij}}}| \leq |Y| \cdot O_P\left( \frac{L r_n}{(L\delta(1-\delta)/2)^{3/2}} \right) / \sqrt{n} = O_P\left( L^{2}\cdot \frac{r_n}{L^{3/2} \sqrt{n}} \right) = O_P\left( \frac{r_n\sqrt{L}}{\sqrt{n}} \right).
  \]  

For term (II), we have  
  \[
  |Y - Z| = \left| \sum_{\ell} (P_{\ell,ij} - \hat{P}_{\ell,ij}) \right| \leq L \|\hat{P} - P\|_{\max} = O_P(L r_n).
  \]  

So, we have 
  \[
  |\frac{Y-Z}{\sqrt{nD_{ij}}}| \leq \frac{O_P(L r_n)}{\sqrt{n} \cdot \sqrt{L \delta(1-\delta)}} = O_P\left( r_n \sqrt{\frac{L}{n}} \right).
  \]  

Combining and substituting \(r_n\):  
\[
|\widetilde{A}^{\text{agg}}_{ij} - \widetilde{A}^{\text{ideal}}_{ij}| = O_P\left( r_n \sqrt{\frac{L}{n}} \right) = O_P\left( \frac{K_0(m + \sqrt{\log n})}{n} \sqrt{\frac{L}{n}} \right).
\]

Under Assumption \ref{assump:a3}, this bound converges to 0 in probability. Specifically, the dominant term is
\[
O_P\left( \frac{\sqrt{LK^2_0 \mathrm{max}(\log n, m^{2})}}{n^{3/2}} \right)=O_P\left( \sqrt{\frac{LK_0^2 \mathrm{max}(\log n, m^2)}{n^3}} \right).
\]

By Assumption \ref{assump:a3},  we have 
\[
|\widetilde{A}^{\text{agg}}_{ij} - \widetilde{A}^{\text{ideal}}_{ij}| \xrightarrow{P} 0.
\]  

Thus, we have \(\|\widetilde{A}^{\text{agg}} - \widetilde{A}^{\text{ideal}}\|_{\max} \xrightarrow{P} 0\). By Assumption \ref{assump:a3}, we have
\[
\|\widetilde{A}^{\text{agg}} - \widetilde{A}^{\text{ideal}}\|_F^2 = \sum_{i,j} |\widetilde{A}^{\text{agg}}_{ij} - \widetilde{A}^{\text{ideal}}_{ij}|^2 \leq n^2 \|\widetilde{A}^{\text{agg}} - \widetilde{A}^{\text{ideal}}\|_{\max}^2 =O_P(\frac{LK_0^2\mathrm{max}(\log n, m^2)}{n}) =o_P(1).
\]

Let \(\Delta = \widetilde{A}^{\text{agg}} - \widetilde{A}^{\text{ideal}}\). Then:  
\[
\begin{aligned}
&\left| \operatorname{tr}\left( (\widetilde{A}^{\text{agg}})^3 \right) - \operatorname{tr}\left( (\widetilde{A}^{\text{ideal}})^3 \right) \right| \\
&= \left| \operatorname{tr}\left( (\widetilde{A}^{\text{ideal}} + \Delta)^3 - (\widetilde{A}^{\text{ideal}})^3 \right) \right| \\
&= \left| \operatorname{tr}\left( 3(\widetilde{A}^{\text{ideal}})^2 \Delta + 3\widetilde{A}^{\text{ideal}} \Delta^2 + \Delta^3 \right) \right| \\
&\leq 3 \| (\widetilde{A}^{\text{ideal}})^2 \Delta \|_F + 3 \| \widetilde{A}^{\text{ideal}} \Delta^2 \|_F + \| \Delta^3 \|_F.
\end{aligned}
\]  

For the first term, we have  
  \[
  \| (\widetilde{A}^{\text{ideal}})^2 \Delta \|_F \leq \| (\widetilde{A}^{\text{ideal}})^2 \| \cdot \| \Delta \|_F \leq \| \widetilde{A}^{\text{ideal}} \|^2 \cdot \| \Delta \|_F.
  \] 
   
  By Lemma \ref{lem:ideal_normal}, \(\| \widetilde{A}^{\text{ideal}} \| = O_P(1)\), and \(\| \Delta \|_F = o_P(1)\), so this is \(o_P(1)\). For the second term, we have  
  \[
  \| \widetilde{A}^{\text{ideal}} \Delta^2 \|_F \leq \| \widetilde{A}^{\text{ideal}} \| \cdot \| \Delta^2 \|_F \leq \| \widetilde{A}^{\text{ideal}} \| \cdot \| \Delta \|_F^2 = o_P(1).
  \]  
  
For the third term, we have  
  \[
  \| \Delta^3 \|_F \leq \| \Delta \|^2 \cdot \| \Delta \|_F \leq \| \Delta \|_F^3 = o_P(1).
  \]  
  
Thus, \(\operatorname{tr}\left( (\widetilde{A}^{\text{agg}})^3 \right) - \operatorname{tr}\left( (\widetilde{A}^{\text{ideal}})^3 \right) = o_P(1)\), and \(T - T^{\text{ideal}} = o_P(1)\). 
\end{proof}
\subsection{Proof of Theorem \ref{thm:main}}
\begin{proof}
By Lemma \ref{lem:ideal_normal}, \(T^{\text{ideal}} \xrightarrow{d} N(0,1)\). By Lemma \ref{lem:plugin_error}, \(T - T^{\text{ideal}} = o_P(1)\). By Slutsky's theorem, \(T \xrightarrow{d} N(0,1)\).
\end{proof}

\subsection{Proof of Theorem \ref{thm:power}}
\begin{proof}

\textbf{Step 1: Existence of a large-signal submatrix.}
Since $K_0 < K$, by the pigeonhole principle, there exists an estimated community that contains nodes from at least two distinct true communities. Combined with Assumption \ref{assump:a2} (balanced communities), we obtain the following lemma.

\begin{lem}\label{lem:large_submatrix}
Under Assumption \ref{assump:a2} and $K_0 < K$ (i.e., under the alternative hypothesis), there exist:
\begin{itemize}
    \item an estimated community $a \in [K_0]$,
    \item two distinct true communities $k_1, k_2 \in [K]$,
    \item a third true community $k_3 \in [K]$, and
    \item an estimated community $b \in [K_0]$,
\end{itemize}
such that defining
\[
S_1 = \{ i \in \mathcal{C}_{k_1} : \hat{\theta}_i = a \}, \quad S_2 = \{ i \in \mathcal{C}_{k_2} : \hat{\theta}_i = a \}, \quad T' = \{ j \in \mathcal{C}_{k_3} : \hat{\theta}_j = b \},
\]
we have
\[
|S_1| \ge \frac{c_0 n}{2 K^2}, \quad |S_2| \ge \frac{c_0 n}{2 K^2}, \quad |T'| \ge \frac{c_0 n}{2 K^2},
\]
and
\begin{align}\label{Equ1}
\left| \sum_{\ell=1}^{L} \bigl( B_{\ell, k_1 k_3} - B_{\ell, k_2 k_3} \bigr) \right| \ge \frac{\eta L}{2}.
\end{align}
\end{lem}

\begin{proof}
For each true community \(k \in [K]\), denote by \(\mathcal{C}_k = \{ i: \theta_i = k \}\) the set of nodes belonging to true community \(k\). By Assumption \ref{assump:a2}, we have
\[
|\mathcal{C}_k| \ge \frac{c_0 n}{K} \quad \text{for all } k \in [K].
\]

The estimated community assignment \(\hat{\theta}\) partitions each \(\mathcal{C}_k\) into at most \(K_0\) disjoint subsets, each corresponding to one of the \(K_0$ estimated communities. By the pigeonhole principle, for each \(k\) there exists at least one estimated community label \(a_k \in [K_0]\) such that
\[
\bigl| \{ i \in \mathcal{C}_k : \hat{\theta}_i = a_k \} \bigr| \ge \frac{|\mathcal{C}_k|}{K_0}.
\]
Consequently,
\[
\bigl| \{ i \in \mathcal{C}_k : \hat{\theta}_i = a_k \} \bigr| \ge \frac{c_0 n}{KK_0}.
\]

Since \(K_0 < K\) and both are integers, we have \(K_0 \le K-1\). Therefore,
\[
\frac{1}{K_0} \ge \frac{1}{K-1} > \frac{1}{K},
\]
which implies
\[
\frac{c_0 n}{KK_0} \ge \frac{c_0 n}{K(K-1)} \ge \frac{c_0 n}{K^2} \frac{K}{K-1}.
\]
For \(K \ge 2\), the factor \(\frac{K}{K-1} \ge 1\). Hence, a simpler and slightly weaker lower bound is
\[
\bigl| \{ i \in \mathcal{C}_k : \hat{\theta}_i = a_k \} \bigr| \ge \frac{c_0 n}{K^2}.
\]

\(c_0 n/(2K^2) \le c_0 n/K^2\) gives
\[
\bigl| \{ i \in \mathcal{C}_k : \hat{\theta}_i = a_k \} \bigr| \ge \frac{c_0 n}{2K^2}.
\]

Now, we have defined a mapping \(k \mapsto a_k\) from the \(K\) true communities to the \(K_0\) estimated communities. Since \(K > K_0\), the pigeonhole principle guarantees that there exist two distinct true communities \(k_1\) and \(k_2\) such that \(a_{k_1} = a_{k_2}\). Denote this common estimated community by \(a\). Define
\[
S_1 = \{ i \in \mathcal{C}_{k_1} : \hat{\theta}_i = a \}, \quad S_2 = \{ i \in \mathcal{C}_{k_2} : \hat{\theta}_i = a \}.
\]

From the construction of \(a_{k_1}\) and \(a_{k_2}\), we immediately obtain
\[
|S_1| \ge \frac{c_0 n}{2K^2}, \qquad |S_2| \ge \frac{c_0 n}{2K^2}.
\]

By Assumption \ref{assump:a4}, for the pair \((k_1, k_2)\), there exists a true community \(k_3 \in [K]\) such that
\[
\bigl| \bar{B}_{k_1 k_3} - \bar{B}_{k_2 k_3} \bigr| \ge \eta.
\]

Expanding the definition of \(\bar{B}\) gives
\[
\Bigl| \sum_{\ell=1}^L \bigl( B_{\ell, k_1 k_3} - B_{\ell, k_2 k_3} \bigr) \Bigr| \ge \eta L.
\]

Since \(\eta L \ge (\eta/2) L\) for all \(L \ge 1\), we obtain
\[
\Bigl| \sum_{\ell=1}^L \bigl( B_{\ell, k_1 k_3} - B_{\ell, k_2 k_3} \bigr) \Bigr| \ge \frac{\eta L}{2},
\]
which is Equation (\ref{Equ1}).

Finally, consider the true community \(k_3\). Applying the same pigeonhole argument as before, there exists an estimated community \(b \in [K_0]\) such that
\[
\bigl| \{ j \in \mathcal{C}_{k_3} : \hat{\theta}_j = b \} \bigr| \ge \frac{|\mathcal{C}_{k_3}|}{K_0} \ge \frac{c_0 n}{K \cdot K_0} \ge \frac{c_0 n}{2K^2}.
\]

Define \(T' = \{ j \in \mathcal{C}_{k_3} : \hat{\theta}_j = b \}\). Then \(|T'| \ge \frac{c_0 n}{2K^2}\). We have thus identified \(a, k_1, k_2, k_3, b\) that satisfy all the required conditions, completing the proof.
\end{proof}

\textbf{Step 2: Decomposition of the submatrix.}
Let $S = S_1 \cup S_2$ and consider the submatrix $R = (\widetilde{A}^{\text{agg}}_{ij})_{i \in S, j \in T'}$. For $i \in S$ and $j \in T'$, since $\hat{\theta}_i = a$ and $\hat{\theta}_j = b$, we have $\hat{P}_{\ell,ij} = \hat{B}_{\ell, a b}$. Therefore,
\[
\widetilde{A}^{\text{agg}}_{ij} = \frac{\sum_{\ell=1}^{L} (A_{\ell,ij} - \hat{B}_{\ell, a b})}{\sqrt{n \hat{D}_{ab}}}, \quad \text{where } \hat{D}_{ab} = \sum_{\ell=1}^{L} \hat{B}_{\ell, a b} (1 - \hat{B}_{\ell, a b}).
\]

Denote the true probability $P_{\ell,ij} = B_{\ell, \theta_i\theta_j}$. Decompose $R$ into signal and noise parts:
\[
R = M + E, \quad \text{with } M_{ij} = \frac{\sum_{\ell=1}^{L} (P_{\ell,ij} - \hat{B}_{\ell, a b})}{\sqrt{n \hat{D}_{ab}}}, \quad E_{ij} = \frac{\sum_{\ell=1}^{L} (A_{\ell,ij} - P_{\ell,ij})}{\sqrt{n \hat{D}_{ab}}}.
\]

For $i \in S_1$ ($\theta_i = k_1$), $P_{\ell,ij} = B_{\ell, k_1 k_3}$; for $i \in S_2$ ($\theta_i = k_2$), $P_{\ell,ij} = B_{\ell, k_2 k_3}$.

\textbf{Step 3: Bounding the denominator $\hat{D}_{ab}$.}
\begin{lem}\label{lem:bound_Dab}
Suppose that Assumptions \ref{assump:a1} and \ref{assump:a3} hold, and \(K_0<K\) (i.e., under the alternative hypothesis), there exists a constant \(\kappa = \kappa(\delta) > 0\) such that:
\[
\mathbb{P}\left( \kappa L \le \hat{D}_{ab} \le \frac{L}{4} \right) \to 1,
\]
where \(\hat{D}_{ab} = \sum_{\ell=1}^L \hat{B}_{\ell,ab}(1 - \hat{B}_{\ell,ab})\) and \(\kappa := \frac{\delta(1-\delta)}{2}\).
\end{lem}

\begin{proof}
We prove the lemma in three parts: (i) establishing the upper bound, (ii) establishing the lower bound under the conditions of Lemma \ref{lem:large_submatrix}, and (iii) verifying probabilistic convergence.

\textbf{Part 1: Upper bound.}
For any \(x \in [0,1]\), the function \(g(x) = x(1-x)\) satisfies \(g(x) \le 1/4\). Therefore,
\[
\hat{D}_{ab} = \sum_{\ell=1}^L \hat{B}_{\ell,ab}(1-\hat{B}_{\ell,ab}) \le \sum_{\ell=1}^L \frac{1}{4} = \frac{L}{4}.
\]

\textbf{Part 2: Lower bound.}
The lower bound is proved under the conditions of Lemma \ref{lem:large_submatrix}, which provides specific estimated communities \(a, b \in [K_0]\) and subsets \(S_1, S_2 \subseteq \hat{\mathcal{C}}_a\), \(T' \subseteq \hat{\mathcal{C}}_b\) with sizes bounded below:
\begin{align*}
|S_1| \ge \frac{c_0 n}{2K^2}, |S_2| \ge \frac{c_0 n}{2K^2}, |T'|\ge \frac{c_0 n}{2K^2}.
\end{align*}

Consequently, the sizes of the estimated communities satisfy:
\[
\hat{n}_a := |\hat{\mathcal{C}}_a| \ge |S_1| + |S_2| \ge \frac{c_0 n}{K^2}, \qquad
\hat{n}_b := |\hat{\mathcal{C}}_b| \ge |T'| \ge \frac{c_0 n}{2K^2}.
\]

Thus,
\begin{equation}\label{eq:nab-bound}
\hat{n}_a \hat{n}_b \ge \frac{c_0^2}{2K^4} n^2 =: c_1 n^2,
\end{equation}
where \(c_1 > 0\) is a constant depending on \(c_0\) and \(K\).

Recall the definition of \(\hat{B}_{\ell,ab}\) from Equation (\ref{eq:B_hat}):
\[
\hat{B}_{\ell,ab} = \frac{1}{\hat{n}_a \hat{n}_b} \sum_{i \in \hat{\mathcal{C}}_a} \sum_{j \in \hat{\mathcal{C}}_b} A_{\ell,ij}.
\]

Define the conditional expectation given the true and estimated labels:
\[
\bar{B}_{\ell,ab} := \mathbb{E}[\hat{B}_{\ell,ab} \mid \theta, \hat{\theta}] = \frac{1}{\hat{n}_a \hat{n}_b} \sum_{i \in \hat{\mathcal{C}}_a} \sum_{j \in \hat{\mathcal{C}}_b} B_{\ell,\theta_i\theta_j}.
\]

By Assumption \ref{assump:a1}, \(B_{\ell,\theta_i\theta_j} \in [\delta, 1-\delta]\) for all \(\ell, i, j\). Since \(\bar{B}_{\ell,ab}\) is an average of numbers in \([\delta, 1-\delta]\), we have
\[
\bar{B}_{\ell,ab} \in [\delta, 1-\delta] \quad \ell\in[L].
\]

Conditional on \((\theta, \hat{\theta})\), the random variables \(\{A_{\ell,ij} : i \in \hat{\mathcal{C}}_a, j \in \hat{\mathcal{C}}_b\}\) are independent Bernoulli trials with means \(B_{\ell,\theta_i\theta_j}\). Applying Hoeffding's inequality, for any \(t > 0\):
\[
\mathbb{P}\left( \left| \hat{B}_{\ell,ab} - \bar{B}_{\ell,ab} \right| \ge t \mid \theta, \hat{\theta} \right) \le 2 \exp\left( -2 \hat{n}_a \hat{n}_b t^2 \right).
\]

Choose
\[
t = t_n := \sqrt{ \frac{3 \log n}{2 \hat{n}_a \hat{n}_b} }.
\]

Using the bound in Equation \eqref{eq:nab-bound}, we have
\[
t_n \le \sqrt{ \frac{3 \log n}{2 c_1 n^2} } = O\left( \frac{\sqrt{\log n}}{n} \right) = o(1).
\]

Substituting \(t_n\) into Hoeffding's inequality yields:
\[
\mathbb{P}\left( \left| \hat{B}_{\ell,ab} - \bar{B}_{\ell,ab} \right| \ge t_n \mid \theta, \hat{\theta} \right) \le 2 \exp\left( -3 \log n \right) = 2n^{-3}.
\]

Taking a union bound over all \(L\) layers:
\[
\mathbb{P}\left( \exists \ell \in [L] : \left| \hat{B}_{\ell,ab} - \bar{B}_{\ell,ab} \right| \ge t_n \mid \theta, \hat{\theta} \right) \le 2L n^{-3}.
\]

By Assumption \ref{assump:a3}, we have \(L = o(n/\log n)\), which gives \(2L n^{-3} = o(1)\). Therefore,
\[
\mathbb{P}\left( \forall \ell \in [L] : \left| \hat{B}_{\ell,ab} - \bar{B}_{\ell,ab} \right| \le t_n \mid \theta, \hat{\theta} \right) = 1 - o(1).
\]

Since \(t_n = o(1)\), there exists \(N_0\) such that for all \(n \ge N_0\), \(t_n < \delta/2\). Then, on the event that \(\left| \hat{B}_{\ell,ab} - \bar{B}_{\ell,ab} \right| \le t_n\) for all \(\ell\), we have:
\[
\hat{B}_{\ell,ab} \in [\bar{B}_{\ell,ab} - t_n, \bar{B}_{\ell,ab} + t_n] \subseteq [\delta - t_n, 1-\delta + t_n] \subseteq [\delta/2, 1-\delta/2].
\]

Now consider the function \(g(x) = x(1-x)\). On the interval \([\delta/2, 1-\delta/2]\), the minimum is attained at the endpoints:
\[
\min_{x \in [\delta/2, 1-\delta/2]} g(x) = g(\delta/2) = \frac{\delta}{2}\left(1 - \frac{\delta}{2}\right).
\]

We verify that this minimum is at least \(\frac{\delta(1-\delta)}{2}\):
\[
\frac{\delta}{2}\left(1 - \frac{\delta}{2}\right) - \frac{\delta(1-\delta)}{2} = \frac{\delta^2}{4} > 0.
\]

Thus, for all \(x \in [\delta/2, 1-\delta/2]\),
\[
g(x) \ge \frac{\delta}{2}\left(1 - \frac{\delta}{2}\right) \ge \frac{\delta(1-\delta)}{2} =: \kappa.
\]

Consequently, on the high-probability event, for each \(\ell\):
\[
\hat{B}_{\ell,ab}(1 - \hat{B}_{\ell,ab}) \ge \kappa.
\]

Summing over \(\ell\) gives:
\[
\hat{D}_{ab} = \sum_{\ell=1}^L \hat{B}_{\ell,ab}(1 - \hat{B}_{\ell,ab}) \ge \kappa L.
\]

\textbf{Part 3: Probabilistic convergence.}
Define the event:
\[
\mathcal{A}_n := \left\{ \kappa L \le \hat{D}_{ab} \le \frac{L}{4} \right\}.
\]

We have shown that for any realization of \((\theta, \hat{\theta})\) satisfying the size bound in Equation \eqref{eq:nab-bound} (which holds with probability one under the construction of Lemma \ref{lem:large_submatrix}),
\[
\mathbb{P}\left( \mathcal{A}_n \mid \theta, \hat{\theta} \right) \ge 1 - 2L n^{-3} = 1 - o(1).
\]

Since the conditional probability converges to 1 uniformly over such \((\theta, \hat{\theta})\), the unconditional probability also converges to 1 by the law of total probability and the bounded convergence theorem
\[
\mathbb{P}(\mathcal{A}_n) = \mathbb{E}\left[ \mathbb{P}\left( \mathcal{A}_n \mid \theta, \hat{\theta} \right) \right] \to 1 \quad \text{as } n \to \infty.
\]

This completes the proof. \qedhere
\end{proof}

\textbf{Step 4: Lower bound on the spectral norm of the signal matrix $M$.}

The matrix $M$ is defined as the signal part of the submatrix $R$, i.e., for $i \in S = S_1 \cup S_2$ and $j \in T'$,
\[
M_{ij} = \frac{\sum_{\ell=1}^{L} (P_{\ell,ij} - \hat{P}_{\ell,ij})}{\sqrt{n \hat{D}_{ab}}}.
\]

Recall that $S_1$ and $S_2$ are subsets of the estimated community $\hat{\mathcal{C}}_a$ (with $\hat{\theta}_i = a$ for $i \in S$), and $T'$ is a subset of the estimated community $\hat{\mathcal{C}}_b$ (with $\hat{\theta}_j = b$ for $j \in T'$). For $i \in S_1$ we have true community $\theta_i = k_1$, and for $i \in S_2$, $\theta_i = k_2$. For all $j \in T'$, $\theta_j = k_3$.

By definition, $\hat{P}_{\ell,ij} = \hat{B}_{\ell, \hat{\theta}_i \hat{\theta}_j}$. Since $\hat{\theta}_i = a$ and $\hat{\theta}_j = b$ for all $i \in S, j \in T'$, we have uniformly
\[
\hat{P}_{\ell,ij} = \hat{B}_{\ell, a b} \quad \text{for all } i \in S, j \in T', \ell \in [L].
\]

Meanwhile, the true probability $P_{\ell,ij} = B_{\ell, \theta_i \theta_j}$ takes two distinct values:
\[
P_{\ell,ij} = 
\begin{cases}
B_{\ell, k_1 k_3}, & i \in S_1, \\
B_{\ell, k_2 k_3}, & i \in S_2.
\end{cases}
\]

Thus, the entries of $M$ are constant on each block:
\[
M_{ij} = 
\begin{cases}
\dfrac{\sum_{\ell=1}^{L} (B_{\ell, k_1 k_3} - \hat{B}_{\ell, a b})}{\sqrt{n \hat{D}_{ab}}}, & i \in S_1, \; j \in T', \\[10pt]
\dfrac{\sum_{\ell=1}^{L} (B_{\ell, k_2 k_3} - \hat{B}_{\ell, a b})}{\sqrt{n \hat{D}_{ab}}}, & i \in S_2, \; j \in T'.
\end{cases}
\]

Define the aggregated differences
\[
\alpha_1 = \sum_{\ell=1}^L (B_{\ell,k_1 k_3} - \hat{B}_{\ell,ab}), \qquad 
\alpha_2 = \sum_{\ell=1}^L (B_{\ell,k_2 k_3} - \hat{B}_{\ell,ab}).
\]

Crucially, the same estimate $\hat{B}_{\ell,ab}$ appears in both $\alpha_1$ and $\alpha_2$ because all involved rows share the same estimated community label $a$ and all columns share the same estimated community label $b$. Therefore,
\[
\alpha_1 - \alpha_2 = \sum_{\ell=1}^{L} (B_{\ell, k_1 k_3} - \hat{B}_{\ell, a b}) - \sum_{\ell=1}^{L} (B_{\ell, k_2 k_3} - \hat{B}_{\ell, a b})
                  = \sum_{\ell=1}^{L} (B_{\ell, k_1 k_3} - B_{\ell, k_2 k_3}). 
\]

Now $M$ can be written in outer-product form as
\[
M = \frac{1}{\sqrt{n\hat{D}_{ab}}} \begin{bmatrix} \alpha_1 \mathbf{1}_{|S_1|} \\ \alpha_2 \mathbf{1}_{|S_2|} \end{bmatrix} \mathbf{1}_{|T'|}^\top.
\]

The spectral norm of $M$ is then given by
\[
\|M\| = \frac{1}{\sqrt{n\hat{D}_{ab}}} \left\| \begin{bmatrix} \alpha_1 \mathbf{1}_{|S_1|} \\ \alpha_2 \mathbf{1}_{|S_2|} \end{bmatrix} \right\| \cdot \left\| \mathbf{1}_{|T'|} \right\|
       = \frac{1}{\sqrt{n\hat{D}_{ab}}} \sqrt{ \alpha_1^2 |S_1| + \alpha_2^2 |S_2| } \cdot \sqrt{|T'|}.
\]

We proceed to lower bound each term. From Lemma \ref{lem:large_submatrix} we have the community separation condition
\[
\left| \sum_{\ell=1}^{L} \bigl( B_{\ell, k_1 k_3} - B_{\ell, k_2 k_3} \bigr) \right| \ge \frac{\eta L}{2},
\]
which gives
\[
|\alpha_1 - \alpha_2| \ge \frac{\eta L}{2}. 
\]

Applying the triangle inequality,
\[
|\alpha_1| + |\alpha_2| \ge |\alpha_1 - \alpha_2| \ge \frac{\eta L}{2},
\]
which immediately implies
\[
\max(|\alpha_1|, |\alpha_2|) \ge \frac{\eta L}{4}. 
\]

Without loss of generality, assume $|\alpha_1| \ge \eta L/4$, so that $\alpha_1^2 \ge (\eta L/4)^2$.

By Lemma \ref{lem:bound_Dab}, we get
\[
\hat{D}_{ab} = \sum_{\ell=1}^L \hat{B}_{\ell,ab}(1-\hat{B}_{\ell,ab}) \le \frac{L}{4}. 
\]

From Lemma \ref{lem:large_submatrix}, we have
\[
|S_1| \ge \frac{c_0 n}{2 K^2}, \qquad |T'| \ge \frac{c_0 n}{2 K^2}. 
\]

Noting that $\sqrt{\alpha_1^2 |S_1| + \alpha_2^2 |S_2|} \ge |\alpha_1| \sqrt{|S_1|}$ (since all terms are nonnegative), we obtain
\[
\begin{aligned}
\|M\| &\ge \frac{1}{\sqrt{n\hat{D}_{ab}}} |\alpha_1| \sqrt{|S_1|} \sqrt{|T'|} \\
      &\ge \frac{1}{\sqrt{n \cdot (L/4)}} \cdot \frac{\eta L}{4} \cdot \sqrt{\frac{c_0 n}{2K^2}} \cdot \sqrt{\frac{c_0 n}{2K^2}} \\
      &= \frac{\eta}{4} \cdot \frac{2}{\sqrt{nL}} \cdot L \cdot \frac{c_0 n}{2K^2} \\
      &= \frac{\eta c_0}{4} \cdot \frac{\sqrt{nL}}{K^2}.
\end{aligned}
\]

Thus, there exists a constant $C_M = \eta c_0/4 > 0$ such that with probability tending to 1,
\begin{align}\label{eq2}
\|M\| \ge C_M \frac{\sqrt{nL}}{K^2}.
\end{align}

\textbf{Step 5: Controlling the spectral norm of the noise matrix $E$.}
\begin{lem}\label{lem:noise_bound}
Under Assumptions \ref{assump:a1} and \ref{assump:a3}, with probability tending to 1, $\|E\| = O(\sqrt{\log n})$.
\end{lem}
\begin{proof}
Recall the noise matrix $E$ with dimensions $m_1 = |S|$, $m_2 = |T'|$, where $S = S_1 \cup S_2$, and entries
\[
E_{ij} = \frac{\sum_{\ell=1}^{L} (A_{\ell,ij} - P_{\ell,ij})}{\sqrt{n \hat{D}_{ab}}}, \qquad i \in S, \; j \in T',
\]
where $\hat{D}_{ab} = \sum_{\ell=1}^L \hat{B}_{\ell,ab}(1 - \hat{B}_{\ell,ab})$. By Lemma \ref{lem:bound_Dab}, there exists a constant $\kappa = \delta(1-\delta)/2 > 0$ such that
\[
\mathbb{P}\bigl( \hat{D}_{ab} \geq \kappa L \bigr) \to 1 \quad \text{as } n \to \infty.
\]

Define the event $\mathcal{E}_n = \{\hat{D}_{ab} \geq \kappa L\}$. Then $\lim_{n\to\infty} \mathbb{P}(\mathcal{E}_n) = 1$. Conditional on $\mathcal{E}_n$, for any $i \in S$, $j \in T'$:
\begin{itemize}
    \item $\mathbb{E}[E_{ij} \mid \theta, \hat{\theta}, \mathcal{E}_n] = 0$ because $\mathbb{E}[A_{\ell,ij} \mid \theta] = P_{\ell,ij}$.
    \item Using $P_{\ell,ij}(1-P_{\ell,ij}) \leq 1/4$ from Assumption \ref{assump:a1} gives
    \[
    \mathbb{E}[E_{ij}^2 \mid \theta, \hat{\theta}, \mathcal{E}_n] 
    = \frac{\sum_{\ell=1}^L P_{\ell,ij}(1-P_{\ell,ij})}{n \hat{D}_{ab}} 
    \leq \frac{L/4}{n \kappa L} = \frac{1}{4\kappa n}.
    \]
    \item Since $|A_{\ell,ij} - P_{\ell,ij}| \leq 1$,
    \[
    |E_{ij}| \leq \frac{L}{\sqrt{n \hat{D}_{ab}}} \leq \frac{L}{\sqrt{n \kappa L}} = \sqrt{\frac{L}{\kappa n}} \triangleq R_n.
    \]
\end{itemize}

By Assumption \ref{assump:a3} ($L = o(n/\log n)$), we have $R_n = o(1/\sqrt{\log n})$. Let $\widetilde{E}$ be the Hermitian dilation of $E$:
\[
H = \begin{pmatrix} 0 & E \\ E^\top & 0 \end{pmatrix} \in \mathbb{R}^{d \times d},
\]
where $d = m_1 + m_2 \leq 2n$. Note that $\|E\| = \|H\|$. For each $(i,j)$ with $i \in [m_1]$, $j \in [m_2]$, define
\[
Y_{ij} = \begin{pmatrix} 0 & E_{ij} e_i^{(m_1)} (e_j^{(m_2)})^\top \\ 
                      E_{ij} e_j^{(m_2)} (e_i^{(m_1)})^\top & 0 \end{pmatrix},
\]
where $e_i^{(m_1)}$ and $e_j^{(m_2)}$ are standard basis vectors. Then
\[
H = \sum_{i=1}^{m_1} \sum_{j=1}^{m_2} Y_{ij}.
\]

Conditional on $(\theta, \hat{\theta})$ and $\mathcal{E}_n$, the matrices $\{Y_{ij}\}$ are independent, satisfy $\mathbb{E}[Y_{ij} \mid \theta, \hat{\theta}, \mathcal{E}_n] = 0$, and have spectral norm bounded by
\[
\|Y_{ij}\| = |E_{ij}| \leq R_n.
\]

Define
\[
V = \sum_{i=1}^{m_1} \sum_{j=1}^{m_2} \mathbb{E}[Y_{ij}^2 \mid \theta, \hat{\theta}, \mathcal{E}_n].
\]

We compute $Y_{ij}^2$:
\[
Y_{ij}^2 = 
\begin{pmatrix}
E_{ij}^2 e_i^{(m_1)} (e_i^{(m_1)})^\top & 0 \\
0 & E_{ij}^2 e_j^{(m_2)} (e_j^{(m_2)})^\top
\end{pmatrix}.
\]

Thus,
\[
V = \begin{pmatrix}
\sum_{i=1}^{m_1} \bigl( \sum_{j=1}^{m_2} \mathbb{E}[E_{ij}^2 \mid \theta, \hat{\theta}, \mathcal{E}_n] \bigr) e_i^{(m_1)} (e_i^{(m_1)})^\top & 0 \\
0 & \sum_{j=1}^{m_2} \bigl( \sum_{i=1}^{m_1} \mathbb{E}[E_{ij}^2 \mid \theta, \hat{\theta}, \mathcal{E}_n] \bigr) e_j^{(m_2)} (e_j^{(m_2)})^\top
\end{pmatrix}.
\]

The matrix $V$ is block-diagonal. Its spectral norm equals the maximum of the norms of the two diagonal blocks. For the first block, using $\mathbb{E}[E_{ij}^2 \mid \theta, \hat{\theta}, \mathcal{E}_n] \leq 1/(4\kappa n)$,
\[
\left\| \sum_{i=1}^{m_1} \bigl( \sum_{j=1}^{m_2} \mathbb{E}[E_{ij}^2 \mid \theta, \hat{\theta}, \mathcal{E}_n] \bigr) e_i^{(m_1)} (e_i^{(m_1)})^\top \right\|
= \max_{1 \leq i \leq m_1} \sum_{j=1}^{m_2} \mathbb{E}[E_{ij}^2 \mid \theta, \hat{\theta}, \mathcal{E}_n] \leq \frac{m_2}{4\kappa n}.
\]

Similarly, the second block has norm at most $\frac{m_1}{4\kappa n}$. Since $m_1, m_2 \leq n$, we obtain
\[
\|V\| \leq \frac{1}{4\kappa}.
\]

We apply Theorem 1.6 (Matrix Bernstein: rectangular case) of \cite{tropp2012user} to the independent, centered, symmetric matrices $\{Y_{ij}\}$. With parameters $R = R_n$ and $\sigma^2 = \|V\| \leq 1/(4\kappa)$, for any $t > 0$,
\[
\mathbb{P}\bigl( \|H\| \ge t \mid \theta, \hat{\theta}, \mathcal{E}_n \bigr) 
\le 2d \cdot \exp\left( -\frac{t^2/2}{\sigma^2 + R t/3} \right).
\]

Choose $t = 2\sqrt{\frac{\log n}{\kappa}}$. This choice satisfies $t = O(\sqrt{\log n})$. We now verify that for large $n$, $R t/3 \leq \sigma^2$. Since $\sigma^2 \geq 0$, it suffices to show $R t/3 \leq 1/(4\kappa)$. Observe
\[
R t = \sqrt{\frac{L}{\kappa n}} \cdot 2\sqrt{\frac{\log n}{\kappa}} = 2\sqrt{\frac{L \log n}{\kappa^2 n}}.
\]

By Assumption \ref{assump:a3}, $L = o(n/\log n)$, hence $R t = o(1)$. Therefore, for sufficiently large $n$, $R t/3 \leq 1/(4\kappa)$. Consequently,
\[
\sigma^2 + R t/3 \leq \frac{1}{4\kappa} + \frac{1}{4\kappa} = \frac{1}{2\kappa}.
\]

Now plug into the Bernstein inequality, we get
\begin{align*}
\mathbb{P}\bigl( \|H\| \ge t \mid \theta, \hat{\theta}, \mathcal{E}_n \bigr) 
&\le 2d \cdot \exp\left( -\frac{t^2/2}{1/(2\kappa)} \right) \\
&= 2d \cdot \exp\left( -\kappa t^2 \right) \\
&= 2d \cdot \exp\left( -\kappa \cdot 4\frac{\log n}{\kappa} \right) \\
&= 2d \cdot e^{-4\log n} \\
&= 2d \cdot n^{-4}.
\end{align*}

Since $d = m_1 + m_2 \leq 2n$, we have
\[
\mathbb{P}\bigl( \|H\| \ge 2\sqrt{\frac{\log n}{\kappa}} \mid \theta, \hat{\theta}, \mathcal{E}_n \bigr) \le 4n \cdot n^{-4} = 4n^{-3}.
\]

Finally,
\begin{align*}
\mathbb{P}\bigl( \|E\| \ge 2\sqrt{\frac{\log n}{\kappa}} \bigr)
&= \mathbb{P}\bigl( \|H\| \ge 2\sqrt{\frac{\log n}{\kappa}} \bigr) \\
&\le \mathbb{P}\bigl( \|H\| \ge 2\sqrt{\frac{\log n}{\kappa}}, \mathcal{E}_n \bigr) + \mathbb{P}(\mathcal{E}_n^c) \\
&= \mathbb{E}\left[ \mathbb{P}\bigl( \|H\| \ge 2\sqrt{\frac{\log n}{\kappa}} \mid \theta, \hat{\theta}, \mathcal{E}_n \bigr) \mathbf{1}_{\mathcal{E}_n} \right] + \mathbb{P}(\mathcal{E}_n^c) \\
&\le 4n^{-3} + \mathbb{P}(\mathcal{E}_n^c).
\end{align*}

Because $\mathbb{P}(\mathcal{E}_n^c) \to 0$, we conclude
\[
\lim_{n\to\infty} \mathbb{P}\bigl( \|E\| \ge 2\sqrt{\frac{\log n}{\kappa}} \bigr) = 0.
\]

Hence, $\|E\| = O_P(\sqrt{\log n})$. More specifically, for any $\epsilon > 0$, there exists $C > 0$ such that $\mathbb{P}(\|E\| \ge C\sqrt{\log n}) \le \epsilon$ for all large $n$. This completes the proof.
\end{proof}

\textbf{Step 6: Lower bound on the spectral norm of $R$ and $\widetilde{A}^{\text{agg}}$.}

Recall from Steps 4--5 that we have constructed the submatrix $R = (\widetilde{A}^{\text{agg}}_{ij})_{i \in S, j \in T'}$, where $S = S_1 \cup S_2$ and $T'$ are node subsets with cardinalities at least $c_0 n/(2K^2)$, and we have the decomposition $R = M + E$. From Equation (\ref{eq2}) and Lemma \ref{lem:noise_bound}, we have the following probabilistic bounds:

\begin{enumerate}
    \item There exists a constant $C_M = \eta c_0/4 > 0$ such that
    \begin{align}\label{eqA}
 \mathbb{P}\left( \|M\| \ge C_M \frac{\sqrt{nL}}{K^2} \right) \to 1 \quad \text{as } n \to \infty.
    \end{align}
    
    \item There exists a constant $C_E > 0$ such that
    \begin{align}\label{eqB}
\mathbb{P}\left( \|E\| \le C_E \sqrt{\log n} \right) \to 1 \quad \text{as } n \to \infty. 
    \end{align}
    That is, $\|E\| = O_P(\sqrt{\log n})$.
\end{enumerate}

Define the event
\[
\mathcal{A}_n = \left\{ \|M\| \ge C_M \frac{\sqrt{nL}}{K^2} \right\} \cap \left\{ \|E\| \le C_E \sqrt{\log n} \right\}.
\]

By the union bound,
\[
\mathbb{P}(\mathcal{A}_n^c) \le \mathbb{P}\left( \|M\| < C_M \frac{\sqrt{nL}}{K^2} \right) + \mathbb{P}\left( \|E\| > C_E \sqrt{\log n} \right).
\]

From Equations (\ref{eqA}) and (\ref{eqB}), both terms on the right-hand side converge to $0$ as $n \to \infty$. Hence, we have
\[
\lim_{n \to \infty} \mathbb{P}(\mathcal{A}_n) = 1.
\]

On the event $\mathcal{A}_n$, we apply the triangle inequality for the spectral norm:
\[
\|R\| = \|M + E\| \ge \|M\| - \|E\| \ge C_M \frac{\sqrt{nL}}{K^2} - C_E \sqrt{\log n}.
\]

Now, by Assumption \ref{assump:a5}, $\sqrt{nL}/K^2 \to \infty$ as $n \to \infty$. Since $\sqrt{\log n} = o(\sqrt{nL}/K^2)$, we have
\[
\frac{C_E \sqrt{\log n}}{C_M \sqrt{nL}/K^2} \to 0 \quad \text{as } n \to \infty.
\]

Thus, there exists an integer $n_0$ such that for all $n \ge n_0$,
\[
C_E \sqrt{\log n} \le \frac{C_M}{2} \cdot \frac{\sqrt{nL}}{K^2}.
\]

Consequently, on $\mathcal{A}_n$ for all $n \ge n_0$,
\[
\|R\| \ge C_M \frac{\sqrt{nL}}{K^2} - \frac{C_M}{2} \cdot \frac{\sqrt{nL}}{K^2} = \frac{C_M}{2} \cdot \frac{\sqrt{nL}}{K^2}.
\]

Therefore, we have established
\begin{align}\label{eq3}
\lim_{n \to \infty} \mathbb{P}\left( \|R\| \ge \frac{C_M}{2} \cdot \frac{\sqrt{nL}}{K^2} \right) = 1.
\end{align}

Next, we relate the spectral norm of the submatrix $R$ to that of the full matrix $\widetilde{A}^{\text{agg}}$. Observe that $R$ is precisely the submatrix of $\widetilde{A}^{\text{agg}}$ obtained by selecting rows indexed by $S$ and columns indexed by $T'$. For any matrix $X \in \mathbb{R}^{n \times n}$ and any index sets $I, J \subseteq [n]$, the spectral norm of the submatrix $X_{IJ}$ does not exceed the spectral norm of $X$. This can be seen from the variational characterization of the spectral norm:
\[
\|X_{IJ}\| = \max_{\substack{u \in \mathbb{R}^{|I|}, v \in \mathbb{R}^{|J|} \\ \|u\|=\|v\|=1}} u^\top X_{IJ} v 
          = \max_{\substack{\tilde{u} \in \mathbb{R}^n, \tilde{v} \in \mathbb{R}^n \\ \|\tilde{u}\|=\|\tilde{v}\|=1 \\ \operatorname{supp}(\tilde{u}) \subseteq I, \operatorname{supp}(\tilde{v}) \subseteq J}} \tilde{u}^\top X \tilde{v}
          \le \max_{\substack{\tilde{u}, \tilde{v} \in \mathbb{R}^n \\ \|\tilde{u}\|=\|\tilde{v}\|=1}} \tilde{u}^\top X \tilde{v} = \|X\|,
\]
where $\operatorname{supp}(x)$ denotes the set of indices where $x$ is nonzero. Applying this inequality with $X = \widetilde{A}^{\text{agg}}$, $I = S$, $J = T'$, we obtain
\[
\|\widetilde{A}^{\text{agg}}\| \ge \|R\|.
\]

Combining this with Equation (\ref{eq3}), we conclude
\begin{align}\label{eq4}
\lim_{n \to \infty} \mathbb{P}\left( \|\widetilde{A}^{\text{agg}}\| \ge \frac{C_M}{2} \cdot \frac{\sqrt{nL}}{K^2} \right) = 1. 
\end{align}

Finally, let $\lambda_{\max}$ denote the eigenvalue of $\widetilde{A}^{\text{agg}}$ with the largest absolute value. Since $\widetilde{A}^{\text{agg}}$ is a real symmetric matrix, its spectral norm equals the spectral radius. Thus, Equation (\ref{eq4}) is equivalent to
\[
\lim_{n \to \infty} \mathbb{P}\left( |\lambda_{\max}| \ge \frac{C_M}{2} \cdot \frac{\sqrt{nL}}{K^2} \right) = 1.
\]

\textbf{Step 7: Controlling the contribution of other eigenvalues to \(T(K_0)\).}

Write \(\widetilde{A}^{\text{agg}} = \widetilde{M} + \widetilde{E}\), where
\[
\widetilde{M}_{ij} = 
\begin{cases} 
  \dfrac{\sum_{\ell=1}^L (\hat{P}_{\ell,ij} - P_{\ell,ij})}{\sqrt{n \hat{D}_{ij}}} & i \neq j, \\
  0 & i = j,
\end{cases}
\qquad
\widetilde{E}_{ij} = 
\begin{cases} 
  \dfrac{\sum_{\ell=1}^L (A_{\ell,ij} - P_{\ell,ij})}{\sqrt{n \hat{D}_{ij}}} & i \neq j, \\
  0 & i = j,
\end{cases}
\]
with \(\hat{D}_{ij} = \sum_{\ell=1}^L \hat{P}_{\ell,ij}(1-\hat{P}_{\ell,ij})\).

\begin{lem}[Spectral norm of the noise matrix]\label{lem:noise_full}
Under Assumptions \ref{assump:a1}--\ref{assump:a3}, \(\|\widetilde{E}\| = O_P(\sqrt{\log n})\).
\end{lem}

\begin{proof}
By Lemma \ref{lem:bound_Dab}, with probability tending to 1,
\[
\hat{D}_{ij} \geq \kappa L \quad \text{for all } i \neq j,
\]
where \(\kappa = \delta(1-\delta)/2\). Also, \(\hat{D}_{ij} \leq L/4\). Conditional on \(\theta\) and \(\hat{\theta}\), the entries \(\{\widetilde{E}_{ij}\}_{i<j}\) are independent, with \(\mathbb{E}[\widetilde{E}_{ij} \mid \theta, \hat{\theta}] = 0\) and
\[
\operatorname{Var}(\widetilde{E}_{ij} \mid \theta, \hat{\theta}) = \frac{\sum_{\ell=1}^L P_{\ell,ij}(1-P_{\ell,ij})}{n \hat{D}_{ij}} \leq \frac{L/4}{n \kappa L} = \frac{1}{4\kappa n}.
\]

Moreover,
\[
|\widetilde{E}_{ij}| \leq \frac{L}{\sqrt{n \kappa L}} = \sqrt{\frac{L}{\kappa n}} \triangleq R_n.
\]

Since \(L = o(n/\log n)\) by Assumption \ref{assump:a3}, we have \(R_n = o(1/\sqrt{\log n})\).

We apply the matrix Bernstein inequality (Theorem 1.6 of \citep{tropp2012user}) to the symmetric matrix \(\widetilde{E} = \sum_{i<j} X_{ij}\), where \(X_{ij}\) has entries \((i,j)\) and \((j,i)\) equal to \(\widetilde{E}_{ij}\) and zeros elsewhere. Then \(\|X_{ij}\| = |\widetilde{E}_{ij}| \leq R_n\), and \(\mathbb{E}[X_{ij} \mid \theta, \hat{\theta}] = 0\). The variance statistic is
\[
\sigma^2 = \left\| \sum_{i<j} \mathbb{E}[X_{ij}^2 \mid \theta, \hat{\theta}] \right\|.
\]
Now,
\[
\sum_{i<j} \mathbb{E}[X_{ij}^2 \mid \theta, \hat{\theta}] = \sum_{i=1}^n \left( \sum_{j \neq i} \mathbb{E}[\widetilde{E}_{ij}^2 \mid \theta, \hat{\theta}] \right) e_i e_i^T
\]
is a diagonal matrix with diagonal entries \(d_i \leq (n-1) \cdot \frac{1}{4\kappa n} \leq \frac{1}{4\kappa}\). Hence, \(\sigma^2 \leq \frac{1}{4\kappa}\).

For any \(t > 0\),
\[
\mathbb{P}\left( \|\widetilde{E}\| \geq t \mid \theta, \hat{\theta} \right) \leq 2n \exp\left( -\frac{t^2/2}{\sigma^2 + R_n t/3} \right).
\]
Choose \(t = 2\sqrt{\frac{2\log n}{\kappa}}\). For sufficiently large \(n\), \(R_n t/3 \leq \sigma^2\). Then
\[
\mathbb{P}\left( \|\widetilde{E}\| \geq 2\sqrt{\frac{2\log n}{\kappa}} \mid \theta, \hat{\theta} \right) \leq 2n \exp\left( -\frac{4\log n}{\kappa} \cdot \frac{1}{2/\kappa} \right) = 2n^{-1}.
\]
Thus, \(\|\widetilde{E}\| = O_P(\sqrt{\log n})\).
\end{proof}

\begin{lem}[Frobenius norm of the signal matrix]\label{lem:M_F_bound}
Under Assumptions \ref{assump:a1}--\ref{assump:a3}, 
\[
\|\widetilde{M}\|_F^2 = O_P\left( \frac{L K_0^2 (m+\sqrt{\log n})^2}{n} \right),
\]
and consequently \(\|\widetilde{M}\|_F = o_P(1)\).
\end{lem}

\begin{proof}
From the proof of Lemma \ref{lem:plugin_error} (see the bound for \(|\widetilde{A}^{\text{agg}}_{ij} - \widetilde{A}^{\text{ideal}}_{ij}|\)), we have
\[
|\widetilde{M}_{ij}| = O_P\left( \frac{\sqrt{L} K_0 (m+\sqrt{\log n})}{n^{3/2}} \right),
\]
which gives
\[
\|\widetilde{M}\|_F^2 = \sum_{i \neq j} |\widetilde{M}_{ij}|^2 = O_P\left( n^2 \cdot \frac{L K_0^2 (m+\sqrt{\log n})^2}{n^3} \right) = O_P\left( \frac{L K_0^2 (m+\sqrt{\log n})^2}{n} \right).
\]
By Assumption \ref{assump:a3}, we get \(\|\widetilde{M}\|_F = o_P(1)\).
\end{proof}

Now let \(\lambda_1 \geq \lambda_2 \geq \dots \geq \lambda_n\) be the eigenvalues of \(\widetilde{A}^{\text{agg}}\). From Equation (\ref{eq4}), with probability tending to 1,
\[
|\lambda_1| = \|\widetilde{A}^{\text{agg}}\| \geq C \frac{\sqrt{nL}}{K^2}
\]
for some constant \(C > 0\). Thus, \(|\lambda_1|^3 = O_P\left( (nL)^{3/2}/K^6 \right)\).

We bound the other eigenvalues using the decomposition \(\widetilde{A}^{\text{agg}} = \widetilde{M} + \widetilde{E}\). Since \(\widetilde{M}\) has rank at most \(K_0\), let its nonzero eigenvalues be \(\mu_1, \dots, \mu_r\) (\(r \leq K_0\)). By Weyl's inequality, for each \(i\),
\[
|\lambda_i - \mu_i| \leq \|\widetilde{E}\| = O_P(\sqrt{\log n}),
\]
where we set \(\mu_i = 0\) for \(i > r\). For \(i > K_0\), \(\mu_i = 0\), so \(|\lambda_i| \leq \|\widetilde{E}\| = O_P(\sqrt{\log n})\). For \(i = 2, \dots, K_0\),
\[
|\lambda_i| \leq |\mu_i| + \|\widetilde{E}\| \leq \|\widetilde{M}\|_F + \|\widetilde{E}\| = o_P(1) + O_P(\sqrt{\log n}) = O_P(\sqrt{\log n}).
\]

Hence, for all \(i \geq 2\), we have \(|\lambda_i| = O_P(\sqrt{\log n})\). Consequently,
\[
\sum_{i=2}^n |\lambda_i|^3 = O_P(n(\log n)^{3/2}).
\]

By Assumption \ref{assump:a5}, we have \(\sum_{i=2}^n |\lambda_i|^3 = o_P(|\lambda_1|^3)\). Therefore,
\[
|T(K_0)| = \frac{1}{\sqrt{6}} \left| \sum_{i=1}^n \lambda_i^3 \right| \geq \frac{1}{\sqrt{6}} \left( |\lambda_1|^3 - \sum_{i=2}^n |\lambda_i|^3 \right) = \frac{1}{\sqrt{6}} |\lambda_1|^3 \left( 1 - o_P(1) \right).
\]

Thus, there exists a constant \(C' > 0\) such that
\[
\lim_{n \to \infty} \mathbb{P}\left( |T(K_0)| \geq C' \cdot \frac{(nL)^{3/2}}{K^6} \right) = 1.
\]

This completes the proof of Theorem \ref{thm:power}.
\end{proof}
\subsection{Proof of Theorem \ref{thm:consistency}}
\begin{proof}

\textbf{Setup.} The sequential estimator is defined as
\[
\hat{K} = \min\bigl\{ K_0 \ge 1 : |T(K_0)| \le t_n \bigr\},
\]
where the threshold sequence $\{t_n\}$ satisfies:
\begin{itemize}
    \item $t_n \to \infty$ as $n \to \infty$,
    \item $t_n = o\!\left( (nL)^{3/2} / K^6 \right)$.
\end{itemize}

We prove consistency by establishing:
\begin{enumerate}
    \item $\lim_{n \to \infty} \mathbb{P}(\hat{K} < K) = 0$ (no underestimation).
    \item $\lim_{n \to \infty} \mathbb{P}(\hat{K} > K) = 0$ (no overestimation).
\end{enumerate}

\textbf{Part 1: No underestimation.}

We need to show that $\mathbb{P}(\hat{K} < K) \to 0$. Observe that the event $\{\hat{K} < K\}$ is equivalent to the union over $K_0 = 1, \dots, K-1$ of the events $\{|T(K_0)| \le t_n\}$. By the union bound,
\[
\mathbb{P}(\hat{K} < K) \le \sum_{K_0=1}^{K-1} \mathbb{P}\bigl( |T(K_0)| \le t_n \bigr).
\]

Thus, it suffices to prove that for each fixed $K_0$ with $1 \le K_0 < K$, 
\[
\lim_{n \to \infty} \mathbb{P}\bigl( |T(K_0)| \le t_n \bigr) = 0.
\]

Fix $K_0$ with $1 \le K_0 < K$. Under the same Assumptions \ref{assump:a1}--\ref{assump:a5}, Theorem \ref{thm:power} establishes that there exists a constant $c_2 > 0$ (depending only on $\eta$, $c_0$, and $\delta$) such that
\[
\lim_{n \to \infty} \mathbb{P}\left( |T(K_0)| \ge c_2 \frac{(nL)^{3/2}}{K^6} \right) = 1.
\]

Since the threshold sequence satisfies $t_n = o\!\left( (nL)^{3/2} / K^6 \right)$, we have, for sufficiently large $n$,
\[
c_2 \frac{(nL)^{3/2}}{K^6} > t_n.
\]
Consequently,
\[
\mathbb{P}\bigl( |T(K_0)| \le t_n \bigr) \le 1 - \mathbb{P}\left( |T(K_0)| \ge c_2 \frac{(nL)^{3/2}}{K^6} \right) \to 0.
\]
This holds for each fixed $K_0$ with $1 \le K_0 < K$. Since $K$ is either fixed or grows slowly with $n$ under Assumption \ref{assump:a3}, taking the union over $K_0 = 1, \dots, K-1$ yields
\begin{align}\label{eqA2}
\lim_{n \to \infty} \mathbb{P}(\hat{K} < K) = 0.
\end{align}

\textbf{Part 2: No overestimation.}

Define the overestimation event:
\[
E_n^{\text{over}} = \{ |T(K)| > t_n \}.
\]

Under the null hypothesis $H_0: K_0 = K$, Theorem \ref{thm:main} gives:
\[
T(K) \stackrel{d}{\longrightarrow} N(0,1).
\]

Since $t_n \to \infty$, for any $\epsilon > 0$, choose $M_\epsilon$ such that $\mathbb{P}(|Z| > M_\epsilon) < \epsilon$ for $Z \sim N(0,1)$. By convergence in distribution, there exists $N_\epsilon$ such that for all $n \ge N_\epsilon$:
\[
\mathbb{P}(|T(K)| > M_\epsilon) < 2\epsilon.
\]

Because $t_n \to \infty$, for sufficiently large $n$, $t_n > M_\epsilon$. Thus:
\[
\mathbb{P}(E_n^{\text{over}}) = \mathbb{P}(|T(K)| > t_n) \le \mathbb{P}(|T(K)| > M_\epsilon) < 2\epsilon.
\]

Since $\epsilon > 0$ is arbitrary, we conclude:
\begin{align}\label{EqA3}
\lim_{n \to \infty} \mathbb{P}(E_n^{\text{over}}) = 0.
\end{align}

\textbf{Part 3: Consistency.}

The event $\{\hat{K} = K\}$ is the complement of $\{\hat{K} < K\} \cup \{\hat{K} > K\}$. Therefore:
\[
\mathbb{P}(\hat{K} = K) = 1 - \mathbb{P}(\hat{K} < K) - \mathbb{P}(\hat{K} > K).
\]

Taking limits as $n \to \infty$ and using Equations (\ref{eqA2}) and (\ref{EqA3}), we have
\[
\lim_{n \to \infty} \mathbb{P}(\hat{K} = K) \ge 1 - 0 - 0 = 1.
\]

Since probabilities are bounded above by $1$, we have:
\[
\lim_{n \to \infty} \mathbb{P}(\hat{K} = K) = 1.
\]

\textbf{Existence of threshold sequence $t_n$.} We verify that $t_n = \log n$ satisfies the required conditions. Clearly $t_n \to \infty$. Moreover:
\[
\frac{t_n}{(nL)^{3/2}/K^6} = \frac{K^6 \log n}{(nL)^{3/2}}.
\]

From Assumption \ref{assump:a5}, $\frac{nL^3}{K^{12}(\log n)^3} \to \infty$, which implies:
\[
\frac{K^{12} (\log n)^3}{n L^3} \to 0.
\]

Taking square roots:
\[
\frac{K^6 (\log n)^{3/2}}{n^{1/2} L^{3/2}} \to 0.
\]

Since $(\log n)^{1/2} \to \infty$, we have:
\[
\frac{K^6 \log n}{(nL)^{3/2}} = \frac{K^6 (\log n)^{3/2}}{n^{1/2} L^{3/2}} \cdot \frac{1}{(\log n)^{1/2}} \to 0.
\]

Thus $t_n = \log n$ satisfies $t_n = o((nL)^{3/2}/K^6)$. This completes the proof.
\end{proof}

\subsection{Proof of Theorem \ref{thm:ratio-behavior}}
\begin{proof}
For $2 \le K_0 \le K-1$, by Theorem \ref{thm:power} and its proof, there exist constants $0 < c < C$ such that
\begin{align*}
&\lim_{n \to \infty} \mathbb{P}\!\left( c \frac{(nL)^{3/2}}{K^6} \le |T(K_0-1)| \le C \frac{(nL)^{3/2}}{K^6} \right) = 1, \\
&\lim_{n \to \infty} \mathbb{P}\!\left( c \frac{(nL)^{3/2}}{K^6} \le |T(K_0)| \le C \frac{(nL)^{3/2}}{K^6} \right) = 1.
\end{align*}

Define the event
\[
\mathcal{E}_n = \left\{ c \frac{(nL)^{3/2}}{K^6} \le |T(K_0-1)| \le C \frac{(nL)^{3/2}}{K^6} \right\}
              \cap \left\{ c \frac{(nL)^{3/2}}{K^6} \le |T(K_0)| \le C \frac{(nL)^{3/2}}{K^6} \right\}.
\]

By union bound and the limits above, $\mathbb{P}(\mathcal{E}_n) \to 1$. On $\mathcal{E}_n$, we have
\[
\eta_{K_0} = \frac{|T(K_0-1)|}{|T(K_0)|} \le \frac{C \frac{(nL)^{3/2}}{K^6}}{c \frac{(nL)^{3/2}}{K^6}} = \frac{C}{c}, 
\]
which gives $\mathbb{P}(\eta_{K_0} \le C/c) \to 1$.

For $K_{0}=K$, consider $\eta_K = |T(K-1)| / |T(K)|$. Since $K-1 < K$, Theorem \ref{thm:power} provides $c > 0$ such that
\begin{align}\label{eqD1}
\lim_{n \to \infty} \mathbb{P}\!\left( |T(K-1)| \ge c \frac{(nL)^{3/2}}{K^6} \right) = 1.
\end{align}

From Theorem \ref{thm:main}, $T(K) \xrightarrow{d} N(0,1)$, implying $|T(K)| = O_P(1)$. More precisely, for any $\epsilon > 0$, there exists $M_\epsilon > 0$ such that
\begin{align}\label{eqD2}
\limsup_{n \to \infty} \mathbb{P}\!\left( |T(K)| > M_\epsilon \right) < \epsilon. 
\end{align}

Now fix any constant $M > 0$. We analyze:
\[
\mathbb{P}(\eta_K \le M) = \mathbb{P}\!\left( |T(K-1)| \le M |T(K)| \right).
\]

Decompose this probability as:
\begin{align*}
\mathbb{P}(\eta_K \le M) &\le \mathbb{P}\!\left( |T(K-1)| < c \frac{(nL)^{3/2}}{K^6} \right)+ \mathbb{P}\!\left( M |T(K)| \ge c \frac{(nL)^{3/2}}{K^6} \right).
\end{align*}

By Equation (\ref{eqD1}), the first term tends to 0. For the second term, note that Assumption \ref{assump:a5} implies
\[
\frac{(nL)^{3/2}}{K^6 \log n} \to \infty.
\]

Thus, for sufficiently large $n$, we have $c (nL)^{3/2}/K^6 > M M_\epsilon$, and
\[
\mathbb{P}\!\left( M |T(K)| \ge c \frac{(nL)^{3/2}}{K^6} \right)
\le \mathbb{P}\!\left( |T(K)| > M_\epsilon \right).
\]

Using Equation (\ref{eqD2}), the right-hand side is eventually less than $\epsilon$. Since $\epsilon$ is arbitrary, we conclude $\mathbb{P}(\eta_K \le M) \to 0$ for any fixed $M$, which is equivalent to $\eta_K \xrightarrow{P} \infty$.
\end{proof}

\subsection{Proof of Theorem \ref{thm:srnast-consistency}}
\begin{proof}
From Theorem \ref{thm:main}, under the null hypothesis $H_0: K_0 = K$, we have $T(K) \xrightarrow{d} N(0,1)$. Consequently,
\begin{equation}\label{eq:TK_bound}
|T(K)| = O_P(1).
\end{equation}

For any candidate $K_0 < K$ (underfitting), Theorem \ref{thm:power} guarantees the existence of a constant $C > 0$ such that
\begin{equation}\label{eq:TK0_divergence}
\lim_{n \to \infty} \mathbb{P}\!\left( |T(K_0)| \ge C \frac{(nL)^{3/2}}{K^6} \right) = 1.
\end{equation}

This shows that $|T(K_0)|$ diverges at rate at least $(nL)^{3/2}/K^6$ when $K_0 < K$. Regarding the ratio statistic $\eta_{K_0} = |T(K_0-1)|/|T(K_0)|$, Theorem \ref{thm:ratio-behavior} provides two key results:
\begin{enumerate}
\item For $2 \le K_0 \le K-1$, $\eta_{K_0} = O_P(1)$.
\item When $K_0 = K$, $\eta_K \xrightarrow{P} \infty$.
\end{enumerate}

In fact, from the proofs of Theorems \ref{thm:power} and \ref{thm:ratio-behavior}, one can see that $\eta_K$ grows at least as fast as $(nL)^{3/2}/K^6$, since the numerator $|T(K-1)|$ is of that order while the denominator $|T(K)|$ remains bounded.

For the SR-NAST algorithm to be consistent, the threshold $t_{\mathrm{ratio},n}$ must satisfy two asymptotic conditions derived from the above properties.

First, consider the case $K = 1$. Equation \eqref{eq:TK_bound} gives $|T(1)| = O_P(1)$. To ensure the algorithm stops at $K_0 = 1$ with high probability, we need $|T(1)| \le t_{\mathrm{ratio},n}$ with probability tending to 1. This requires $t_{\mathrm{ratio},n} \to \infty$, because any bounded threshold would lead to a non-vanishing probability of $|T(1)|$ exceeding it.

Second, consider $K \ge 2$. The algorithm must (a) proceed past $K_0 = 1$ when $K \ge 2$, (b) not stop prematurely at any $K_0 < K$, and (c) stop exactly at $K_0 = K$. These requirements translate into the following:
\begin{itemize}
\item For (a): We need $\mathbb{P}(|T(1)| > t_{\mathrm{ratio},n}) \to 1$. By Equation \eqref{eq:TK0_divergence}, this holds if $t_{\mathrm{ratio},n} = o\!\left( (nL)^{3/2}/K^6 \right)$.
\item For (b): For each $K_0 = 2,\dots,K-1$, we need $\mathbb{P}(\eta_{K_0} > t_{\mathrm{ratio},n}) \to 0$. Since $\eta_{K_0} = O_P(1)$, this holds if $t_{\mathrm{ratio},n} \to \infty$.
\item For (c): We need $\mathbb{P}(\eta_K > t_{\mathrm{ratio},n}) \to 1$. Since $\eta_K$ diverges at least as fast as $(nL)^{3/2}/K^6$, this again requires $t_{\mathrm{ratio},n} = o\!\left( (nL)^{3/2}/K^6 \right)$.
\end{itemize}

Therefore, a threshold sequence $\{t_{\mathrm{ratio},n}\}$ ensures consistency of SR-NAST if it satisfies:
\begin{equation}\label{eq:threshold_conditions}
t_{\mathrm{ratio},n} \to \infty \quad \text{and} \quad t_{\mathrm{ratio},n} = o\!\left( \frac{(nL)^{3/2}}{K^6} \right).
\end{equation}

With $t_{\mathrm{ratio},n}$ satisfying Equation \eqref{eq:threshold_conditions}, we now establish $\mathbb{P}(\hat{K} = K) \to 1$ for all possible $K$.

\noindent\textit{Subcase 1: $K = 1$.} By Equation \eqref{eq:TK_bound}, $|T(1)| = O_P(1)$. Equation \eqref{eq:threshold_conditions} guarantees
\[
\mathbb{P}\!\left( |T(1)| \le t_{\mathrm{ratio},n}\right) \to 1.
\]

Thus, SR-NAST stops at $K_0 = 1$ and returns $\hat{K} = 1$ with probability tending to 1.

\noindent\textit{Subcase 2: $K \ge 2$.} We trace the algorithm:
\begin{itemize}
\item At $K_0 = 1$: Since $K_0 < K$, Equation \eqref{eq:TK0_divergence} applies. Because $t_{\mathrm{ratio},n} = o\!\left( (nL)^{3/2}/K^6 \right)$, so
\[
\mathbb{P}\!\left( |T(1)| > t_{\mathrm{ratio},n} \right) \to 1.
\]
Hence, the algorithm proceeds to $K_0 = 2$ with probability tending to 1.
\item For $K_0 = 2, \dots, K-1$: Theorem \ref{thm:ratio-behavior} gives $\eta_{K_0} = O_P(1)$, which gives
\[
\mathbb{P}\!\left( \eta_{K_0} > t_{\mathrm{ratio},n} \right) \to 0.
\]
Thus, the algorithm does not stop at any $K_0 < K$ with probability tending to 1.
\item At $K_0 = K$: Since $\eta_K$ grows at least as fast as $(nL)^{3/2}/K^6$,  Equation \eqref{eq:threshold_conditions} guarantees
\[
\mathbb{P}\!\left( \eta_K > t_{\mathrm{ratio},n} \right) \to 1.
\]
Therefore, the algorithm stops at $K_0 = K$ with probability tending to 1.
\end{itemize}

Consequently, $\mathbb{P}(\hat{K} = K) \to 1$ for $K \ge 2$ as well. Since both subcases are covered, we conclude that $\lim_{n \to \infty} \mathbb{P}(\hat{K} = K) = 1$ for all $K \ge 1$, completing the proof of Theorem \ref{thm:srnast-consistency}.
\end{proof}
\section{Useful theoretical results}
\subsection{Block probability estimation error in multi-layer SBM}
\begin{lem}\label{BlockErrorMain}
Suppose that Assumption \ref{assump:a3} holds, then under the null hypothesis \(H_0: K = K_0\), for any layer \(\ell \in [L]\) and communities \(k, l \in [K_0]\), the estimated block probability \(\hat{B}_{\ell,kl}\) satisfies:  
\[
|\hat{B}_{\ell,kl} - B_{\ell,kl}| =O_P\left( \sqrt{\frac{\log n}{\hat{n}^2_\mathrm{min}}} \right)+ O_P\left(\frac{m \hat{n}_{\text{max}}}{\hat{n}^2_{\text{min}}}\right),
\]  
where \(\hat{n}_{\text{max}} = \max_k \hat{n}_k\), \(\hat{n}_{\text{min}} = \min_k \hat{n}_k\), and \(m = \|\hat{\theta} - \theta\|_0\) is the number of misclustered nodes.
\end{lem}

\begin{proof}
Let \(\hat{n}_k = |\hat{\mathcal{C}}_k|\) be the size of estimated community \(k\). We redefine the conditional expectation \(\tilde{B}_{\ell,kl}\) to match the estimator's definition
\[
\tilde{B}_{\ell,kl} = 
\begin{cases} 
\frac{1}{\hat{n}_k \hat{n}_l} \sum_{i \in \hat{\mathcal{C}}_k} \sum_{j \in \hat{\mathcal{C}}_l} B_{\ell, \theta_i\theta_j} & k \neq l, \\
\frac{1}{\hat{n}_k (\hat{n}_k - 1)} \sum_{i \neq j \in \hat{\mathcal{C}}_k} B_{\ell, \theta_i\theta_j} & k = l.
\end{cases}
\]

The error is decomposed as
\[
|\hat{B}_{\ell,kl} - B_{\ell,kl}| \leq \underbrace{\left| \hat{B}_{\ell,kl} - \tilde{B}_{\ell,kl} \right|}_{\text{sampling error}} + \underbrace{\left| \tilde{B}_{\ell,kl} - B_{\ell,kl} \right|}_{\text{bias error}}.
\]

\noindent\textbf{Part 1: Bounding the sampling error \(\left| \hat{B}_{\ell,kl} - \tilde{B}_{\ell,kl} \right|\)}

\begin{itemize}
\item \textit{Case \(k \neq l\)}: Given \(\theta\) and \(\hat{\theta}\), the variables \(\{A_{\ell,ij} : i \in \hat{\mathcal{C}}_k, j \in \hat{\mathcal{C}}_l\}\) are independent, bounded in \([0,1]\), and Bernoulli with mean \(B_{\ell, \theta_i\theta_j}\). By Hoeffding's inequality, we have
\[
\mathbb{P}\left( \left| \hat{B}_{\ell,kl} - \tilde{B}_{\ell,kl} \right| \geq t \mid \theta, \hat{\theta} \right) \leq 2 \exp\left(-2 t^2 \hat{n}_k \hat{n}_l \right).
\]
Set \(t = \sqrt{\frac{3 \log n}{\hat{n}_k \hat{n}_l}}\), we get
\[
\mathbb{P}\left( \left| \hat{B}_{\ell,kl} - \tilde{B}_{\ell,kl} \right| \geq \sqrt{\frac{3 \log n}{\hat{n}_k \hat{n}_l}} \mid \theta, \hat{\theta} \right) \leq 2n^{-6}.
\]

\item \textit{Case \(k = l\)}: The estimator excludes diagonal elements. Given \(\theta\) and \(\hat{\theta}\), the variables \(\{A_{\ell,ij} : i \neq j \in \hat{\mathcal{C}}_k\}\) are independent, bounded in \([0,1]\), and Bernoulli with mean \(B_{\ell, \theta_i\theta_j}\). The sum has \(N = \hat{n}_k(\hat{n}_k - 1)\) terms. By Hoeffding's inequality, we get
\[
\mathbb{P}\left( \left| \hat{B}_{\ell,kk} - \tilde{B}_{\ell,kk} \right| \geq t \mid \theta, \hat{\theta} \right) \leq 2 \exp\left( -\frac{t^2 \hat{n}_k (\hat{n}_k - 1)}{8} \right).
\]
Set \(t = \sqrt{\frac{24 \log n}{\hat{n}_k (\hat{n}_k - 1)}}\), we have
\[
\mathbb{P}\left( \left| \hat{B}_{\ell,kk} - \tilde{B}_{\ell,kk} \right| \geq \sqrt{\frac{24 \log n}{\hat{n}_k (\hat{n}_k - 1)}} \mid \theta, \hat{\theta} \right) \leq 2n^{-3}.
\]
To unify both cases, adjust the constant to \(t = \sqrt{\frac{48 \log n}{\hat{n}_k \hat{n}_l}}\) for all \(k,l\) (since \(\hat{n}_k(\hat{n}_k - 1) \geq \hat{n}_k^2 / 2\) for large \(n\)), giving probability bound \(2n^{-6}\).
\end{itemize}

By Assumption \ref{assump:a3},
union bounding over all \(O(K_0^2 L)\) blocks obtains
\[
\mathbb{P}\left( \exists \ell, k, l : \left| \hat{B}_{\ell,kl} - \tilde{B}_{\ell,kl} \right| \geq \sqrt{\frac{48 \log n}{\hat{n}_k \hat{n}_l}} \right) \leq 2L K_0^2 n^{-6} = o(1).
\]

Thus, uniformly, we have
\[
\left| \hat{B}_{\ell,kl} - \tilde{B}_{\ell,kl} \right| = O_P\left( \sqrt{\frac{\log n}{\hat{n}_k \hat{n}_l}} \right)=O_P\left( \sqrt{\frac{\log n}{\hat{n}^2_\mathrm{min}}} \right).
\]

\noindent\textbf{Part 2: Bounding the bias error \(\left| \tilde{B}_{\ell,kl} - B_{\ell,kl} \right|\)}

Define:
\[
\mathcal{G}_k = \{i \in \hat{\mathcal{C}}_k : \theta_i = k\}, \quad 
\mathcal{B}_k = \{i \in \hat{\mathcal{C}}_k : \theta_i \neq k\}, \quad 
|\mathcal{B}_k| \leq m, \quad |\mathcal{G}_k| = \hat{n}_k - |\mathcal{B}_k|.
\]

Similarly define \(\mathcal{G}_l, \mathcal{B}_l\) for community \(l\).

\begin{itemize}
\item \textit{Case \(k \neq l\)}:
\[
\tilde{B}_{\ell,kl} - B_{\ell,kl} = \frac{1}{\hat{n}_k \hat{n}_l} \sum_{i \in \hat{\mathcal{C}}_k} \sum_{j \in \hat{\mathcal{C}}_l} \left( B_{\ell, \theta_i\theta_j} - B_{\ell,kl} \right).
\]
The term \(\left| B_{\ell, \theta_i\theta_j} - B_{\ell,kl} \right| \leq 1\) and is zero if \(i \in \mathcal{G}_k, j \in \mathcal{G}_l\). The number of non-zero terms is at most
\[
|\mathcal{B}_k| \hat{n}_l + |\mathcal{B}_l| \hat{n}_k - |\mathcal{B}_k||\mathcal{B}_l| \leq 2m \hat{n}_{\text{max}}.
\]
Thus, we have
\[
\left| \tilde{B}_{\ell,kl} - B_{\ell,kl} \right| \leq \frac{2m \hat{n}_{\text{max}}}{\hat{n}_k \hat{n}_l}\leq \frac{2m \hat{n}_{\text{max}}}{\hat{n}^2_{\text{min}}}.
\]

\item \textit{Case \(k = l\)}:
\[
\tilde{B}_{\ell,kk} - B_{\ell,kk} = \frac{1}{\hat{n}_k (\hat{n}_k - 1)} \sum_{i \neq j \in \hat{\mathcal{C}}_k} \left( B_{\ell, \theta_i\theta_j} - B_{\ell,kk} \right).
\]
The term is zero if \(i,j \in \mathcal{G}_k\). The number of non-zero terms is
\[
\hat{n}_k(\hat{n}_k - 1) - |\mathcal{G}_k|(|\mathcal{G}_k| - 1) \leq 2\hat{n}_k m.
\]
Thus, we get
\[
\left| \tilde{B}_{\ell,kk} - B_{\ell,kk} \right| \leq \frac{2\hat{n}_k m}{\hat{n}_k (\hat{n}_k - 1)} = \frac{2m}{\hat{n}_k - 1}.
\]
For large \(n\), \(\hat{n}_k - 1 \geq \hat{n}_k / 2\), so, we have
\[
\frac{2m}{\hat{n}_k - 1} \leq \frac{4m}{\hat{n}_k}\leq \frac{4m}{\hat{n}_{\text{min}}}\leq \frac{4m \hat{n}_{\text{max}}}{\hat{n}^2_{\text{min}}}..
\]
\end{itemize}

Thus, for both cases, we have
\[
\left| \tilde{B}_{\ell,kl} - B_{\ell,kl} \right| \leq \frac{4m \hat{n}_{\text{max}}}{\hat{n}^2_{\text{min}}}.
\]

\noindent\textbf{Combining errors:}
\[
|\hat{B}_{\ell,kl} - B_{\ell,kl}| = O_P\left( \sqrt{\frac{\log n}{\hat{n}^2_\mathrm{min}}} \right)+ O_P\left(\frac{m \hat{n}_{\text{max}}}{\hat{n}^2_{\text{min}}}\right).
\]
\end{proof}
\subsection{Uniform lower bound for estimated variance terms}
\begin{lem}\label{UniformLowerBoundForEstimatedVariance}
Suppose that Assumptions \ref{assump:a1}-\ref{assump:a3} hold, then under the null hypothesis \(H_0: K = K_0\), there exists a constant \(\kappa = \kappa(\delta) > 0\) such that:
\[
\min_{i \neq j} \hat{D}_{ij} \geq \kappa L \quad \text{with probability tending to 1 as } n \to \infty,
\]
where \(\hat{D}_{ij} = \sum_{\ell=1}^L \hat{P}_{\ell,ij}(1 - \hat{P}_{\ell,ij})\) and \(\kappa = \delta(1-\delta)/2\).
\end{lem}
\begin{proof}
By Assumption \ref{assump:a1}, \(P_{\ell,ij} = B_{\ell,\theta_i,\theta_j} \in [\delta, 1-\delta]\) for all \(\ell, i, j\). The function \(g(x) = x(1-x)\) satisfies
   \[
   \min_{x \in [\delta,1-\delta]} g(x) = \delta(1-\delta) > 0.
   \]  
   
   Thus, for the true variance term, we have 
   \[
   \min_{i \neq j} D_{ij} = \min_{i \neq j} \sum_{\ell=1}^L P_{\ell,ij}(1-P_{\ell,ij}) \geq L \cdot \delta(1-\delta).
   \]

By Lemma \ref{BlockErrorMain} and Assumptions \ref{assump:a2}-\ref{assump:a3}, we have  
   \[
   \max_{\ell,i,j} |\hat{P}_{\ell,ij} - P_{\ell,ij}| = O_P\left( \sqrt{\frac{\log n}{\hat{n}^2_\mathrm{min}}} \right) + O_P\left(\frac{m \hat{n}_{\text{max}}}{\hat{n}^2_{\text{min}}}\right)=O_P\left( \frac{K_0(m + \sqrt{\log n})}{n} \right) = o_P(1),
   \]  
   where \(m = \|\hat{\theta} - \theta\|_0\). This implies 
   \[
   \mathbb{P}\left( \max_{\ell,i,j} |\hat{P}_{\ell,ij} - P_{\ell,ij}| \geq \eta \right) \to 0 \quad \forall \eta > 0.
   \]

Define \(\mathcal{I} = [\delta/2, 1-\delta/2] \supset [\delta, 1-\delta]\). The function \(g(x) = x(1-x)\) is uniformly continuous on \(\mathcal{I}\) with minimum value 
   \[
   g_{\min} = \min_{x \in \mathcal{I}} g(x) = \frac{\delta}{2}\left(1 - \frac{\delta}{2}\right) \geq \frac{\delta(1-\delta)}{2} = \kappa,
   \]  
   where the inequality holds because \(\frac{\delta}{2}(1 - \frac{\delta}{2}) - \frac{\delta(1-\delta)}{2} = \delta^2/4 > 0\). By uniform continuity, \(\exists \eta = \eta(\delta) > 0\) such that if \(|x - y| < \eta\) and \(y \in [\delta,1-\delta]\), then \(x \in \mathcal{I}\) and  
   \[
   |g(x) - g(y)| < \frac{\delta(1-\delta)}{4}.
   \]  
   
   Consequently, when \(|\hat{P}_{\ell,ij} - P_{\ell,ij}| < \eta\), we have
   \[
   \hat{P}_{\ell,ij}(1 - \hat{P}_{\ell,ij}) > P_{\ell,ij}(1 - P_{\ell,ij}) - \frac{\delta(1-\delta)}{4} \geq \frac{3\delta(1-\delta)}{4} > \kappa.
   \]

By Lemma \ref{BlockErrorMain} and Hoeffding's inequality, for each \((\ell,i,j)\), weh ave 
   \[
   \mathbb{P}\left( |\hat{P}_{\ell,ij} - P_{\ell,ij}| \geq \eta \right) \leq O(n^{-c}) \quad (c> 0).
   \]  
   
   Applying a union bound over all \(O(n^2 L)\) entries gets 
   \[
   \mathbb{P}\left( \max_{\ell,i,j} |\hat{P}_{\ell,ij} - P_{\ell,ij}| \geq \eta \right) \leq O(n^2 L \cdot n^{-c}) = O(L n^{2-c)}.
   \]  
   
   Under Assumption \ref{assump:a3}, choosing \(c > 3\) gets 
   \[
   L n^{2-c_1} \leq o(n^{3-c_1}) \to 0.
   \]  
   
   Thus w.h.p., \(|\hat{P}_{\ell,ij} - P_{\ell,ij}| < \eta\) for all \(\ell,i,j\), which implies  
   \[
   \hat{D}_{ij} = \sum_{\ell=1}^L \hat{P}_{\ell,ij}(1 - \hat{P}_{\ell,ij}) > \sum_{\ell=1}^L \kappa = \kappa L \quad \forall i \neq j.
   \]  
   This completes the proof.
\end{proof}
\end{appendices}


\bibliography{reference}% common bib file

\end{document}